\documentclass[11pt, letterpaper]{amsart}

\usepackage{color}
\usepackage[english]{babel}
\usepackage[latin1]{inputenc}
\usepackage{times}
\usepackage[T1]{fontenc}
\usepackage{amsfonts}
\usepackage{epsfig,psfrag,latexsym}
\usepackage{amsmath, amssymb}
\usepackage{graphicx}
\usepackage{amsthm}

\newtheorem{theorem}{Theorem}[section]
\newtheorem{corollary}[theorem]{Corollary}
\newtheorem{lemma}[theorem]{Lemma}
\newtheorem{proposition}[theorem]{Proposition}
\theoremstyle{definition}
\newtheorem{definition}[theorem]{Definition}

\newtheorem{assumption}[theorem]{Assumption}

\newtheorem{ass}[theorem]{Assumption}

\theoremstyle{remark}
\newtheorem{remark}[theorem]{Remark}

\numberwithin{equation}{section}

\newcommand{\reals}{\mathbb R}

\newcommand{\eps}{\varepsilon}

\newcommand{\such}{\ | \ }

\newcommand{\prob}{\mathbb{P}}
\newcommand{\qprob}{\mathbb{Q}}

\newcommand{\esp}{\mathbb{E}}
\newcommand{\espalt}[2]{\esp^{#1}\bra{#2}}
\newcommand{\espaltm}[3]{\esp^{#1}_{#2}\bra{#3}}

\newcommand{\F}{\mathcal{F}}

\newcommand{\M}{\mathcal{M}}
\newcommand{\tM}{\tilde{\mathcal{M}}}

\newcommand{\filt}{\mathbb{F}}
\newcommand{\filtration}{\filt = \pare{\F_t}_{0\leq t\leq T}}

\newcommand{\probtriple}{\pare{\Omega, \F, \prob}}

\newcommand{\esssup}[2]{\mathop{\textrm{ess sup}}_{#1}\left[ #2\right]}
\newcommand{\essinf}[2]{\mathop{\textrm{ess inf}}_{#1}\left[ #2\right]}

\newcommand{\relent}[2]{H\left(#1\such #2\right)}

\newcommand{\N}{\mathbb{N}}
\newcommand{\R}{\mathbb{R}}
\newcommand{\EN}{\mathcal{E}}

\newcommand{\Lb}{\mathbb{L}}

\newcommand{\nada}[1]{}

\newcommand{\dfn}{\, := \,}
\newcommand{\rdfn}{\, =: \,}

\newcommand{\Ua}{\mathcal{U}_{a}}

\newcommand{\pare}[1]{\left(#1\right)}
\newcommand{\bra}[1]{\left[#1\right]}
\newcommand{\cbra}[1]{\left\{#1\right\}}
\newcommand{\dbra}[1]{[\kern-0.15em[ #1 ]\kern-0.15em]}
\newcommand{\dbraco}[1]{[\kern-0.15em[ #1 [\kern-0.15em[}

\newcommand{\transpm}{\lambda}

\newcommand{\epn}{p_*^n}
\newcommand{\eqn}{q_*^n}
\newcommand{\pr}{P}

\title{The pricing of contingent claims and optimal positions in asymptotically complete markets}

\author{Michail Anthropelos}
\address{Department of Banking and Financial Management\\
University of Piraeus\\
Piraeus, Greece}
\email{anthropel@unipi.gr}
\thanks{M. Anthropelos is supported in part by the Research Center of the University of Piraeus}

\author{Scott Robertson}
\address{Questrom School of Business\\
Boston University\\
Boston, MA 02215}
\email{scottrob@bu.edu}
\thanks{S. Robertson is supported in part by the National Science Foundation (NSF)
  under grants DMS-1651180 and DMS-1613159.}

\author{Konstantinos Spiliopoulos}
\address{Department of Mathematics \& Statistics\\
Boston University\\
Boston, MA 02215}
\email{kspiliop@math.bu.edu}
\thanks{K. Spiliopoulos is supported in part by the NSF under grant number DMS-1312124 and during revisions of the article by  the NSF
  CAREER award DMS-1550918. }

\date{\today}

\begin{document}

\begin{abstract}
We study utility indifference prices and optimal purchasing
quantities for a contingent claim, in an incomplete
semi-martingale market, in the presence of vanishing hedging
errors and/or risk aversion. Assuming that the average
indifference price converges to a well defined limit, we prove
that optimally taken positions become large in absolute value at a
specific rate. We draw motivation from and make connections to
Large Deviations theory, and in particular, the celebrated
G\"{a}rtner-Ellis theorem. We analyze a series of well studied
examples where this limiting behavior occurs, such as fixed
markets with vanishing risk aversion, the basis risk model with
high correlation, models of large markets with vanishing trading
restrictions and the Black-Scholes-Merton model with either
vanishing default probabilities or vanishing transaction costs.
Lastly, we show that the large claim regime could naturally arise
in partial equilibrium models.
\end{abstract}

\keywords{Indifference Pricing, Incomplete Markets, Utility Functions, Large
  Position Size}
\maketitle



\section{Introduction}\label{S:intro}

The goal of this paper is to study the relationship between utility indifference prices and optimal positions for a contingent claim, in a general incomplete semi-martingale market, under the assumption of vanishing hedging errors. In particular, for an exponential utility investor, we wish to verify the heuristic adage that when purchasing optimal quantities one obtains the delicate relationship
\begin{equation*}
\textrm{position size }\times \textrm{ risk aversion }\times \textrm{incompleteness parameter} \approx \textrm{ constant}.
\end{equation*}

Here, the incompleteness parameter represents the hedging error associated with the claim.  From the above we see that as the market becomes complete (or, at least as the given claim in question becomes asymptotically hedgeable), optimal position sizes tend to become large.  In fact, optimal position sizes may also become large as risk aversion vanishes in a fixed market, and our analysis is robust enough to cover both cases.

The financial motivation for studying this situation is that large
positions are indeed being taken.  For example, the over the
counter derivatives markets now has more than $700$ trillion
notional outstanding (see \cite{BIS_Data}). Other examples include
mortgage backed securities, life insurance contracts and mortality
derivatives.  These products are not completely replicable and a
position on them implies unhedgeable risk. Therefore, it is
natural to study the situation within the framework of utility
based analysis in incomplete markets. Moreover, the observation
that position size is connected to hedging error can be understood
as follows. In a complete market there is only one fair price $d$
for a given claim. Hence, if one is able to purchase claims for
price $p\neq d$ then it is optimal to take an infinite position.
Of course, in reality  one cannot take an infinite position and
complete markets are an ideal situation. However, these
considerations indicate that large positions may arise
endogenously, if the hedging error or risk aversion is small. We
also mention that this is the underlying motivation for the
indifference price approximations in the basis risk models of
\cite{davis1997opi,MR1926237}, which we revisit in the current
paper.

Starting at least from \cite{HN1989}, utility indifference pricing
has attracted a lot of attention, see for example \cite{MR2547456}
for detailed overview. Recently, indifference pricing for large
position sizes has been studied in \cite{BEM_2012, Robertson_2012,
Robertson_Spil_2014}. In \cite{Robertson_Spil_2014} the authors consider a sequence of a particular semi-complete
market indexed by $n$ that becomes
complete as $n\rightarrow\infty$ and, assuming the unhedgeable
component of the non-traded asset vanishes in accordance to a
Large Deviation Principle (LDP), it is shown that optimal purchase
quantities become large at precisely the Large Deviations scaling.

To help motivate our results, let us briefly outline the main
idea. Let $n\in\mathbb{N}$ and consider a semi-martingale market
with available risky assets for investment $S^n$, and an investor
who owns a non-traded contingent claim $B$.  The investor has
exponential utility with risk aversion $a_n>0$, where, in addition
to the assets, we allow the risk aversion to change with $n$ so
that $U_{a_n}(x) = -(1/a_n)e^{-a_nx}, x\in \reals$. Let $\mathcal{A}^n$ be the set of
admissible trading strategies and
$X^{\pi^n}= (\pi^n\cdot S^n)$ be the resultant wealth process, for some $\pi^n\in\mathcal{A}^n$. The
optimal utility that the investor can achieve by trading in $S^n$ with
initial capital $x$ and $q$ units of $B$ is
\begin{equation*}
u^n_{a_n}(x,q) = \sup_{\pi^n\in\mathcal{A}^n}\espalt{}{U_{a_n}(x +X^{\pi^n}_T + qB)};\quad u^n_{a_n}(x) = u^n_{a_n}(x,0).
\end{equation*}
Then, the average bid utility indifference price $p^n_{a_n}(x,q)$ is defined through the balance equation
\begin{equation*}
u^n_{a_n}(x-qp^n_{a_n}(x,q), q) = u^n_{a_n}(x).
\end{equation*}
It is well known that $p^n_{a_n}$ does not depend upon $x$, and writing $p^n_{a_n}(q)$, takes the form
\begin{equation*}
p^n_{a_n}(q) = -\frac{1}{a_nq} \log\left(\espalt{\qprob^n_0}{e^{-a_nq\hat{Y}^n_{a_n}(q)}}\right),
\end{equation*}
where $\qprob^n_0$ is the minimal entropy measure in the $n^{th}$
market and $\hat{Y}^n_{a_n}(q)$ is related to the normalized
residual risk (see \cite{MR2667897, MR2048829} amongst others) of
owning $q$ units of $B$. Thus, $p^n_{a_n}$ can be viewed as a
``generalized'' version of the scaled cummulant generating
function $\Lambda_n(q)/q$, where
$\Lambda_n(q)\dfn\log\left(\espalt{}{e^{qY_n}}\right)$ for a
sequence of random variables $\cbra{Y_n}$ from Large Deviations
theory (see \cite{MR1619036} for a classical manuscript). Taking a
cue from the celebrated G\"{a}rtner-Ellis theorem, which deduces
an LDP for the tail probabilities of $\cbra{Y_n}$ from the
assumption that $\lambda\mapsto (1/r_n) \Lambda_n(\lambda r_n)$
converges to a sufficiently regular function as $r_n\rightarrow
\infty$, we naturally ask what conclusions can be deduced from the
assumption that $\ell\mapsto p^n_{a_n}(\ell r_n)$ converges to a
well defined limit for $\ell\in\reals$ and $r_n\rightarrow\infty$.
Specifically, we assume (see Assumption \ref{A:GE_Opt1}) that
there exist a sequence $\cbra{r_n}$ of positive numbers with
$r_n\rightarrow \infty$ and a $\delta>0$ such that for all
$|\ell|<\delta$ the limit
\begin{equation}\label{E:limit_px_intro}
p^\infty(\ell)= \lim_{n\uparrow\infty} p^n_{a_n}(\ell r_n),
\end{equation}
exists, is finite, and is continuous  at $\ell=0$. The price $p^\infty(0)$ is thus the limiting price ignoring position size, and when the market is asymptotically complete, represents the unique arbitrage free price in the limiting complete market: see Section \ref{SS:opt_strat}.

As a first consequence, we prove (see Theorems \ref{T:opt_pos_lb},
\ref{T:opt_pos_ub}) that large optimal positions arise
endogenously at a rate proportional to $r_n$. Specifically, for
any price $\tilde{p}^n$ which is arbitrage free in the
pre-limiting markets, the optimal position size (as defined in
\cite{MR2212897}) $\hat{q}^{n}=\hat{q}^{n}(\tilde{p}^n)$ is such
that for $n$ large enough
\begin{equation*}
|\hat{q}^{n}|\approx \ell r_{n}, \text{ for some }\ell\in(0,\infty),
\end{equation*}
provided that $\tilde{p}^n\rightarrow\tilde{p}\neq p^\infty(0)$. Namely, we have $|\hat{q}^{n}|\rightarrow\infty$ at the speed of $r_{n}$.

Secondly, in Section \ref{S:PEPQ} we show under which conditions
the large claim regime could arise in an equilibrium setting, with
a particular focus on justifying the assumption that,
asymptotically, one could buy the claim for a price $\tilde{p}\neq
p^\infty(0)$.  Provided that stock market prices are exogenously
given, the equilibrium price of a claim is the one at which the
optimal quantities of the investors sum up to zero, meaning that
the market of the claim is cleared out. If such a (partial)
equilibrium price exists for each $n\in\N$, it is natural to ask
where this sequence converges to, and if the prices induce
investors to enter the large claim regime. Here, we show that if
the investors' random endowments are dominated by $r_n$, then
equilibrium prices converge to $p^{\infty}(0)$; the unique
limiting arbitrage free price. However, if investors' endowments
are growing with rate $r_n$, equilibrium prices may converge to a
limit $\tilde{p}\neq p^\infty(0)$ and hence the large claim regime
of Theorems \ref{T:opt_pos_lb}, \ref{T:opt_pos_ub} occurs. This
happens when one investor already owns large position in $B$, and
yields a family of examples where the large claim regime is in
fact the market's equilibrium. This result helps to explain the
large observed volumes in OTC derivative markets and the
corresponding extreme prices that often appear (see for instance
\cite{AtkEisWeil13, BIS_Data}).

Thirdly, we illustrate through numerous and varied examples that the price convergence in \eqref{E:limit_px_intro} holds, and hence is a natural feature of either asymptotically complete markets or vanishing investor's risk aversion in a fixed market. Moreover, in all of these examples we explicitly identify the speed $r_{n}$ at which optimal positions grow. To be precise, we validate these claims in the following cases: (a) vanishing risk aversion in a fixed market in Section \ref{SS:vanish_ra}, (b) basis risk model with high correlation in Section \ref{SS:br}, (c) large markets with vanishing trading restrictions in Section \ref{SS:SC}, (d) Black-Scholes-Merton model with vanishing default probability in Section \ref{SS:BS_defautl}, and (e) vanishing transaction costs in the Black-Scholes-Merton model in Section \ref{SS:BS_trans}. 

The vanishing transaction costs example of Section
\ref{SS:BS_trans} probably deserves more discussion. The first
interesting point is that our theory unifies frictionless markets
and markets with frictions, such as transaction costs. In
particular, not only do the statements on optimal positions in
frictionless markets carry over, but in both cases, the main
results turn out to be natural outcomes of the same general
statements presented in Appendix \ref{SS:A_technical}. The second
interesting point is that our analysis reveals that the natural
relation between risk aversion, $a_n$, optimal position size,
$\hat{q}_n$, and proportion of the transaction costs, $\lambda_n$
is $a_n \hat{q}_{n}\lambda_{n}^{2}\approx \text{constant}$. Apart
from the conclusion that for fixed risk aversion, this relation
indicates that $r_{n}=\lambda^{-2}_{n}$, i.e.~that
$\hat{q}_{n}\lambda^{2}_{n}\rightarrow\ell\in(0,\infty)$, it also
justifies the appropriateness of the limiting asymptotic regimes,
which were considered previously without justification; for
example, as in \cite{MR1809526, MR3251862}.

 Even though our focus in this paper is on investors with exponential utility, our results are also true within the class of utility functions that decay exponentially for large negative wealths, see Section \ref{SS:opt_pos_gen_util}. In this case, the optimal position is not necessarily unique. However, we prove that optimizers do exist and that under the assumption of convergence of indifference prices with speed $r_{n}$, \emph{for exponential utility}, each optimizer will converge to $\pm \infty$ with speed $r_{n}$.

We conclude the introduction with a discussion on the applicability and usefulness of the results of this paper. First of all, our analysis offers a bridge between complete and incomplete markets.  Complete markets, where computations are often tractable and explicit, are clearly an idealization of reality. However, their more realistic incomplete counterparts are typically intractable when it comes to identifying optimal trading strategies and pricing contingent claims. To connect these two settings, it is thus natural to consider small perturbations away from complete markets. In the case of fixed investor preferences, this paper addresses precisely this situation, and we show that as the perturbation vanishes, large investors may endogenously arise through optimal trading. Secondly, our work also acts as a bridge between risk averse and risk neutral investors.  For example, it is often assumed that market makers are risk neutral, which is of course only approximately true. Our analysis shows, however, that as market makers approach risk neutrality, they will be induced into both taking large positions and offering prices so that other buyers enter into the market in a large way. Thirdly, the equilibrium results of Section \ref{S:PEPQ} show that it takes only one person to be in the large claim regime in order for others to enter that regime by acting optimally. Hence, our results can be also used to both study and justify the emergence of large players in derivative markets, in the setting where players take large positions immediately, as opposed to incrementally increasing their position sizes. Fourthly, our work can help towards correctly pricing claims in the presence of small unheadgable risks (e.g. in the insurance industry), when positions are of significant size.

The rest of the paper is organized as follows. In Section
\ref{S:model} we describe in detail the model and the optimal
investment problem. In Section \ref{S:GE} we lay down our main
assumption on convergence of scaled indifference prices and draw
motivations with and connections to Large Deviations theory. In
Section \ref{S:consequences} we describe the main consequences of
the assumption of convergence of scaled indifference prices.
Namely, we state the theorems on optimal positions and discuss
their consequences. We additionally discuss the limiting behavior
for the optimal wealth process, and justify the interpretation
that the speed $r_{n}$ characterizes the speed at which the market
approaches completion.  Moreover, we prove that the general
results on optimal positions are true for all utility functions in
the class of utility functions that decay exponentially  for large
negative wealths. Section \ref{S:PEPQ} contains the results on the
partial equilibrium model and on its limiting behavior. Section
\ref{S:examples} contains the motivating examples of frictionless
markets that satisfy our assumptions. Section \ref{SS:BS_trans}
contains the example with vanishing transaction costs. Appendices \ref{SS:A_technical}, \ref{S:pf_opt_pos}, and \ref{S:pf_trans} contain most of the proofs.

\section{The Model, Optimal Investment Problem and Indifference Price}\label{S:model}

We fix a horizon $T>0$, probability space $\probtriple$ and
filtration $\filtration$, which is assumed to satisfy the usual
conditions. Additionally, we assume $\F=\F_T$ and zero interest
rates so the risk-free asset is identically equal to $1$.  For
$n\in\mathbb{N}$ we denote by $S^n$ an $\reals^{d_n}$-valued,
locally bounded semi-martingale which represents the risky assets
available for investment.  In the sequel, we consider the
valuation and the optimal position taking in a contingent claim
$B\in\mathbb{L}^0\probtriple$ assumed to satisfy:
\begin{assumption}\label{A:claim}
$\espalt{}{e^{\lambda B}} < \infty$ for all $\lambda\in\reals$.
\end{assumption}

Since the assets are changing with $n$, the class of equivalent
local martingale measures are changing with $n$ as well.  We
denote by $\M^n$ the family of measures $\qprob^n\sim \prob$ on
$\F$ such that $S^{n}$ is a $\qprob^n$ local martingale.  Recall
for two probability measures $\mu\ll\nu$ the relative entropy of
$\mu$ with respect to $\nu$ is given by $\relent{\mu}{\nu} =
\espalt{\nu}{(d\mu/d\nu)\log(d\mu/d\nu)}$. In order to rule out
arbitrage in each market, we make the following standard
assumption as seen in \cite{MR1891730, MR1743972} amongst many
others:
\begin{assumption}\label{A:no_arb_n}
For each $n$, $\tM^n\dfn\cbra{\qprob^n\in\M^n\ :\ \relent{\qprob^n}{\prob} < \infty} \neq \emptyset$.
\end{assumption}

We consider an exponential utility investor with risk aversion $a_n>0$, where, in addition to the assets, we allow the risk aversion to change with $n$. Thus, the investor has utility function
\begin{equation}\label{E:util_funct}
U_{a_n}(x) = -\frac{1}{a_n}e^{-a_nx};\qquad x\in \reals.
\end{equation}

A trading strategy $\pi^n$ is admissible if it is predictable,
$S^n$ integrable, and if the stochastic integral $X^{\pi^n}\dfn
(\pi^n\cdot S^n)$ is a $\qprob^n$ supermartingale for all
$\qprob^n\in\tM^n$. The set of admissible trading strategies for
the $n^{th}$ market is denoted $\mathcal{A}^n$. For an initial
capital $x$ and position $q\in\reals$ in the claim $B$ we define
\begin{equation}\label{E:val_funct_claim}
u^n_{a_n}(x,q) \dfn
\sup_{\pi^n\in\mathcal{A}^n}\espalt{}{U_{a_n}(x + X^{\pi^n}_T +
qB)},
\end{equation}
as the optimal utility an investor can achieve by trading in $S^n$
with initial capital $x$ and $q$ units of $B$. When $q=0$ so that
the investor does not own the claim we denote the value function
by
\begin{equation}\label{E:val_funct_no_claim}
u^n_{a_n}(x) \dfn \sup_{\pi^n\in\mathcal{A}^n}\espalt{}{U_{a_n}(x+X^{\pi^n}_T)}.
\end{equation}

The average (bid) utility indifference price $p^n_{a_n}(x,q)$ for
initial capital $x$ and $q$ units of $B$ is defined through the
balance equation
\begin{equation}\label{E:util_indiff_px}
u^n_{a_n}(x-qp^n_{a_n}(x,q), q) = u^n_{a_n}(x).
\end{equation}
We now summarize a number of well known results regarding the
utility maximization problem for exponential utility under the
current setup and assumptions. For proofs of these facts, see
\cite{MR1891730, MR1743972, MR1920099, MR1891731, MR2152255,
MR2489605}.

Since $u^n_{a_n}(x,q) = e^{-a_n x}u^n_{a_n}(0,q)$ we consider without loss of generality
that $x=0$ throughout.  The value function without $B$,
$u^n_{a_n}(0)$, is attained by an admissible strategy
$\hat{\pi}^n_{a_n}(0)$. Write $\hat{X}^n_{a_n}(0) \dfn
X^{\hat{\pi}^n_{a_n}(0)}$ as the optimal wealth process.
Additionally, denote by $\qprob^n_0\in\tM^n$ the minimal entropy
measure, which exists. Then $\qprob^n_0$ and $\hat{X}^n_{a_n}(0)$
are related by the formula
\begin{equation}\label{E:min_ent_density}
\frac{d\qprob^n_0}{d\prob}\bigg|_{\F_T}  = \frac{e^{-a_n\hat{X}^n_{a_n}(0)_T}}{\espalt{}{e^{-a_n\hat{X}^n_{a_n}(0)_T}}}.
\end{equation}

In a similar fashion, the value function for $q$ units of $B$,
$u^n_{a_n}(0,q)$, is also attained for some admissible trading
strategy $\hat{\pi}^n_{a_n}(q)$ and we write
$\hat{X}^n_{a_n}(q)\dfn X^{\hat{\pi}^n_{a_n}(q)}$ as the resultant
wealth process. The indifference price does not depend upon the
initial capital and we write $p^n_{a_n}(q)$ instead of
$p^n_{a_n}(x,q)$. By its definition, $p^n_{a_n}(q)$ is given by
the abstract formula
\begin{equation}\label{E:indiff_px}
\begin{split}
p^n_{a_n}(q) &= -\frac{1}{a_nq}\log\left(\frac{u^n_{a_n}(0,q)}{u^n_{a_n}(0)}\right),
\end{split}
\end{equation}
and the total price $qp^n_{a_n}(q)$ admits the variational representation
\begin{equation}\label{E:total_price}
qp^n_{a_n}(q) = \inf_{\qprob^n\in\tM^n}\left(q\espalt{\qprob^n}{B}
+ \frac{1}{a_n}\left(\relent{\qprob^n}{\prob} -
\relent{\qprob^n_0}{\prob}\right)\right).
\end{equation}
Note that from \eqref{E:total_price} one can easily deduce that
for $q\in\reals$
\begin{equation}\label{E:ra_switch}
p^n_{a_n}(q) = p^n_1(a_nq).
\end{equation}
Also, using \eqref{E:min_ent_density} and \eqref{E:indiff_px} we
obtain
\begin{equation}\label{E:px_as_cgf}
\begin{split}
p^n_{a_n}(q) &=
-\frac{1}{a_nq}\log\left(\frac{\espalt{}{e^{-a_n\hat{X}^n_{a_n}(q)_T-a_nqB}}}{\espalt{}{e^{-a_n\hat{X}^n_{a_n}(0)_T}}}\right)=-\frac{1}{a_nq}\log\left(\espalt{\qprob^n_0}{e^{-a_nq\hat{Y}^n_{a_n}(q)}}\right),
\end{split}
\end{equation}
where
\begin{equation}\label{E:norm_resid_risk_p}
\hat{Y}^n_{a_n}(q) \dfn \frac{1}{q}\left(\hat{X}^n_{a_n}(q)_T -
\hat{X}^n_{a_n}(0)_T + qB\right).
\end{equation}
$\hat{Y}^n_{a_n}(q)$ is intimately related to the \emph{normalized
residual risk process} of \cite{MR2667897, MR2048829,
stoikov2005optimal} amongst others and can be seen as the per unit
unhedgeable part of the long position on $q$ units of the claim
$B$.

\section{Limiting Prices and Connections to Large Deviations Theory}\label{S:GE}

Equation \eqref{E:px_as_cgf} is the starting point for our analysis. To motivate the result we first make connections
with the Large Deviation Principle (LDP) and G\"{a}rtner-Ellis theorem from Large Deviations, both stated here for the convenience of the
reader, see for example \cite{MR1619036}.
\begin{definition}\label{Def:LDP}
Let $S$ be a Polish space with Borel sigma-algebra $\mathcal{B}(S)$ and $\probtriple$ be a probability space.  We say that a collection of random variables $\cbra{Y_n}_{n\in\N}$ from $\Omega$ to $S$ has a LDP with good rate function $I:S\to [0,\infty]$ and scaling $r_n$ if $r_n\rightarrow \infty$ and
\begin{enumerate}
\item For each $s\ge 0$, the set $\Phi(s) = \left\{ s\in S: I(s)\le s\right\}$
is a compact subset of $S$; in particular, $I$ is lower semi-continuous.
\item For every open $G\subset S$, $\varliminf_{n\uparrow\infty}(1/r_n)\log\left(\prob\bra{ Y_n\in G}\right) \geq -\inf_{s\in G} I(s)$.
\item For every closed $F\subset S$, $\varlimsup_{n\uparrow\infty}(1/r_n)\log\left(\prob\bra{ Y_n\in F}\right)\leq -\inf_{s\in F} I(s)$.
\end{enumerate}
\end{definition}
In this paper we take $S=\mathbb{R}$.

\begin{theorem}[G\"{a}rtner-Ellis] \label{T:GE}
Let $\cbra{Y_n}_{n\in\mathbb{N}}$ be a collection of random
variables on a probability space $\probtriple$.  Let
$\cbra{r_n}_{n\in\mathbb{N}}$ be a sequence of positive reals such
that $\lim_{n\uparrow\infty} r_n = \infty$.  For each $n$ denote
by $\Lambda_n$ the cummulant generating function for $Y_n$
\begin{equation}\label{E:real_cgf}
\Lambda_n(\lambda) \dfn \log\left(\espalt{}{e^{\lambda
Y_n}}\right),\qquad \lambda\in\reals.
\end{equation}
Assume the following regarding $\Lambda_n$:
\begin{enumerate}
\item For all $\lambda\in\reals$ the limit $\Lambda(\lambda) \dfn \lim_{n\uparrow\infty}(1/r_n)\Lambda_n(r_n\lambda)$ exists as an extended real number.
\item $\mathcal{D}^0_\Lambda$, the interior of $\mathcal{D}_\Lambda\dfn \cbra{\lambda: \Lambda(\lambda)<\infty}$, is non-empty with $0\in \mathcal{D}^0_\Lambda$.
\item $\Lambda$ is differentiable throughout $\mathcal{D}^0_\Lambda$ and steep; i.e. $\lim_{\lambda\rightarrow \partial\mathcal{D}_{\Lambda}}|\nabla\Lambda(\lambda)| = \infty$.
\item $\Lambda$ is lower semi-continuous.
\end{enumerate}
Then, the random variables $\cbra{Y_n}_{n\in\mathbb{N}}$ satisfy a LDP with speed $\cbra{r_n}_{n\in\mathbb{N}}$ and good rate function $I(y) = \sup_{\lambda\in\reals}\left(\lambda y - \Lambda(\lambda)\right)$.
\end{theorem}

To connect Theorem \ref{T:GE} with the indifference price in \eqref{E:px_as_cgf}, assume
that the position size $q$ takes the form $q=\ell r_n$ for
$\ell\in\reals$, where $\cbra{r_n}_{n\in\mathbb{N}}$ is a sequence
of positive reals with $\lim_{n\uparrow\infty} r_n = \infty$.  In
this case, using \eqref{E:px_as_cgf} gives
\begin{equation}\label{E:p_gam_rel}
p^n_{a_n}(\ell r_n) = -\frac{1}{a_n\ell r_n}\log\left(\espalt{\qprob^n_0}{e^{-a_n\ell r_n\hat{Y}^n_{a_n}(a_n\ell r_n)}}\right) = -\frac{1}{a_n \ell r_n}\Gamma_n(-a_n\ell r_n),
\end{equation}
where, similarly to $\Lambda_n$ above, we set
\begin{equation}\label{E:gamma_n_def}
\Gamma_n(\lambda) \dfn \log\left(\espalt{\qprob^n_0}{e^{\lambda \hat{Y}^n_{a_n}(-\lambda)}}\right).
\end{equation}
We thus see that convergence of the indifference prices
$p^n_{a_n}(\ell r_n)$ is analogous to the G\"{a}rtner-Ellis
assumption that the scaled cummulant generating functions
$(1/r_n)\Lambda_n(\ell r_n)$ converge. However, besides the
dependence of probability measure on $n$, there is a substantial
difference between $\Gamma_n$ in \eqref{E:gamma_n_def} and
$\Lambda_n$ in \eqref{E:real_cgf}: namely, the random variables
$\hat{Y}^n_{a_n}(\lambda)$ of \eqref{E:gamma_n_def} are changing
with $\lambda$ whereas the random variables $Y_n$ of
\eqref{E:norm_resid_risk_p} are not. Thus, even though convergence
of the scaled indifference prices implies a connection with a LDP
for the random variables $\hat{Y}^n_{a_n}(\lambda)$, we do not
typically expect a LDP from random variables
$\hat{Y}^n_{a_n}(\lambda)$ unless they do not actually depend upon
$\lambda$. An example where this is the case is presented in
Section \ref{SS:SC} below.

We now make the main assumption in an analogous form to the
G\"{a}rtner-Ellis theorem.
\begin{ass}\label{A:GE_Opt1}
There exist a sequence $\cbra{r_n}_{n\in\mathbb{N}}$ of positive
reals with $\lim_{n\uparrow\infty} r_n = \infty$ and a $\delta>0$
such that for all $|\ell|<\delta$ the limit
\begin{equation}\label{E:limit_px}
p^\infty(\ell)\dfn \lim_{n\uparrow\infty} p^n_{a_n}(\ell r_n),
\end{equation}
exists and is finite. In particular, with
\begin{equation}\label{E:d_dn_def}
d_n \dfn p^n_{a_n}(0) = \espalt{\qprob^n_0}{B},\footnote{See
\cite{MR1891730} for a proof of this equivalence.}
\end{equation}
the limit $d\dfn p^\infty(0) =  \lim_{n\uparrow\infty} d_n$
exists. Furthermore, $p^\infty(\ell)$ is continuous at $0$,
i.e.~$\lim_{\ell\rightarrow 0} p^\infty(\ell) = d = p^\infty(0)$.
\end{ass}

\subsection{Discussion}\label{SS:discussion}

\subsubsection{Assumption \ref{A:GE_Opt1} and Vanishing Risk Aversion}\label{SSS:vanish_ra}

The  relation \eqref{E:ra_switch} allows us to vary risk aversion as well as position size. Specifically, Assumption \ref{A:GE_Opt1} takes the form that for
all $|\ell| < \delta$:
\begin{equation}\label{E:limit_px_a}
p^{\infty}(\ell) = \lim_{n\uparrow\infty}p^n_{a_n}(\ell r_n) = \lim_{n\uparrow\infty}p^n_1(\ell a_n r_n).
\end{equation}
From here, it immediately follows that if the market is fixed:
i.e. if $p^n_1(q_n) = p_1(q_n)$ for all $n$ and $q_n$, then if
$a_n\rightarrow 0$ we may set $r_n\dfn a_n^{-1}\rightarrow \infty$
and Assumption \ref{A:GE_Opt1} holds. Indeed, $p_1(\ell a_n r_n) =
p_1(\ell)=: p^\infty(\ell)$ and continuity at $0$ follows from
\cite{MR1891730} which shows that $\lim_{\ell\rightarrow
0}p^\infty(\ell) = d = \espalt{\qprob_0}{B}$. This example is
briefly additionally discussed in Section \ref{SS:vanish_ra}
below, and Theorems \ref{T:opt_pos_lb}, \ref{T:opt_pos_lb} not
withstanding, our focus in the sequel will lie primarily on the
case of fixed risk aversion in a sequence of varying markets.

\subsubsection{Assumption \ref{A:GE_Opt1} and Vanishing Hedging Errors}\label{SSS:vanish_hedge}

Though not explicitly stated, for a fixed risk aversion $a_n\equiv
a$, Assumption \ref{A:GE_Opt1} implies the hedging errors
associated $B$ are vanishing.  This follows both from the
convergence of scaled indifference prices $p^n_a(\ell r_n)$ and,
crucially, from the assumption that $p^\infty$ is continuous at
$0$. To see this latter point, consider again when the market is
fixed so $p^n_a(q_n) = p_a(q_n)$.  Here, for a bounded claim $B$,
as shown in \cite{MR1891730, MR2489605}, we have
\begin{equation*}
\lim_{n\uparrow\infty} p_a(\ell r_n) = \begin{cases}
\underset{\qprob\in\mathcal{\tilde{M}}}{\inf}\espalt{\qprob}{B}, & \ell > 0\\
\espalt{\qprob_0}{B}, & \ell = 0\\
\underset{\qprob\in\mathcal{\tilde{M}}}{\sup}\espalt{\qprob}{B}, &
\ell < 0.\end{cases}
\end{equation*}
Thus, the convergence requirement in Assumption \ref{A:GE_Opt1} holds, but the resultant function $p^\infty$ is not continuous at $0$, so Assumption \ref{A:GE_Opt1} cannot hold in a fixed market (or when there is a limiting market but $B$ is not replicable in this market).

Alternatively, consider when all of Assumption \ref{A:GE_Opt1}
holds. Firstly, \eqref{E:total_price} implies that $q\mapsto
p^n_{a_n}(q)$ is decreasing and $q\mapsto qp^n_{a_n}(q)$ is
concave. Thus, $\ell \mapsto \ell p^n_{a_n}(\ell r_n)$ is concave
as well and, for $|\ell|<\delta$, so is $\ell\mapsto \ell
p^\infty(\ell)$. In particular, $p^\infty(\ell)$ is continuous on
$(-\delta,0)$ and $(0,\delta)$.  Thus, additionally assuming
continuity of $p^\infty$ at $0$ (and hence on all of
$(-\delta,\delta)$), we obtain the useful result:
\begin{equation}\label{E:continuity}
\frac{q_n}{r_n}\rightarrow \ell \in (-\delta,\delta)\Longrightarrow p^n_{a_n}(q_n)\rightarrow p^\infty(\ell).
\end{equation}
Indeed, take $\eps>0$ so that $(\ell-\eps)r_n\leq q_n\leq (\ell+\eps)r_n$ for all $n$ large enough.  Since $p^n_{a_n}(q)$ is decreasing:
\begin{equation*}
p^\infty(\ell + \eps) = \lim_{n\uparrow\infty}
p^n_{a_n}((\ell+\eps)r_n) \leq
\liminf_{n\uparrow\infty}p^n_{a_n}(q_n) \leq
\limsup_{n\uparrow\infty} p^n_{a_n}(q_n) \leq
\lim_{n\uparrow\infty} p^n_{a_n}((\ell-\eps)r_n) =
p^\infty(\ell-\eps).
\end{equation*}
Taking $\eps\downarrow 0$ gives the result. In particular, for all
fixed position sizes $q$ and risk aversions $a$, we have that
$\lim_{n\uparrow\infty} p^n_a(q) = d$, and this essentially
implies the existence of trading strategies
$\pi^n\in\mathcal{A}^n$ which asymptotically hedge $B$.  This
argument is expanded upon, in the case of bounded claims and a
continuous filtration, in Section \ref{SS:opt_strat} below.

\subsubsection{On the strict concavity of $\ell\mapsto \ell p^\infty(\ell)$}\label{SSS:strict_concavity}

Even though $\ell\mapsto \ell p^\infty(\ell)$ is concave under
Assumption \ref{A:GE_Opt1}, as the example in Section
\ref{SS:discussion_2} below shows, it need not be strictly
concave.  However, under the assumption of strict concavity a
number of nice consequences ensue: for example, see Corollary
\ref{C:opt_pos} and the equilibrium results in Section
\ref{S:PEPQ}.

\section{Limiting Scaled Indifference Prices and Consequences}\label{S:consequences}

We now deduce a number of consequences of Assumption \ref{A:GE_Opt1}, the first of which is that the regime where the position size $q=q_n = \ell r_n$ is the appropriate one as $n\uparrow\infty$, if the considered positions are taken optimally. Here, we follow the approach of \cite{MR2212897, Robertson_2012, Robertson_Spil_2014}.

\subsection{Optimal Position Taking}\label{SS:opt_pos}

Define
\begin{equation}\label{E:bar_over_h}
\underline{B}_n \dfn
\inf_{\qprob\in\tM^n}\espalt{\qprob}{B},\qquad \bar{B}_n\dfn
\sup_{\qprob\in\tM^n}\espalt{\qprob}{B}.
\end{equation}
Assume, for all $n$, that $B$ cannot be replicated by trading in
$S^n$, and denote by $I^n$ the range of arbitrage free prices for
$B$: i.e.
\begin{equation}\label{E:arb_free_n}
I^n = (\underline{B}_n,\bar{B}_n).
\end{equation}
For $\tilde{p}^n\in I^n$ the optimal position $\hat{q}_n = \hat{q}_n(\tilde{p}^n)$ is defined as the unique (see \cite{MR2212897}) solution to the equation
\begin{equation}\label{E:opt_q_n}
\sup_{q\in\reals}\left(u^n_{a_n}(-q\tilde{p}^n,q)\right).
\end{equation}
As shown in \cite{MR2212897}, $\hat{q}_n$ satisfies the first order conditions for optimality
\begin{equation}\label{E:dual_measure}
\tilde{p}_n = \espalt{\qprob^{\hat{q}_n(\tilde{p}^n)}}{B},
\end{equation}
where $\qprob^{\hat{q}_n(\tilde{p}^n)}\in\tM^n$ is the dual optimizer for $\hat{q}_n(\tilde{p}^n)$ units of claim $B$ in that it achieves the infimum in \eqref{E:total_price}. To perform the asymptotic analysis we assume consistency (in $n$) between the markets and non-degeneracy in prices as $n\uparrow\infty$. More precisely:

\begin{ass}\label{A:px_range} For $\underline{B}_n,\bar{B}_n$ as in \eqref{E:bar_over_h} we have
\begin{equation}\label{E:asympt_arb_range}
\underline{B}\dfn \limsup_{n\uparrow\infty} \underline{B}_n < \liminf_{n\uparrow\infty}\bar{B}_n\rdfn \bar{B}.
\end{equation}
\end{ass}

\begin{remark} Let Assumption \ref{A:GE_Opt1} hold. Then, since $\underline{B}_n \leq d_n \leq \bar{B}_n$ for all $n$ it follows that $\underline{B}\leq d\leq \bar{B}$ (recall the definitions of $d_n$ and $d$ as given in Assumption \ref{A:GE_Opt1}). Assumption \ref{A:px_range} strengthens this to say that there are $\tilde{p}\neq d$ so that $\tilde{p}$ is arbitrage free for all $n$ large enough.  In particular, there are $I^n\ni \tilde{p}^n \rightarrow \tilde{p}\neq d$. Now, Assumption
\ref{A:px_range} may fail in two ways. First of all, it may be
that $I^n$ is collapsing to the singleton $d$ as
$n\uparrow\infty$. In this case, convergence of limiting prices is
trivial since $p^n_{a_n}(q_n)\rightarrow d$ for all sequences
$\cbra{q_n}$. The second way in which Assumption \ref{A:px_range}
may fail is if there is no consistency between markets in that
there is no price $\tilde{p}\neq d$ such that $\tilde{p}\in I^n$
for all $n$ large. Here, we do not have optimizers (along a subsequence) $\hat{q}_n$.
\end{remark}

Under Assumption \ref{A:px_range}, we present the first main
result, which says that optimal positions are becoming large at a
rate which grows at least like $\ell r_n$ for some $\ell\neq 0$.

\begin{theorem}\label{T:opt_pos_lb}
Let Assumptions \ref{A:claim}, \ref{A:no_arb_n}, \ref{A:GE_Opt1} and \ref{A:px_range} hold. For $I^n\ni \tilde{p}^{n}\rightarrow \tilde{p}$ we have
\begin{itemize}
\item If $\tilde{p} < d$ then
\begin{equation*}
\liminf_{n\uparrow\infty} \frac{\hat{q}_n(\tilde{p}^{n})}{r_n} > 0.
\end{equation*}
\item If $\tilde{p} > d$ then
\begin{equation*}
\liminf_{n\uparrow\infty}\frac{-\hat{q}_n(\tilde{p}^n)}{r_n} > 0.
\end{equation*}
\end{itemize}
\end{theorem}

The problem of obtaining upper bounds for $\limsup_{n\uparrow\infty} |\hat{q}_n(\tilde{p}^{n})|/r_n$ is more subtle. First of all we need to identify the maximal range where $p^n_{a_n}(\ell r_n)$ converges. To do this, set
\begin{equation}\label{E:delta_p}
\delta^+ \dfn \sup\cbra{k > 0 \ :\ \lim_{n\uparrow\infty} p^n_{a_n}(\ell r_n) = p^\infty(\ell),\ \forall\ 0 < \ell < k} \in [\delta,\infty].
\end{equation}
\begin{equation}\label{E:delta_m}
\delta_{-} \dfn \inf\cbra{k < 0\ :\ \lim_{n\uparrow\infty} p^n_{a_n}(\ell r_n) = p^\infty(\ell),\ \forall\ 0 > \ell > k} \in [-\infty,-\delta].
\end{equation}

As discussed in Section \ref{SS:discussion}, $p^n_{a_n}(q)$ is decreasing in $q$ and hence $p^\infty(\ell)$ is decreasing in $\ell$.  Therefore, the limits
\begin{equation}\label{E:p_infity_lim}
p^\infty(\delta^+)\dfn \lim_{\ell\downarrow \delta_{-}} p^\infty(\ell);\qquad p^\infty(\delta_{-})\dfn \lim_{\ell\uparrow\delta^+} p^\infty(\ell),
\end{equation}
exist.  Furthermore, since $\underline{B}_n< p^n_{a_n}(\ell r_n) < \bar{B}_n$ for all $\ell\in\reals$ we have $\underline{B}\leq p^\infty(\delta^+) \leq p^\infty(\delta_{-}) \leq \bar{B}$, however, as the example in Section \ref{SS:discussion_2} below shows, each of these inequalities may be strict.  In particular, the range of limiting indifference prices along the rate $r_n$ may deviate from the arbitrage free prices.

With this notation, we now provide the corresponding upper bounds for optimal positions.

\begin{theorem}\label{T:opt_pos_ub}
Let Assumptions \ref{A:claim}, \ref{A:no_arb_n}, \ref{A:GE_Opt1} and \ref{A:px_range} hold. Define $\delta^+, \delta_{-}$ as in \eqref{E:delta_p} and \eqref{E:delta_m} respectively. For $I^n\ni\tilde{p}^{n}\rightarrow \tilde{p}$ we have
\begin{itemize}
\item If $p^\infty(\delta^+) < \tilde{p} < d$ then
\begin{equation*}
\limsup_{n\uparrow\infty} \frac{\hat{q}_n(\tilde{p}^n)}{r_n} < \delta^+.
\end{equation*}
\item If $d < \tilde{p} < p^\infty(\delta_{-})$ then
\begin{equation*}
\limsup_{n\uparrow\infty} \frac{-\hat{q}_n(\tilde{p}^n)}{r_n} < -\delta_{-}.
\end{equation*}
\end{itemize}
\end{theorem}

Note the strict inequality above implies, for example, that when $\delta^+=\infty$ we have $\limsup_{n\uparrow\infty} \hat{q}_n(\tilde{p}^n)/r_n<\infty$. Lastly, let us discuss when one actually has true convergence.  As seen in Section \ref{SS:discussion} the map $\ell\mapsto \ell p^\infty(\ell)$ is concave. Here, we strengthen this by assuming:
\begin{assumption}\label{A:strict_conc}
The function $\ell\mapsto \ell p^\infty(\ell)$ is strictly concave on $(\delta_{-},\delta^+)$.
\end{assumption}

Then, we have the following Corollary which ensures the limit $\hat{q}_n/r_n$ actually exists:
\begin{corollary}\label{C:opt_pos}
Let Assumptions \ref{A:claim}, \ref{A:no_arb_n}, \ref{A:GE_Opt1}, \ref{A:px_range} and \ref{A:strict_conc} hold. Define $\delta^+, \delta_{-}$ as in \eqref{E:delta_p} and \eqref{E:delta_m} respectively. Let $I^n\ni\tilde{p}^{n}\rightarrow \tilde{p} $. If $p^\infty(\delta^+) < \tilde{p}\neq d < p^\infty(\delta_{-})$ then
\begin{equation*}
\lim_{n\uparrow\infty}\frac{\hat{q}_n(\tilde{p}^n)}{r_n} = \ell\in(\delta_{-},\delta^{+})\setminus\{0\}.
\end{equation*}
\end{corollary}

The proofs of Theorems \ref{T:opt_pos_lb}, \ref{T:opt_pos_ub} and of Corollary \ref{C:opt_pos} are in Appendix \ref{S:pf_opt_pos}.

\subsection{Discussion}\label{SS:discussion_2}

Presently, we point out some conclusions and subtleties associated
to the above results.  First, when we put together Theorems
\ref{T:opt_pos_lb}, \ref{T:opt_pos_ub} we see that if the price
$\tilde{p}^n\in I^n$ converges to $\tilde{p}$ where
$p^\infty(\delta^+)<\tilde{p}<p^\infty(\delta_{-}), p\neq d$ then
up to subsequences we have $\hat{q}_n(\tilde{p}^n)/r_n\rightarrow
\ell \in (\delta_{-},\delta^+)\setminus\{0\}$, which by Corollary
\ref{C:opt_pos} becomes true convergence if $\ell\mapsto\ell
p^\infty(\ell)$ is strictly concave. Note also that by
\eqref{E:continuity}, under optimal positions we have convergence
of indifference prices as well, i.e.
$p^n_{a_n}(\hat{q}_n(\tilde{p}^n)) \rightarrow p^\infty(\ell)$.

Second, assume for example that $\delta^+ = \infty$.  Then, another straightforward calculation shows (recall \eqref{E:asympt_arb_range})
\begin{equation*}
\underline{B} < \tilde{p} < \lim_{\ell\uparrow\infty}
p^\infty(\ell) \Longrightarrow \lim_{n\uparrow\infty}
\frac{\hat{q}_n(\tilde{p})}{r_n} = \infty,
\end{equation*}
provided of course such a $\tilde{p}$ exists. This offers a converse to Theorem \ref{T:opt_pos_ub}.

Third, let us briefly discuss the degenerate case where $r_{n}$ is (chosen) such  that $p^{\infty}(\ell)=d$ for all $\ell\in(\delta_{-},\delta^{+})$. In this case,  a range of different phenomena can occur. For illustration purposes, we consider the following example, taken from \cite{Robertson_Spil_2014}.  In the $n^{th}$ market, the claim decomposes into a replicable piece $D_n$ (with replicating
 capital $d_n$) and a piece $Y_n$ which is independent of $S^n$.  Now, assume $Y_n \sim N(0,\gamma_n)$ under $\prob$ and fix the risk aversion $a_n\equiv a$.  Here, the indifference price is
\begin{equation*}
p^n_{a}(q) = d_n -\frac{1}{aq}\log\left(\espalt{}{e^{-aqY_n}}\right) = d_n -\frac{1}{2}aq\gamma^2_n.
\end{equation*}

The range of arbitrage free prices is maximal: i.e. $\underline{B}_n
= -\infty$, $\bar{B}_n = \infty$. For $\tilde{p}^n\in\reals$ the
optimal purchase quantity found by minimizing $q\tilde{p}^n -
qp^n_{a}(q)$ is
\begin{equation*}
\hat{q}_n(\tilde{p}^n)  = -\frac{\tilde{p}^n-d_n}{a\gamma^2_n}.
\end{equation*}

Now, assume that $\gamma_n\rightarrow 0, d_n\rightarrow d$. With
$r_n = \gamma^{-2}_n\rightarrow\infty$, Assumption
\ref{A:GE_Opt1} holds with $p^\infty(\ell) = d-(1/2)a\ell$,
$\delta_{-}=-\infty$ and $\delta^+=\infty$. Note that $\ell p^\infty(\ell) = \ell d - (1/2)a\ell^2$ is strictly concave. Here, if
$\tilde{p}^n\rightarrow\tilde{p}\in\reals$ we have that
\begin{equation*}
\frac{\hat{q}_n(\tilde{p}^n)}{r_n} = -\frac{\tilde{p}^n-d_n}{a} \rightarrow -\frac{\tilde{p}-d}{a}.
\end{equation*}
So, both Theorems \ref{T:opt_pos_lb}, \ref{T:opt_pos_ub} hold.

Now, change $r_n$ so that $r_n = \gamma_n^{-1}\rightarrow \infty$. Then, Assumption
\ref{A:GE_Opt1} still holds with $p^\infty(\ell) = d$,
$\delta_{-}=-\infty$ and $\delta^+=\infty$.  In this instance, however, the map $\ell p^\infty(\ell) = \ell d$ is not strictly concave. Here, if
$\tilde{p}^n\rightarrow\tilde{p}\in\reals$ (which is still
arbitrage free since this property does not depend upon $r_n$) we
have
\begin{equation*}
\frac{\hat{q}_n(\tilde{p}^n)}{r_n} = -\frac{\tilde{p}^n-d_n}{a\gamma_n}.
\end{equation*}
So, if $\tilde{p}<d$ the ratio goes to $\infty$, if $\tilde{p} >
d$ the ratio goes to $-\infty$ and if $\tilde{p}=d$ then a variety
of phenomena can occur depending on the rates at which
$\gamma_n\rightarrow 0, \tilde{p}^n\rightarrow \tilde{p}$ and
$d_n\rightarrow d$.  Even though the behavior is degenerate in
this case, it \emph{does not} contradict either Theorem
\ref{T:opt_pos_lb} or \ref{T:opt_pos_ub}. In particular, Theorem
\ref{T:opt_pos_ub} is vacuous in this case since $p^\infty(\ell) =
d$ for all $\ell$.

The above example is related to the well known fact from Large Deviations that a LDP may hold for the same sequence of random variables with two different rates $\cbra{r_n},\cbra{r'_n}$ with $r_n/r'_n \rightarrow 0$. The resulting rate functions however, in an analogous manner to the resultant limiting indifference prices above, may provide drastically different levels of information.

\subsection{On the Normalized Optimal Wealth Process}\label{SS:opt_strat}

For a given $n$, fixed risk aversion $a$ and position size $q_n$, recall the optimal wealth process $\hat{X}^n_{a}(q_n)$ from Section \ref{S:model}. Heuristically, as $|q_n|\rightarrow\infty$ one expects $\hat{X}^n_a(q_n)$, as well as the optimal strategy $\hat{\pi}^n_a(q_n)$,  to grow on the order of $|q_n|$. However, if we normalize the wealth process by the position size then it is reasonable to ask if some type of convergence takes place. To this end we define the normalized wealth process $\tilde{X}$ via
\begin{equation}\label{E:norm_opt_wealth}
\tilde{X}^n_{a}(q_n) \dfn \frac{1}{q_n}\hat{X}^n_{a}(q_n).
\end{equation}
Note that $\tilde{X}^n_{a}(q_n)$ is in fact a wealth process, obtained from the (acceptable) normalized optimal trading strategy
$\tilde{\pi}^n_{a}(q_n) = (1/q_n)\hat{\pi}^n_{a}(q_n)$. We wish to stress that convergence of the normalized optimal wealth process
is a topic on its own and we do not study it in this paper. However, we mention some interesting and motivating straightforward conclusions.

Let us come back to \eqref{E:indiff_px}, re-written here as $-au^n_a(0)e^{-aq_n p^n_a(q_n)} =\espalt{}{e^{-q_na\left(\tilde{X}^n_a(q_n)_T+B\right)}}$. Since $-au^n_a(0)\leq 1$ we immediately see that
\begin{equation}\label{E:exp_trading_lb}
\espalt{}{e^{-q_na\left(\tilde{X}^n_a(q_n)_T + B - p^n_a(q_n)\right)}} = -au^n_a(0)\leq 1.
\end{equation}
By Markov's inequality we have the elementary estimate:
\begin{equation*}
\prob\bra{\tilde{X}^n_a(q_n)_T + B - p^n(q_n) \leq -\gamma} \leq e^{-q_n a\gamma};\qquad \gamma\in\reals.
\end{equation*}
Thus, we see that for any $q_n\uparrow\infty$ the portfolio obtained by buying one unit of $B$ for $p^n_a(q_n)$ and trading according to the normalized optimal trading strategy provides a super-hedge of $0$ in $\prob-$probability in that for all $\gamma > 0$
\begin{equation}
\lim_{n\uparrow\infty} \prob\bra{\tilde{X}^n_a(q_n)_T + B - p^n_a(q_n) \leq -\gamma} = 0,\label{Eq:SuperHedge}
\end{equation}
and in fact, the convergence to $0$ is exponentially fast.  This result essentially follows because of risk aversion and is valid
under the minimal Assumptions \ref{A:claim} and \ref{A:no_arb_n}. If we consider optimal positions then one can say more and characterize the super-hedge more precisely. We first adapt the set-up of \cite{MR2152255} and enforce the following assumptions on the claim $B$ and filtration $\filt$:
\begin{assumption}\label{A:bdd_claim}
$B$ is bounded: i.e. $\|B\|_{\Lb^{\infty}} < \infty$.
\end{assumption}

\begin{ass}\label{A:filt_cont}
The filtration $\filt$ is continuous.
\end{ass}

Under Assumptions \ref{A:bdd_claim}, \ref{A:filt_cont}, Theorem 13 of \cite{MR2152255}, says that for any $q_n$
\begin{equation}\label{E:indiff_strat_rel}
q_n B = q_n p^n_a(q_n) + \frac{a}{2}\langle \hat{L}^n_a(q_n)\rangle_T - \hat{L}^n_a(q_n)_T - \hat{X}_a^n(q_n)_T + \hat{X}_a^n(0)_T,
\end{equation}
where $\hat{L}^n_a(q_n)$ is a $\qprob^n_0$ martingale strongly orthogonal to $S^n$ under $\qprob^n_0$.  Dividing by $q_n$ and setting $\tilde{L}^n_a(q_n) = (1/q_n)\hat{L}^n_a(q_n)$ as the normalized orthogonal $\qprob^n_0$ martingale we obtain
\begin{equation}\label{E:indiff_strat_norm}
\tilde{X}^n_a(q_n)  + B - p^n_a(q_n) = \frac{aq_n}{2}\langle \tilde{L}^n_a(q_n)\rangle_T - \tilde{L}^n_a(q_n)_T + \frac{1}{q_n}\hat{X}^n_a(0).
\end{equation}
Next, as shown in \cite[Theorem 19]{MR2152255}, $\sup_n\left(q_n\espalt{\qprob^n_0}{\langle\tilde{L}^n_a(q_n)\rangle_T}\right) < \infty$, which implies that $\tilde{L}^n_a(q_n)_T$ goes to $0$ in $\qprob^n_0$-$L^2$ as  $q_n\rightarrow\infty$. Lastly, to evaluate $(1/q_n) \hat{X}^n_a(0)_T$ as $q_n\rightarrow\infty$ we impose the following mild asymptotic no arbitrage condition (see \cite[pp. 9]{Robertson_2012}):
\begin{assumption}\label{A:asympt_no_arb}
$\limsup_{n\uparrow\infty} \relent{\qprob^n_0}{\prob} < \infty$.
\end{assumption}
Assumption \ref{A:asympt_no_arb} implies $(1/q_n)\hat{X}^n_a(0)_T$ goes to $0$ in $\qprob^n_0$ probability as $q_n\rightarrow\infty$.  Indeed, using the first order relation in \eqref{E:min_ent_density} a straight-forward calculation shows that for any $\eps, q_n>0$ that
\begin{equation*}
\begin{split}
\qprob^n_0\bra{\frac{1}{q_n}\hat{X}^n_a(0)_T \geq \eps} &\leq e^{\relent{\qprob^n_0}{\prob} - aq_n\eps};\\
\qprob^n_0\bra{\frac{1}{q_n}\hat{X}^n_a(0)_T \leq -\eps} &\leq \frac{\relent{\qprob^n_0}{\prob} + e^{-1}}{\eps a q_n + \relent{\qprob^n_0}{\prob}},
\end{split}
\end{equation*}
from which the statement immediately follows. With these preparations, now consider when, additionally, Assumptions \ref{A:GE_Opt1} and \ref{A:px_range} hold, and positions are taking optimally: i.e.
$q_n = \hat{q}_n =  \hat{q}_n(\tilde{p}^n)$ where $I^n \ni
\tilde{p}^n \rightarrow \tilde{p}$ with $p^\infty(\delta^+) <
\tilde{p}< p^\infty(\delta_-),p\neq d$. Then, from Theorems
\ref{T:opt_pos_lb}, \ref{T:opt_pos_ub} we have up to subsequences
(or, under the Assumptions of Corollary \ref{C:opt_pos}, for all
subsequences) that $\hat{q}_n/r_n\rightarrow \ell \in
(\delta_{-},\delta^+)\setminus\{0\}$ and that $p^n_{a}(\hat{q}_n)
\rightarrow p^\infty(\ell)$. Thus, we obtain  that in
$\qprob^n_0$-probability
\begin{equation}
\tilde{X}_a^n(\hat{q}_n)_{T}+B-   p^\infty(\ell)-
\frac{a\hat{q}_n}{2}\langle L^n_{a}(\hat{q}_n)\rangle_T \rightarrow 0,\label{Eq:LimitingResidualRisk}
\end{equation}
which implies that the excess hedge is precisely
$a\hat{q}_n\langle \tilde{L}^n_a(q_n)\rangle_T/2$ in
$\qprob^n_0-$probability limit as $n\rightarrow\infty$. Even
though this result is interesting, one would like to have the same
statement under the $\prob$ measure. This is true if the measure
$\prob$ is contiguous with respect to the measure $\qprob^n_0$,
i.e. that $\qprob^n_0(A_{n})\rightarrow 0$ implies
$\prob(A_{n})\rightarrow 0$ for every sequence of measurable sets
$\{A_{n}\}_{n\in\mathbb{N}}$, e.g. Chapter 6 of
\cite{vanderVaart}. The classical Le Cam's first lemma (Lemma 6.4
in \cite{vanderVaart}) provides sufficient and necessary
conditions for contiguity.

Lastly, assume that $q_n = q$ is fixed and come back to \eqref{E:indiff_strat_norm}. Taking expectations yields
\begin{equation*}
d_n - p^n_a(q) = \frac{aq}{2}\espalt{\qprob^n_0}{\langle
\hat{L}^n_a(q)\rangle_T},
\end{equation*}
where we recall that $d_n = \espalt{\qprob^n_0}{B}$. As discussed
in Section \ref{SSS:vanish_hedge}, Assumption \ref{A:GE_Opt1}
implies $p^n_a(q)\rightarrow d$ and hence $\lim_{n\uparrow\infty}
\espalt{\qprob^n_0}{\langle \hat{L}^n_a(q)\rangle_T} = 0$ which in
turn implies that both $\langle \hat{L}^n_a(q)\rangle_T$,
$\hat{L}^n_a(q)_T$ go to zero in $\qprob^n_0$ probability as
$n\rightarrow\infty$. Therefore, for fixed position sizes, we have
in view of \eqref{E:indiff_strat_norm}, that $\tilde{X}^n_a(q)_T - (1/q)\hat{X}^n_a(0)_T +
B - d$ goes to zero in $\qprob^n_0$ probability and hence, under
the additional contiguity assumption, the claim is asymptotically
hedgeable.  This makes precise the connection between Assumption
\ref{A:GE_Opt1} and vanishing hedging errors mentioned in Section
\ref{SSS:vanish_hedge}.
\subsection{On a Characterization of $r_n$}  \label{S:CharacterizationOf_r_n}

As in the previous section, we let Assumptions \ref{A:no_arb_n}, \ref{A:GE_Opt1}, \ref{A:bdd_claim} and \ref{A:filt_cont} hold. Using the results of \cite{MR2152255}, we give a characterization for $r_n$ which in a sense justifies the interpretation of $r_{n}$ as the speed at which the market becomes complete. Recalling \eqref{E:d_dn_def}, (\ref{E:indiff_strat_rel}) and the normalized orthogonal martingale $\tilde{L}^n_a(q_n)$ we get
\begin{equation*}
d_n =  p^n_a(q_n)+ \frac{aq_n}{2} \espalt{\qprob^n_0}{\langle
\tilde{L}^n_a(q_n)\rangle_T}.
\end{equation*}

Now, let $q_{n}=\ell r_{n}$ for some $|\ell|<\delta$ (which, by Corollary \ref{C:opt_pos} and \eqref{E:continuity} essentially includes the case of optimal positions). We thus have
\begin{equation}
\lim_{n\uparrow\infty} \frac{r_{n}}{2}\espalt{\qprob^n_0}{\langle \tilde{L}^n_a(\ell r_n)\rangle_T} = \frac{d - p^\infty(\ell)}{a\ell}.\label{Eq:Limit_r_n}
\end{equation}

This conforms to the ``asymptotically complete'' case. The
normalized hedging error under optimal positions $\hat{q}_n\approx
\ell r_n$ is approximately (up to a multiplicative constant)
$\espalt{\qprob^n_0}{\langle \tilde{L}^n_a(\ell r_n)\rangle_T}$.
If the market is becoming complete we expect that for
$n\rightarrow\infty$
\[
\espalt{\qprob^n_0}{\langle \tilde{L}^n_a(\ell r_n)\rangle_T}\rightarrow 0.
\]
The speed at which it goes to $0$ thus becomes $r_n^{-1}$ and at this scaling we have convergence of prices.

In Sections \ref{S:examples} and \ref{SS:BS_trans} we study a number of examples where $r_{n}$ can be computed explicitly. One would like to have an abstract formula that explicitly characterizes $r_{n}$, as \eqref{Eq:Limit_r_n} contains $r_n$ within the normalized hedging error $\langle \tilde{L}^n_a(\ell r_n)\rangle$.  Notice that (\ref{Eq:Limit_r_n}) holds for all $|\ell| < \delta$. So, one is tempted to take limits as $\ell\rightarrow 0$ on both sides, and, if one can interchange the  $n\uparrow\infty$ limit with the $\ell\rightarrow 0$ limit, pass the latter limit inside the expectation, and if $p^{\infty}(\ell)$ is both strictly decreasing and differentiable at $\ell=0$, then for $n$ large enough
\begin{equation*}
r_{n}\approx -\frac{2\dot{p}^\infty(0)}{a}\times \frac{1}{\espalt{\qprob^n_0}{\langle \tilde{L}^n_a(0)\rangle_T}}.
\end{equation*}
Here, the interpretation of $r_n^{-1}$ as a market incompleteness factor is much more transparent.  Indeed, define $\check{X}^n, \check{L}^n$ through the Kunita-Watanabe decomposition of $-B$ with respect to the subspace of $L^2(\qprob^n_0;\F_T)$ generated by trading in $S^n$ so that $B = \espalt{\qprob^n_0}{B} - \check{L}^n_T - \check{X}^n_T$. Then, as shown in \cite[Section 6.1]{MR2152255} we have the following limits in $L^2(\qprob^n_0;\F_T)$
\begin{equation*}
\lim_{q\downarrow 0} \tilde{L}^n_a(q)_T = \check{L}_T;\qquad \lim_{q\downarrow 0}\left(\tilde{X}^n_a(q)_T - \frac{1}{q}\hat{X}^n_a(0)_T\right) = \check{X}^n_T.
\end{equation*}
In other words, $\tilde{L}^n_a(0)$ describes the hedging error associated to $B$, with size $\espalt{\qprob^n_0}{\langle \tilde{L}^n_a(0)\rangle_T} \propto r_n^{-1}$.  Thus $r_n^{-1}$ acts as the market incompleteness factor, and, as the market becomes complete, we see that $r_n\rightarrow\infty$.

The derivation of this statement is of course heuristic. Rigorous proof of this result seems to be quite hard, but we nevertheless
present the argument as  it provides more intuition into the problem. We choose to leave the rigorous derivation of this result
and further consequences as a future interesting work.

\subsection{Optimal Position Taking for General Utilities}\label{SS:opt_pos_gen_util}

The optimal position taking results in Theorems \ref{T:opt_pos_lb} and \ref{T:opt_pos_ub} readily extend to general utility functions on the real line.  This essentially follows from \cite{Robertson_2012}. Throughout this section we fix the risk aversion at $a>0$. Define $\Ua$ as the class of utility functions on $\reals$ (i.e. $U\in C^2(\reals)$, strictly increasing and strictly concave) satisfying
\begin{itemize}
\item The absolute risk aversion of $U$ is bounded between two positive constants: i.e. for $0 < \underline{a}_U < \bar{a}_U$:
\begin{equation}\label{E:u_bdd_ra}
\underline{a}_U \leq \alpha_u(x)\dfn -\frac{U''(x)}{U'(x)} \leq \bar{a}_U; \qquad x\in\reals.
\end{equation}
\item $U$ decays exponentially with rate $a$ for large negative wealths: i.e.
\begin{equation}\label{E:u_exp_decay}
\lim_{x\downarrow-\infty}-\frac{1}{x}\log(-U(x)) = a.
\end{equation}
\end{itemize}
By \eqref{E:u_bdd_ra} it follows that $U$ is bounded from above on $\reals$ and hence through a normalization we assume $0=U(\infty) = \lim_{x\uparrow\infty} U(x)$. From \cite[Section 2.2]{Robertson_2012} it holds that $U\in\Ua$ satisfies both the Inada conditions $\lim_{x\downarrow-\infty}U'(x) = \infty$, $\lim_{x\uparrow\infty} U'(x) = 0$ and the Reasonable Asymptotic Elasticity conditions $\liminf_{x\downarrow -\infty} xU'(x)/U(x) > 1$, $\limsup_{x\uparrow\infty} xU'(x)/U(x) < 1$. Similarly to \eqref{E:val_funct_claim} and \eqref{E:val_funct_no_claim}, define the value function in the $n^{th}$ market with initial capital $x$ and $q$ units of the claim as $u^n_U(x,q)$, where if $q=0$ we write $u^n_U(x)$.  Analogously to \eqref{E:util_indiff_px}, set  $p^n_U(x,q)$ as the (average, bid) utility indifference price defined through the equation
\begin{equation}\label{E:gen_u_indiff_px}
u^n_U(x-qp^n_U(x,q),q) = u^n_U(x).
\end{equation}
So that $p^n_U(x,q)$ is well defined for $x,q\in\reals$ we assume the claim is bounded: i.e. we enforce Assumption
\ref{A:bdd_claim}. Under Assumptions \ref{A:no_arb_n}, \ref{A:bdd_claim} it follows from \cite{MR2489605} that for $x,q\in\reals$,  $p^n_U(x,q)$ is well defined, arbitrage free, decreasing in $q$ with limits (recall \eqref{E:arb_free_n}) $\lim_{q\downarrow-\infty} p^n(x,q) = \bar{B}_n$, $\lim_{q\uparrow\infty} p^n(x,q) = \underline{B}_n $, for each
$n$.

To connect limiting prices for $U$ with those for the exponential utility we additionally enforce the asymptotic no arbitrage condition in Assumption \ref{A:asympt_no_arb}, and recall that using \cite[Theorem
3.3]{Robertson_2012}, it follows from Assumptions \ref{A:no_arb_n}, \ref{A:GE_Opt1}, \ref{A:bdd_claim} and \ref{A:asympt_no_arb} that for all $x\in\reals$ and $0 < |\ell| < \delta$:
\begin{equation}\label{E:gen_u_conv}
\lim_{n\uparrow\infty} p^n_U(x,\ell r_n) = p^\infty(\ell).
\end{equation}\
As for $\ell = 0$, since Assumption \ref{A:GE_Opt1} implies $p^\infty$ is continuous at $0$, the monotonicity of $p^n_U(x,q)$ yields for $0 < \ell < \delta$ that
\begin{equation*}
p^\infty(\ell) = \lim_{n\uparrow\infty} p^n_U(x,\ell r_n) \leq \liminf_{n\uparrow\infty} p^n_U(x,0) \leq \limsup_{n\uparrow\infty} p^n_U(x,0) \leq \lim_{n\uparrow\infty} p^n_U(x,-\ell r_n) = p^\infty(-\ell),
\end{equation*}
so that taking $\ell\downarrow 0$ we obtain that $p^n_U(x,0)\rightarrow p^\infty(0)$. Now, for a given arbitrage free price $\tilde{p}^n\in I^n$, we consider the optimal purchase problem
\begin{equation}\label{E:opt_q_gen}
\sup_{q\in\reals} \left(u^n_U(x-\tilde{p}^nq,q)\right).
\end{equation}
Unlike for the exponential case when the results of \cite{MR2212897} yield a unique maximizer, here, to the best our our knowledge, there are no known results on existence/uniqueness of optimizers (see \cite{MR3292128} for results with utility functions defined on the positive axis).  However, the main results of Theorems \ref{T:opt_pos_lb} and \ref{T:opt_pos_ub} still hold, as the following theorem shows.

\begin{theorem}\label{T:opt_pos_gen_u}
Let Assumptions \ref{A:no_arb_n}, \ref{A:GE_Opt1}, \ref{A:px_range}, \ref{A:bdd_claim} and \ref{A:asympt_no_arb} hold. Assume that $I^n\ni\tilde{p}^n\rightarrow \tilde{p}$. Let $x\in\reals$ be
fixed and recall $\delta^+,\delta_{-}$ from \eqref{E:delta_p},
\eqref{E:delta_m} respectively. Then
\begin{itemize}
\item For each $n$ there exists an optimizer $\hat{q}_n = \hat{q}_n(x,\tilde{p}^n)$ to \eqref{E:opt_q_gen}.
\item If $p^\infty(\delta^+) < \tilde{p} < d$ then for any sequence of maximizers $\cbra{\hat{q}_n}$:
\begin{equation}\label{E:opt_pos_pos_limits}
0 < \liminf_{n\uparrow\infty}\frac{\hat{q}_n}{r_n} < \limsup_{n\uparrow\infty} \frac{\hat{q}_n}{r_n} < \delta^+.
\end{equation}
\item If $d < \tilde{p} < p^\infty(\delta_{-})$ then for any sequence of maximizers $\cbra{\hat{q}_n}$:
\begin{equation}\label{E:opt_pos_neg_limits}
0 < \liminf_{n\uparrow\infty}\frac{-\hat{q}_n}{r_n} < \limsup_{n\uparrow\infty} \frac{-\hat{q}_n}{r_n} < -\delta_{-}.
\end{equation}
\end{itemize}
\end{theorem}

\begin{remark} As with the exponential case, a sufficient condition for the limits to exist in \eqref{E:opt_pos_pos_limits} and \eqref{E:opt_pos_neg_limits} is Assumption \ref{A:strict_conc}.
\end{remark}

\section{On Partial Equilibrium Price Quantity and its Limiting Behavior}\label{S:PEPQ}

The concept of indifference pricing has a subjective nature, in the sense that the
indifference price of an investor is a way she values
unhedgeable positions, and whether or not there is a counter-party
to offset a transaction is a different question. In particular, so
far we have assumed that a sequence of prices $\tilde{p}^n\in I^n$
converges to $\tilde{p}$, without mentioning whether such prices
equilibrate any transactions among different investors. In this
section, we address this issue and we justify that such sequence
of prices could indeed be the equilibrium prices of the given
claim $B$ among (two) investors.

For this, we adapt the notion of the partial equilibrium price
quantity (PEPQ). Provided that the stock dynamics are exogenously
specified, the equilibrium price of a claim $B$ is the one at
which the investors' optimal quantities of the claim sum up to
zero, meaning that the market of the claim is cleared out (the
word partial refers to the fact the investors specify the
equilibrium of the claim and not the stock market). Essentially,
the main motivation of this section is to study under Assumption
\ref{A:GE_Opt1} when our main optimal position taking results
could arise in an equilibrium setting whether all investors act
optimally and the price $\tilde{p}^n$ is the equilibrium price in the
$n^{th}$ market of a given claim $B$. In short, the analysis of this
section prove that if the investors' risky
exposures (random endowments) are dominated by $r_n$, then
$\tilde{p}^n\rightarrow d$. However, if investors' endowments are
growing like $r_n$, equilibrium prices $\tilde{p}^n$ could
converge to a limit different than $d$ and the results of Theorems
\ref{T:opt_pos_lb}, \ref{T:opt_pos_ub} occur. The latter
situation, which happens when at least one investor has an already
undertaken large position in $B$, means that there are cases where
the large regime is in fact the market's equilibrium, and even
more interestingly the equilibrium prices converge to a price
different than the unique limiting arbitrage free price.

In the setting of a locally bounded semi-martingale stock market,
bounded claims, and exponential utility maximizers, the PEPQ is
analyzed in \cite{MR2667897}. Specified to the current setup of
Section \ref{S:model}, we assume, for each $n$, there is a group
of $I$ investors such that each investor $i$ is endowed with a
exogenously given \textit{random endowment}, denoted by $\EN_n^i$.
For a given bounded claim $B$, the investors also wish to trade
$B$ amongst themselves in such a way that acting optimally (in
terms of utility maximization) the market for the claim clears.

For simplicity, we consider the presence of two investors,
although we should point out that the results of this section can
be generalized for markets with more investors. Recall that $I^n$
from \eqref{E:arb_free_n} denotes the (non-empty) range of
arbitrage free prices for $B$ and let $a^i_n>0$ denote the risk
aversion coefficient for investor $i$. Before we give the exact
definition of the PEPQ for a claim $B$, we need to introduce the
notation for the indirect utility and the indifference pricing
under the presence of random endowment. Namely, for the random
endowment $\EN_n^i$ and position size $q$ in $B$, define, in a
similar manner to \eqref{E:val_funct_claim}, the value function
for investor $i$ by
\begin{equation}\label{E:val_funct_claim_end}
u^n_{a^i_n}(x,q|\EN_i^n) \dfn \sup_{\pi^n\in\mathcal{A}^n}\espalt{}{U_{a^i_n}(x + X^{\pi^n}_T + qB+\EN_n^i)};\quad i=1,2.
\end{equation}
Similarly to \eqref{E:util_indiff_px}, the average (bid) indifference price of the investor $i$ with random endowment $\EN_i^n$ at the $n^{th}$ market is denoted by $p^n_{a^i_n}(q|\EN_n^i)$ and is given as the solution of
\begin{equation}\label{E:util_indiff_pEN}
u^n_{a^i_n}(x-q p^n_{a^i_n}(q|\EN_n^i), q|\EN_n^i) = u^n_{a^i_n}(x|\EN_n^i);\quad i=1,2.
\end{equation}
Note that the indifference  price's independence on the (constant)
initial wealth still holds under the presence of the random
endowment, which means that we can again assume $x=0$. Next, for a
given $p^n\in I^n$, consider the optimal purchase quantity problem
for investor $i$ defined by identifying (compare with
\eqref{E:opt_q_n}):
\begin{equation}\label{E:opt_q_n_end}
\hat{q}^i_n(p^n)=\underset{q\in\reals}{\text{argmax}}\left(u^n_{a^i_n}
(-qp^n,q|\EN_n^i)\right);\quad i=1,2.
\end{equation}
As shown in Proposition 5.5 in \cite{MR2667897}, the optimization  problem \eqref{E:opt_q_n_end} admits a representation similar to the
corresponding problem without random endowment (see \eqref{E:optimization}). Namely, we have that
\begin{equation}\label{E:optimization_end}
\hat{q}^i_n(p^n)\in\text{argmin}_{q\in\mathbb{R}}\left(q\tilde{p}^{n}-q
p^n_{a^i_n}(q|\EN_n^i)\right).
\end{equation}

A PEPQ is then defined as a pair $(\epn,\eqn)\in I^n\times\R$ such
that
\[\eqn=\hat{q}^1_n(\epn)\quad\text{and}\quad
-\eqn=\hat{q}^2_n(\epn).\]
In other words, at price $\epn$ it is
optimal for investor $1$ to buy $\eqn$ and investor $2$ to sell
$\eqn$ units of $B$, thus the market clears out. Taking representation \eqref{E:optimization_end} into account, it is
then a matter of simple calculations to get the following
condition for the PEPQ for each $n$ (see also Proposition 5.6
and Corollary 5.7 in \cite{MR2667897}):
\begin{equation}\label{eq:PEQ condition}
\eqn=\underset{q\in\R}{\text{argmax}}\left(q\left(p^n_{a^1_n}(q|\EN_1^n)+p^n_{a^2_n}(-q|\EN_2^n)\right)\right).
\end{equation}
The equilibrium price $\epn$ is then given by
\begin{equation}\label{E:PEPQprice}
\epn=\espalt{\qprob_1^n(\eqn)}{B}=\espalt{\qprob_2^n(-\eqn)}{B},
\end{equation}
where $\qprob_i^n(q)$ denotes the dual optimizer in $\tilde{\M}^n$
for the position $qB+\EN^n_i$ and risk aversion $a_n^i$ (recall
the first order condition \eqref{E:dual_measure} without random
endowment)\footnote{Note that $\qprob_i^n(0)$ is not necesarily $\qprob^n_0$ due to the presence of $\EN^n_i$.}. According to Theorem 5.8 in \cite{MR2667897}, for a non-replicable bounded claim $B$ (i.e. satisfying Assumption \ref{A:bdd_claim}) a PEPQ $(\epn,
\eqn)\in I^n\times\R$ always exists for each $n\in\N$, and it is unique with $\eqn\neq 0$ if and only if $a^1_n\EN^n_1-a^2_n\EN^n_2$ is non-replicable.

Now, consider when $n\uparrow\infty$ and Assumption \ref{A:GE_Opt1} holds for each sequence $\{a^i_n\}_{n\in\N}$. The questions that naturally arise are where the sequence of the equilibrium prices converges to and under which conditions the regime of Theorems \ref{T:opt_pos_lb}, \ref{T:opt_pos_ub} occurs.
As $n\uparrow\infty$, if one ignores the position size and has non-vanishing risk aversion, the
hedging error of positions in $B$ approaches zero and hence it is
expected that equilibrium prices converge to price $d$. It turns
out that this is the case provided however that the size of the
investors' endowments is dominated by the ``market incompleteness''
parameter $r_n$ from Assumption \ref{A:GE_Opt1}. When at least one
of the endowments increases with $n$ sufficiently fast, the
equilibrium prices may converge to a limit different than $d$,
which implies a situation similar to the regime of Theorems
\ref{T:opt_pos_lb}, \ref{T:opt_pos_ub}. In the sequel we provide a
family of such examples where the endowment of one of the investor
is an increasing position on the claim $B$.

Before, we present the precise arguments we should clarify how
Assumption \ref{A:GE_Opt1} works in the case of two investors,
$i=1,2$. The statement that Assumption \ref{A:GE_Opt1} holds for
function $p^n_{a^i_n}:\R\mapsto I^n$ (defined in
\eqref{E:util_indiff_px}), means that there exist a sequence
$\cbra{r^i_n}_{n\in\mathbb{N}}$ of positive reals with
$r^i_n\nearrow\infty$ and a constant $\delta_i>0$ such that for
all $|\ell|<\delta_i$ the limit $p_i^\infty(\ell)\dfn
\lim_{n\uparrow\infty} p^n_{a^i_n}(\ell r^i_n)$ exists, is finite
and $\lim_{\ell\rightarrow 0} p_i^\infty(\ell) = d $. Note that it
readily follows from the relation $p_{a_n^2}^n(q)=p_{a_n^1}^n(q
a_n^2/a_n^1)$ (which holds for each $n$) that if Assumption
\ref{A:GE_Opt1} holds for function $p_{a_n^1}^n$, it will also
hold for function $p_{a_n^2}^n$ provided that the sequence
$\{a_n^2/a_n^1\}_{n\in\N}$ is bounded away from zero and infinity.
For this, we could set $r_n^2\dfn r_n^1a_n^2/a_n^1$ (possibly
going to an increasing subsequence), $p_2^{\infty}=p_1^{\infty}$
and $\delta_2=\delta_1$.

\smallskip
For the proofs of this section we need to introduce the notion of
the (bid) indifference price for every arbitrary bounded payoff
$C\in\Lb^{\infty}$ under risk aversion $a_n>0$ in the $n^{th}$
market, denoted by $\pr^n_{a_n}(C)$ and defined as the solution of
the following equation
\begin{equation}\label{E:util_indiff}
\sup_{\pi^n\in\mathcal{A}^n}\espalt{}{U_{a_n}(x + X^{\pi^n}_T + C-\pr^n_{a_n}(C))}=\sup_{\pi^n\in\mathcal{A}^n}\espalt{}{U_{a_n}(x + X^{\pi^n}_T)};\quad i=1,2.
\end{equation}
Note that under this notation $qp^n_{a_n}(q)=\pr^n_{a_n}(qB)$, for
all $q\in\R$ with $p^n_{a_n}$ defined in \eqref{E:util_indiff_px}.
The following Lemma generalizes the findings of Theorems
\ref{T:opt_pos_lb} and \ref{T:opt_pos_ub} under the presence of
random endowment provided that the endowment is dominated by the
associated $r_n$.

\begin{lemma}\label{L:endowments}
Let Assumptions \ref{A:no_arb_n}, \ref{A:px_range},
\ref{A:bdd_claim} hold and impose Assumption \ref{A:GE_Opt1} for
function $p^n_{a^i_n}:\R\mapsto I^n$. If for $i=1,2$, $\EN_i^n\in
\Lb^{\infty}$, for each $n$ and
$||\EN_i^{n}||_{\Lb^{\infty}}/r^i_n\rightarrow 0$, then the
statements of Theorems \ref{T:opt_pos_lb} and \ref{T:opt_pos_ub}
hold also for the function $p^n_{a^i_n}(\cdot|\EN_i^n):\R\mapsto
I^n$.
\end{lemma}

\begin{proof}
In view of the proof of Theorem \ref{T:opt_pos_lb} and under the imposed assumptions, we first have to show that function $p^n_{a^i_n}(\cdot|\EN_i^n):\R\mapsto I^n$ satisfies Assumption \ref{A:one_p}. Indeed, the first bullet point follows
by a simple change of measure $d\prob_i^n/d\prob:=c^n_ie^{-a^i_n\EN_i^n}$, for some constant
$c^n_i$ and the corresponding variational representation of the
indifference price \eqref{E:total_price} considered under measure
$\prob_i^n$; while the second bullet point readily follows by the
boundedness of claim $B$. For the third and forth items, it is enough to show that for all $|\ell|<\delta_i$, $\lim_{n\rightarrow\infty}p^n_{a^i_n}(\ell r^i_n|\EN_i^n)=p^{\infty}_i(\ell)$. For this, we note that the indifference price of an exponential utility maximizer under some random endowment can be written as the difference of two indifference prices without endowments (see among others, Appendix of \cite{MR2667897} and recall definition \eqref{E:util_indiff}):
\begin{equation}\label{Eq: conditional indifference price}
qp^n_{a^i_n}(q|\EN_i^n)=\pr^n_{a^i_n}(qB+\EN^n_i)-\pr^n_{a^i_n}(\EN^n_i),\quad\forall q\in\reals,
\end{equation}
Hence, for any $|\ell|<\delta_i$
\[p_{a_n^i}^n(\ell r^i_n|\EN^n_i)=\frac{\pr^n_{a_n^i}(\ell r^i_nB+\EN^n_i)-\pr^n_{a_n^i}(\EN^n_i)}{\ell r^i_n}
    \leq p^n_{a_n^i}(\ell r^i_n)+2\frac{||\EN_i^{n}||_{\Lb^{\infty}}}{|\ell| r^i_n}
    \rightarrow p_i^{\infty}(\ell),
\]
 where the limiting argument follows by the imposed assumptions on function $p^n_{a_n^i}$ and $\EN_i^n$. We similarly show that $p_{a_n^i}^n(\ell r^i_n|\EN^n_i)\geq p^n_{a_n^i}(\ell r^i_n)-2\frac{||\EN_i^{n}||_{\Lb^{\infty}}}{|\ell| r^i_n}\rightarrow p_i^{\infty}(\ell)$, which finishes the proof that function $q\mapsto
p^n_{a^i_n}(q|\EN_i^n)$ satisfies Assumption \ref{A:one_p}. We
then observe that requirements of Proposition \ref{prop: one_p}
are also met for function $p^n_{a^i_n}(\cdot|\EN_i^n):\R\mapsto
I^n$, since by \eqref{Eq: conditional indifference price} it
readily follows that
$p^n_{a^i_n}(\infty|\EN_i^n)=p^n_{a^i_n}(\infty)$. Hence, the rest
of the proof follows the same argument lines as the ones in proofs
of Theorems \ref{T:opt_pos_lb}, \ref{T:opt_pos_ub}.
\end{proof}

Returning to the PEPQ, we exclude trivial cases for each $n\in\N$
by imposing the following assumption.
\begin{assumption}\label{A:endowments}
For each $n$, $\EN_i^n\in \Lb^{\infty}$ for both $i=1,2$ and
$a_n^1\EN^n_1-a_n^2\EN^n_2$ is non-replicable.
\end{assumption}

As mentioned above, this assumption guarantees the existence and
the uniqueness of the PEPQ $(\epn,\eqn)$ for each $n$ with
$\eqn\neq 0$. Imposing Assumption \ref{A:GE_Opt1} for indifference
prices of both investors, we first address the conditions that
give the convergence of the equilibrium prices to $d$.

\begin{proposition}\label{P:PEPQ_bounded_endowments}
Let Assumptions \ref{A:no_arb_n}, \ref{A:px_range},
\ref{A:bdd_claim}, \ref{A:endowments} hold, and impose Assumption
\ref{A:GE_Opt1} for function $p_{a_n^1}^n(q)$ and Assumption
\ref{A:strict_conc} for function $qp_1^{\infty}(q)$. If we further
assume that $||\EN_i^{n}||_{\Lb^{\infty}}/r^1_n\rightarrow 0$, for
both $i=1,2$ and the sequence $\{a_n^2/a_n^1\}_{n\in\N}$ is
bounded away from zero and infinity, the sequence of the partial
equilibrium prices $\epn$ of claim $B$ converges to $d$.
\end{proposition}

\begin{proof}
Let $\epn$ denote an arbitrarily chosen convergent subsequence of
the equilibrium prices of $B$ with limit $\hat{p}$ (note that
$B\in\Lb^{\infty}$ guarantees the existence of such subsequence)
and assume that $\hat{p}\neq d$, and in particular $\hat{p}< d$.

Under Assumptions \ref{A:bdd_claim} and \ref{A:endowments}, it
follows by Theorem 5.1 of \cite{MR2212897} that the map $q\mapsto
qp_{a_n^i}^n(q|\EN^n_i)$ is strictly concave for each $i=1,2$, and
also that
\begin{equation}\label{E:strict conca}
    \espalt{\qprob_i^{n}(q)}{B}=\frac{\partial}{\partial
q}qp_{a_n^i}^n(q|\EN^n_i).
\end{equation}
Now, that $\espalt{\qprob_1^{n}(0)}{B}\neq\espalt{\qprob_2^{n}(0)}{B}$ holds due to Assumption \ref{A:endowments}. Thus, first assume for some subsequence (still labeled $n$) that $\espalt{\qprob_1^{n}(0)}{B}>\espalt{\qprob_2^{n}(0)}{B}$, for sufficiently large $n$. Then $\eqn>0$ and in fact
$\espalt{\qprob_1^{n}(0)}{B}>\epn>\espalt{\qprob_2^{n}(0)}{B}$. In view of Theorem
\ref{T:opt_pos_lb} and Lemma \ref{L:endowments}, we have that the
inequality $\hat{p}< d$ implies the existence of a further subsequence of
$\eqn$ (still labeled $n$) such that
$\lim_{n\rightarrow\infty}\eqn/r_n^1 = \ell > 0$. We reach then a
contradiction if we show that for sufficiently large $n$, the
position $-\eqn$ is not optimal for investor 2. Since $\hat{p}<
d$, we get from Assumption \ref{A:GE_Opt1} that there exists $c>0$
such that for any sufficiently large $n$,
$\epn<\espalt{\qprob^{n}_0}{B}-c$. This implies that
\[0\leq (-\eqn)\left(p_{a_n^2}^n(-\eqn|\EN^n_2) -\epn\right)<
(-\eqn)\left(p_{a_n^2}^n(-\eqn|\EN^n_2)
-\espalt{\qprob^{n}_0}{B}+c\right),\] where the first inequality
holds because the position $-\eqn$ is optimal for investor 2 at
price $\epn$, for each $n$. Using the relation \eqref{Eq:
conditional indifference price} and the representation
\eqref{E:total_price} we get that (recall definition
\eqref{E:util_indiff_pEN})
\begin{align*}
0 & < \frac{\pr_{a_n^2}^{n}(-\eqn B+\EN_2^{n})}{\eqn}-\frac{\pr_{a_n^2}^{n}(\EN_2^{n})}{\eqn}+\espalt{\qprob^{n}_0}{B}-c\\
&=
\inf_{\qprob\in\tM^{n}}\left\{\espalt{\qprob}{-B+\frac{\EN_2^{n}}{\eqn}}
+ \frac{1}{a_n^2\eqn}\left(\relent{\qprob}{\prob} -
\relent{\qprob^{n}_0}{\prob}\right)\right\}-\frac{\pr_{a_n^2}^{n}(\EN_2^{n})}{\eqn}+\espalt{\qprob^{n}_0}{B}-c\\
& \leq
\espalt{\qprob^{n}_0}{\frac{\EN_2^{n}}{\eqn}}-\frac{\pr_{a_n^2}^{n}(\EN_2^{n})}{\eqn}-c\leq
2\frac{||\EN_2^{n}||_{\Lb^{\infty}}}{\eqn}-c = 2\frac{||\EN_2^n||_{\Lb^{\infty}}}{r^1_n}\frac{r^1_n}{\eqn} - c.
\end{align*}

Since $||\EN_2^n||_{\Lb^{\infty}}/r^1_n\rightarrow 0$ and $r^1_n/\eqn\rightarrow 1/\ell$ it follows that $c\leq 0$, a contradiction since $c>0$. Similarly, when $\espalt{\qprob_1^{n}(0)}{B}<\espalt{\qprob_2^{n}(0)}{B}$, for sufficiently large $n$, then $\eqn<0$ and up to a subsequence $\eqn/r^2_n\rightarrow -\ell < 0$. In this case, we follow the same arguments to show that the position $-\eqn$ could not be optimal for the investor 1 for sufficiently large $n$. Finally, the case where $\hat{p}>d$ is symmetric to the analysis above and hence omitted.
\end{proof}

Withdrawing however the assumption
$||\EN_i^{n}||_{\Lb^{\infty}}/r_n\rightarrow 0$ could give the
interesting cases where the equilibrium prices converge to a price
different than the unique arbitrage free price of the limiting
market and the regime of Theorems \ref{T:opt_pos_lb},
\ref{T:opt_pos_ub} occurs. A family of such examples are presented
in the following Proposition.

\begin{proposition}\label{P:PEPQ_example}
Let Assumptions \ref{A:no_arb_n}, \ref{A:px_range} and
\ref{A:bdd_claim} hold. Impose also Assumption \ref{A:GE_Opt1} for
function $p_{1}^n(p)$ with constant risk aversion equal to 1 and
Assumption \ref{A:strict_conc} for the corresponding function
$qp^{\infty}(q)$. If for each $n\in\N$ and $i=1,2$, $a_n^i\equiv
a_i$ and $\EN_i^n\equiv b_i^n B$, for some $a_i>0$ and
$b_i^n\in\R$, the following statements hold:
\begin{itemize}
\item [i.] For each market $n\in\N$, the unique PEPQ pair
$(\epn,\eqn)$ is given by $q_n^*=(a_2b_2^n-a_1b_1^n)/(a_1+a_2)$
and $\epn=\espalt{\qprob^{-ab^n}}{B}$, with $1/a:=1/a_1+1/a_2$ and
$b^n:=b_1^n+b_2^n$. \item [ii.] Letting for each $n\in\N$,
$b_2^n=\kappa r_n$, for some $\kappa\in (0,\delta_+/a)$ and
$b_1^n=b_1\in\R$, we get that
$\lim_{n\rightarrow\infty}\eqn/r_n=\ell>0$ and $\epn\rightarrow
\hat{p}<d$.
\end{itemize}
\end{proposition}

\begin{proof}
The proof of the first item i.~is based on  standard arguments
of the related literature (see for example Theorem 3.2 in
\cite{ElBarr05}). We recall that the equilibrium quantity is the
solution of the  optimization problem \eqref{eq:PEQ condition} and
thanks to the strict concavity of the function $q\mapsto
qp^n_{a_i}(q|\EN_i^n)$ we get that for any $q\in\R$ and every
$n\in\N$,
\[q\left(p^n_{a_1}(q|\EN_1^n)+p^n_{a_2}(-q|\EN_2^n)\right)\leq b^np^n_{a}(b^n).\]
We then observe that in fact
$b^np^n_{a}(b^n)=\eqn\left(p^n_{a_1}(\eqn|\EN_1^n)+p^n_{a_2}(-\eqn|\EN_2^n)\right)$,
which means that $\eqn$ is indeed the equilibrium quantity. The
fact that equilibrium price $\epn$ equals to
$\espalt{\qprob^{-ab^n}}{B}$  readily follows by
\eqref{E:PEPQprice}.

For the second item, we have that $q_n^*/r_n=(a_2\kappa
r_n-a_1b_1)/(a_1+a_2)\rightarrow a_2\kappa/(a_1+a_2)>0$. Since
$\epn$ is the equilibrium price for each $n$, we have that $\epn<
p^n_1(\eqn|\EN_1^n)$, since $\eqn$ is optimal position for
investor 1 at price $\epn$. Then by using the representation
\eqref{Eq: conditional indifference price} as in the proof of
Lemma \ref{L:endowments}, we get that
\[\lim_{n\rightarrow\infty}p^n_{a_1}(q_n^*|\EN_1^n)=\lim_{n\rightarrow\infty}p^n_{a_1}(a_2\kappa r_n/(a_1+a_2))=p^{\infty}(a\kappa).\]

Recall that $\epn=\espalt{\qprob^{-ab^n}}{B}$ and note that strict
concavity of the function $q\mapsto qp^n_{a_1}(q|\EN_1^n)$ and
equation \eqref{E:strict conca} give that $\epn$ is decreasing in
$n$ and hence it has a limiting point $\hat{p}$. Thus, we have
that $\lim_{n\rightarrow\infty}p^*_n=\hat{p}\leq
p^{\infty}(a\kappa)<p^{\infty}(0)=d$, where the last strict
inequality follows by Assumption \ref{A:strict_conc}.
\end{proof}

Proposition \ref{P:PEPQ_example} indicates that there are cases
where the equilibrium quantity increases to infinity at the same
time where the equilibrium price is different than the limiting
arbitrage free price. It is important to point out here that both
investors act optimally at that equilibrium prices even though the
limiting price is different than $d$. The essential element is of
course that one of the investor is endowed with a large position
on the claim and she is willing to sell portion of her position at
a price which induces the other investor acting optimally to enter
to a large claim regime too. In other words, Proposition
\ref{P:PEPQ_example} justifies the large volume of some OTC
derivative markets and the corresponding extreme prices as long as
some of the participants in the market are already exposed to a
risk that is highly correlated with the payoff of the tradeable
derivatives. This situation fits to the observed extreme volumes
and prices for example in the Mortgage Backed Securities market in
the recent years.

\begin{remark}
The proof of Proposition \ref{P:PEPQ_example} can easily be
generalized in the case where the endowments are of the form
$\EN_i^n=b_i^nB+E^n_i,$ with the choices of $b_i^n$ as in the
Proposition \ref{P:PEPQ_example} and $E^n_i$ being bounded random
endowments such that $||E_i^{n}||_{\Lb^{\infty}}/r_n\rightarrow
0$.
\end{remark}

\section{Examples where the limiting scaled indifference price exist}\label{S:examples}

The power of Assumption \ref{A:GE_Opt1} is its validity in a wide
variety of models. In this section we give four well studied
market model examples. Then, in the next section we pay particular
attention to an example with transactions costs. Remarkably, even
though the standard duality results no longer apply, a version of
Assumption \ref{A:GE_Opt1} still holds and more importantly, so do
the conclusions of Theorems \ref{T:opt_pos_lb} and
\ref{T:opt_pos_ub}.

\subsection{Vanishing Risk Aversion in a Fixed Market}\label{SS:vanish_ra}

As shown Section \ref{SSS:vanish_ra} for a fixed market, if the risk aversion vanishes (i.e. $a_n\rightarrow 0$) then Assumption \ref{A:GE_Opt1} holds with $r_n = a_n^{-1}$ and $p^\infty(\ell) = p_1(\ell)$.  In addition, as the class of acceptable trading strategies $\mathcal{A}$ is a cone it follows for any $q_n$ that $\hat{\pi}_{a_n}(q_n) = (1/a_n)\hat{\pi}_1(a_nq_n)$. So, for $q_n = \ell r_n = \ell/a_n$, not only do indifference prices trivially converge, but the optimal trading strategy is explicitly known, i.e. it is $(1/a_n)\hat{\pi}_1(\ell) = r_n\hat{\pi}_1(\ell) = (q_n/\ell)\hat{\pi}_1(\ell)$.  Note that in this instance the normalized optimal trading strategy trivially converges but does not necessarily provide a super hedge.

\subsection{Basis Risk Model with High Correlation}\label{SS:br}

This example is considered in detail in \cite{davis1997opi, MR1926237, Robertson_2012, MR2094149} amongst others. Here, we have for each $n$ one risky asset $S^n$ which evolves according to
\begin{equation*}
\begin{split}
\frac{dS^n_t}{S^n_t} =& \mu(Y_t)dt + \sigma(Y_t)\left(\rho_n dW_t + \sqrt{1-\rho_n^2}d\tilde{W}_t\right),\\
dY_t =& b(Y_t)dt + a(Y_t)dW_t,
\end{split}
\end{equation*}
where $W$ and $\tilde{W}$ are two independent Brownian motions. The filtered probability space is the standard two-dimensional augmented Wiener space. The coefficients $a,b$ have appropriate regularity and are such that $Y$ has a unique strong solution taking values in an open subset $E$ of $\reals$. Set $\lambda \dfn \mu/\sigma$ as the market price of risk and assume that $\sigma^2(y)>0,y\in E$ and that $\lambda$ is bounded on $E$. $B = B(Y_T)$ for some continuous bounded function $B$ on $E$. As shown in \cite[Section 5.3]{Robertson_2012}, $\underline{B}_n = \underline{B} = \inf_{y\in E}B(y)$ and $\bar{B}_n = \bar{B}= \sup_{y\in E}B(y)$ for all $n$. Set $r_n = (1-\rho^2_n)^{-1}$. As shown in \cite{MR2094149} (see also \cite{Robertson_2012}), for a fixed risk aversion $a>0$ and $\ell\in\reals, \ell\neq 0$:
\begin{equation*}
p^n_a\left(\ell r_n\right) =
-\frac{1}{a\ell}\log\left(\frac{\espalt{}{e^{-\rho_n\int_0^T\lambda(Y_t)dW_t
- \frac{1}{2}\int_0^T\lambda^2(Y_t)dt - a\ell
B(Y_T)}}}{\espalt{}{e^{-\rho_n\int_0^T\lambda(Y_t)dW_t -
\frac{1}{2}\int_0^T\lambda^2(Y_t)dt}}}\right).
\end{equation*}
For $\ell = 0$ one has
\begin{equation*}
d_n = p^n_a(0) = \espalt{\qprob^n_0}{B(Y_T)} = \frac{\espalt{}{e^{-\rho_n\int_0^T\lambda(Y_t)dW_t -
\frac{1}{2}\int_0^T\lambda^2(Y_t)dt}B(Y_T)}}{\espalt{}{e^{-\rho_n\int_0^T\lambda(Y_t)dW_t -
\frac{1}{2}\int_0^T\lambda^2(Y_t)dt}}}.
\end{equation*}
Thus, if $\rho_n\rightarrow 1$ (limit of high correlation) then $r_n\rightarrow\infty$ and
\begin{equation*}
\begin{split}
\lim_{n\uparrow\infty} p^n_a(\ell r_n) &= p^\infty(\ell) =
-\frac{1}{a\ell}\log\left(\espalt{\qprob}{e^{-a\ell
B(Y_T)}}\right);\qquad \ell\neq 0;\\
\lim_{n\uparrow\infty} p^n_a(0) &= p^\infty(0) = \espalt{\qprob}{B(Y_T)},
\end{split}
\end{equation*}
where $\qprob$ is the unique martingale measure in the $\rho=1$ market where the filtration is restricted to $\mathbb{F}^W$. Furthermore, using l'Hopital's rule one obtains $\lim_{\ell\rightarrow 0}p^\infty(\ell) = \espalt{\qprob}{B(Y_T)} = p^\infty(0)$ so that Assumption \ref{A:GE_Opt1} is satisfied with $\delta= \infty$.

\subsection{Large Markets with Vanishing Trading Restrictions}\label{SS:SC}

The next example is simplified version of the general semi-complete setup considered in \cite{Robertson_Spil_2014}.  Here, $\probtriple$ is assumed to support a sequence of independent Brownian motions $W^1,W^2,...$. The filtration is the augmented version of $\filt^{W^1,W^2,...}$.  There is a sequence of (potentially tradeable) assets $S^1,S^2,...$ with dynamics
\begin{equation*}
\frac{dS^i_t}{S^i_t} = \mu^i dt + \sum_{j=1}^i\sigma^{ij} dW^j_t;\qquad i = 1,2,3,...,
\end{equation*}
where $\mu = (\mu^1,\mu^2,...)$ satisfies $\sum_{i=1}^\infty (\mu^i)^2 < \infty$ and $\sigma$ is the lower triangular square root of the symmetric matrix $\Sigma = \cbra{\Sigma^{ij}}_{i,j=1,2,...}$, assumed positive definite so that for some $\lambda > 0$ and all $\xi = (\xi^1,\xi^2,...)$ with $\sum_{i=1}^\infty (\xi^i)^2<\infty$, we have $\xi'\Sigma\xi \geq \lambda \xi'\xi$.

The claim (as is typical in life insurance markets) is given as the sum of independent, $\filt^{W^i}$ adapted claims $B^i$: $B = \sum_{i=1}^\infty B^i$. To make $B$ well defined and amenable to large claim analysis we assume $\espalt{}{e^{\lambda B^i}} < \infty, i=1,2,...$ and  $\sum_{i=1}^\infty \log\left(\espalt{}{e^{\lambda B^i}}\right) < \infty$ for all $\lambda\in\reals$.

For $n=1,2,...$ we construct the $n^{th}$ market by restricting trading to the first $n$ assets.  Thus, as $n\uparrow\infty$ the claim is asymptotically hedgeable, though for each $n$ the market
is incomplete. As shown in \cite{Robertson_Spil_2014}, $\underline{B}_n = d^n + \essinf{\prob}{Y_n}$ and $\bar{B}_n = d^n + \esssup{\prob}{Y_n}$ where $d^n$ is the unique replicating capital for $\sum_{i=1}^n B^i$ and $Y_n\dfn \sum_{i=n+1}^\infty B^i$. Under Assumption \ref{A:GE_Opt1}, $d^n\rightarrow d = \espalt{\qprob_0}{B}$ where $\qprob_0$ is the unique martingale measure in the limiting complete market.

Since $\sum_{i=1}^\infty \log\left(\espalt{}{e^{\lambda B^i}}\right) < \infty$ for all $\lambda\in\reals$, we know that $\lim_{n\uparrow\infty} \espalt{}{Y_n^2} = 0$. Assume furthermore that $Y_n$ is converging to $0$ sufficiently fast so that it satisfies a LDP with scaling $r_n\rightarrow\infty$ and good rate function $I$ such that $\cbra{I=0} = \cbra{0}$. Lastly, assume that for some $\delta>0$, $|\lambda|<\delta$ implies
\begin{equation}\label{E:LC_1}
\limsup_{n\uparrow\infty}
\frac{1}{r_n}\sum_{i=n}^\infty\log\left(\espalt{}{e^{\lambda r_n
B^i}}\right) < \infty.
\end{equation}
For example, this will hold if $B^i\sim N(0,\delta_i^2)$, with $\sum_{i=1}^\infty \delta_i^2 < \infty$.  Fix the risk aversion $a_n = a>0$. As shown in \cite{Robertson_Spil_2014}, at $\ell = 0$ we have $\lim_{n\uparrow\infty} p^n_a(0) = d = p^\infty(0)$. Furthermore, for $0<|\ell|<\delta/a$
\begin{equation*}
\lim_{n\uparrow\infty} p^n_{a}(\ell r_n) = p^{\infty}(\ell) =  d -\frac{1}{a\ell}\sup_{y\in\reals}(-\ell a y - I(y)).
\end{equation*}
Additionally, as can be deduced from $I(y)=0 \leftrightarrow y=0$, \eqref{E:LC_1} and the lower-semicontinuity of $I$, it follows that
\begin{equation*}
\lim_{\ell\rightarrow 0} \frac{1}{a\ell}\sup_{y\in\reals}(-\ell a y - I(y)) = 0,
\end{equation*}
so that $p^\infty(\ell)\rightarrow d = p^\infty(0)$ as $\ell\rightarrow 0$.  Thus, Assumption \ref{A:GE_Opt1} holds. Lastly, it is also shown in \cite{Robertson_Spil_2014} that for all $q\in\reals$ the normalized residual risk process $\hat{Y}^n_a(q)$ of \eqref{E:norm_resid_risk_p} is precisely $Y_n$ and, as such, does not depend upon $q$.

\subsection{Black-Scholes-Merton Model with Vanishing Default Probability}\label{SS:BS_defautl}

This example is taken from \cite{Ishikawa_Robertson_2015} and the setup is similar to that considered in \cite{MR2931345}. Here, we consider
the Black-Scholes-Merton model, except that the stock may default
at the first jump time of an independent Poisson process. The
claim is a defaultable bond paying $1$ if the stock has not
defaulted by time $T$. The owner of the bond wishes to hedge the
claim by trading in $S^n$, but needs to take into account the event
of default, since the stock is stuck at $0$ after default occurs.

Fix $n$ and let $\lambda_n > 0$. For each $n$, the probability
space is assumed to support a Brownian motion $W$ as well as an
independent Poisson process $N^n$ with intensity $\lambda_n$.
Denote by $\tilde{N}^n$ the compensated Poisson process so that
$\tilde{N}^n_t = N^n_t - \lambda_n(\tau_n\wedge t)$, where $\tau_n
= \inf\cbra{t\geq 0: N^n=1}$. The filtration is that generated by
$N^n$ and $W$, augmented so that it satisfies the usual
conditions.  The (single) risky asset $S^n$ evolves according to
\begin{equation*}
\begin{split}
\frac{dS^n_t}{S^n_{t-}}&=1_{t\leq \tau_n}\left(\mu dt + \sigma dW_t\right) - dN^n_t,\\
&=1_{t\leq \tau_n}\left((\mu+\lambda_n)dt + \sigma dW_t - d\tilde{N}^n_t\right).
\end{split}
\end{equation*}
The claim is a defaultable bond which pays $1$ if $S^n$ defaults
before $T$: i.e. $B= 1_{\tau^n\leq T}$\footnote{As the claim depends upon $n$ here it does not fit precisely into the setup of Section \ref{S:model}. However, as inspection of the Propositions in Appendix \ref{SS:A_technical} shows, the results of Theorems \ref{T:opt_pos_lb}, \ref{T:opt_pos_ub} readily extend to a sequence of claims $B_n$ if they are uniformly bounded.}. Here, $\underline{B}_n =
0$ and $\bar{B}_n = 1$, this is because we can equivalently
change the default intensity to take any positive value. Thus,
Assumption \ref{A:px_range} holds even though $d=1$ and hence
$d\not\in I^n$ for all $n$.

 As shown in \cite{Ishikawa_Robertson_2015}, $u^n_a(0,q) = -\frac{1}{a}F^n(0;q)$ where $F^n(\cdot;q)$ solves the ODE
\begin{equation*}
\begin{split}
\dot{F}^n(t;q)& - \lambda F^n(t;q) - \frac{\mu^2}{2\sigma^2}F^n(t;q) + \min_{\phi}\left(\frac{1}{2}\sigma^2\phi^2F^n(t;q) + \lambda_n e^{\tfrac{\mu}{\sigma^2} - \phi}\right) = 0;\qquad t\leq T,\\
F^n(T;q) &= e^{-aq}.
\end{split}
\end{equation*}
It is easy to see that the optimal $\hat{\phi}^n$ in the above minimization satisfies $\hat{\phi}^n(t;q)e^{\hat{\phi}^n(t;q)} = \lambda_n (F^n(t;q))^{-1}e^{\tfrac{\mu}{\sigma^2}}$, where one can show that $F^n(t;q)>0$. Now, let $\lambda_n\downarrow 0$ (vanishing default probabilities) and set $r_n = -\log(\lambda_n)$. With $q_n = \ell r_n$, one can show that for $\ell <1/a$:
\begin{equation*}
\lim_{n\uparrow\infty} p^n_a(\ell r_n) = \lim_{n\uparrow\infty} -\frac{1}{\ell a r_n}\log\left(\frac{F^n(0;\ell r_n)}{F^n(0;0)}\right) = p^\infty_a(\ell)  = 1.
\end{equation*}
Since
\begin{equation*}
\lim_{\ell\rightarrow 0} p^\infty_a(\ell) = 1 = \lim_{n\uparrow\infty} p^n_a(0),
\end{equation*}
we see that Assumption \ref{A:GE_Opt1} is satisfied, though the map $\ell\mapsto \ell p^\infty(\ell) = \ell$ is not strictly concave.

\section{Vanishing Transaction Costs in the Black-Scholes-Merton Model}\label{SS:BS_trans}

In this section we show that the existence of limiting
indifference prices and the resultant statements about optimal
position taking even extend to models with frictions, where the
standard duality results used in Section \ref{S:model} are not as
fully developed (see \cite{czichowsky2014duality} for a recent
treatment of the topic).  As such, this example is given its own
section.

We consider the Black-Scholes-Merton model with proportional
transactions costs, as studied in \cite{MR1809526, MR2968041,
constantinides1986capital, MR1205985, MR3183924, MR3251862,
MR2676941, kallsen2013option, MR1284980} amongst many others.  We
take the approach of \cite{constantinides1986capital} and
especially \cite{MR1809526, MR3251862}. Using the notation of
\cite{MR1809526}, the stock $S$ evolves according to a geometric
Brownian motion
\begin{equation}\label{E:trans_GBM}
  \frac{dS_t}{S_t} = \mu dt + \sigma dW_t; \qquad t\leq T.
\end{equation}
Here, the filtered probability space is the standard one-dimensional Wiener space. Now, fix a time $t\leq T$ and $s>0$ and assume $S_t = s$. Denote by $X$ and $Y$ respectively the processes of dollar holdings in the money market and shares of stock owned associated to a trading strategy $L,M$ where $L_t=M_t=0$ and $L$ represents the cumulative transfers (in shares of stock) from the money market to the stock and $M$ represents the cumulative transfers from the stock to the money market. We denote by $\mathcal{A}_t$ the set of $(L,M)$ where $L,M$ are adapted, non-decreasing and left-continuous with $L_t = M_t = 0$. There is a proportional transaction cost $\transpm \in (0,1)$ by trading. In other words, for a given initial position $(x,y)$ where $x\in\reals$ is the initial capital and $y\in\reals$ the initial shares held in $S$ the corresponding processes evolve according to
\begin{equation}\label{E:wealth_dynamics}
\begin{split}
X_\tau &= X^{L,M,x,t}_\tau = x - \int_t^\tau S_u(1+\transpm)dL_u + \int_t^\tau S_u(1-\transpm)dM_u;\qquad t\leq \tau\leq T,\\
Y_\tau &= Y^{L,M,y,t}_\tau = y + L_\tau - M_\tau;\qquad t\leq \tau \leq T.
\end{split}
\end{equation}

The claim $B$ is a European call option on $S$: i.e. $B =
(S_T-K)^+$, and suppose that the investor is considering selling the call.
For an exponential investor with fixed risk aversion $a>0$ the
value function without the claim is given by
\begin{equation}\label{E:trans_vf_no_claim}
u_a(x,y;s,t,\transpm) = \sup_{L,M\in\mathcal{A}_t}\espaltm{}{s,t}{U_a(X_T + Y_TS_T)}.
\end{equation}
Here, $\espaltm{}{s,t}{\cdot}$ refers to conditioning on time $t$ given $S_t = s$. The value function for $q$ units of the call is
\begin{equation}\label{E:trans_vf_claim}
u_a(x,y,q;s,t,\transpm) = \sup_{L,M\in\mathcal{A}_t}\espaltm{}{s,t}{U_a(X_T + Y_TS_T - q(S_T-K)^+)}.
\end{equation}
The indifference price $p_a(x,y,q;s,t,\transpm)$ is then defined through the balance equation
\begin{equation}\label{E:trans_px}
u_a(x+qp_a(x,y,q;s,t,\transpm),y,q;s,t,\transpm) = u_a(x,y;s,t,\transpm).
\end{equation}

\begin{remark}\label{R:ask_price} $p_a(x,y,q;s,t,\transpm)$ is thus the average \emph{ask} indifference price, as opposed to the average \emph{bid} indifference price defined in Section \ref{S:model}.  However, using the arguments of Section \ref{S:model} and definition \eqref{E:util_indiff} for a general claim $B$, the bid and ask prices are related by $p^{\textrm{ask}}_a(q;B)=-p^{\textrm{bid}}_a(q;-B)$, where $p^{\textrm{bid}}_a(q;B)$ denotes the average bid price $(1/q)\pr^{\textrm{bid}}_{a}(qB)$.
\end{remark}

Though the results in \cite{MR1809526} are stated in the joint
limit of vanishing transactions costs (i.e.~$\transpm_n\rightarrow
0$) and infinite risk aversion (i.e. $a=a_n\rightarrow\infty$),
they easily (as the authors therein mention) translate into
asymptotics in the joint limit that $\transpm_n\rightarrow 0$ and
$q=q_n\rightarrow\infty$ for a fixed risk aversion $a$.  This
translation is made precise in the following proposition.

\begin{proposition}\label{prop:bs_conv_price}

Fix $s>0, 0 \leq t\leq T$, $x\in\reals$, $y\in\reals$, $\transpm\in(0,1)$ and $a>0$. The (ask) indifference price $p_a$  is independent of $x$ and hence write $p_a = p_a(y,q;s,t,\transpm)$. Now, let $\transpm_n\rightarrow 0$ and set $r_n \dfn \transpm_n^{-2}$. For $\ell > 0$ and $q_n = \ell r_n = \ell\transpm_n^{-2}$ we have for all $y_n$ such that $\lim_{n\uparrow\infty} \transpm_n^3|y_n| = 0$:
\begin{equation*}
\lim_{n\uparrow\infty} p_a(y_n,q_n;s,t,\transpm_n) = p^\infty_a(\ell;s,t) \dfn \Psi(s,t;\sqrt{a\ell}),
\end{equation*}
where for $b>0$, $\Psi(;b) :(0,\infty)\times [0,T]\mapsto\reals$ is the unique continuous viscosity solution to the non-linear Black-Scholes PDE
\begin{equation}\label{Eq:NonlinearBlackScholesEquation}
\begin{split}
\Psi_t + \frac{1}{2}\sigma^2 s^2 \Psi_{ss}\left(1+S(b^2 s^2\Psi_{ss})\right) = 0&;\qquad (s,t)\in (0,\infty)\times (0,T);\\
\Psi(s,T) = (s-K)^+&;\qquad s\in (0,\infty);\\
\lim_{s\uparrow\infty}\frac{\Psi(s,t)}{s} = 1&;\qquad t\leq T\textrm{  uniformly in } t.
\end{split}
\end{equation}
Here, $S:\reals\mapsto (-1,\infty)$ satisfies
\begin{equation*}
\dot{S}(A) = \frac{1+S(A)}{2\sqrt{AS(A)}-A};\quad S(0) = 0;\quad \lim_{A\downarrow-\infty}S(A)=-1;\quad \lim_{A\uparrow\infty} S(A)/A = 1.
\end{equation*}
\end{proposition}

\begin{remark} The above result allows for $y_n$ to vary since intuitively a position size of $q_n$ in the call would be associated to an initial position of $q_ny$ in the stock for some $y\in\reals$.  Note that for $y_n = q_n y = \ell y \transpm_n^{-2}$ we have $\transpm_n^3|y_n|\rightarrow 0$.
\end{remark}

To obtain the optimal position taking results analogous to
Theorems \ref{T:opt_pos_lb}, \ref{T:opt_pos_ub}, it is first
necessary to identify the range of limiting prices
$p^\infty_a(\ell;s,t)$ in Proposition \ref{prop:bs_conv_price} as
$\ell$ varies between $0$ and $\infty$. In other words, we must
consider asymptotics for $\Psi(;b)$ for small and large $b$.

As $b\downarrow 0$, Theorem \ref{T:VanishingTransactionsCost} below proves continuity in that $\Psi(s,t;b)\rightarrow \Psi(s,t;0)$. But, for $b=0$, \eqref{Eq:NonlinearBlackScholesEquation} is just the regular Black-Scholes PDE which admits a unique (explicit) classical solution. Thus, as $\ell\downarrow 0$, the limiting indifference price converges to the unique price in complete, $\transpm_n = 0$ market given $S_t =s$.

\begin{theorem}\label{T:VanishingTransactionsCost}
Let $\Psi(;b):(0,\infty)\times [0,T]\mapsto\reals$ be the unique,
continuous, viscosity solution to the non-linear Black-Scholes PDE
equation (\ref{Eq:NonlinearBlackScholesEquation}). Then as
$b\rightarrow 0$, we have locally uniformly that
$\Psi(;b)\rightarrow\Psi(;0)$, where $\Psi(;0)$ is the unique
continuous solution to the linear Black-Scholes PDE.
\end{theorem}

Next, we identify the limit of $\Psi(;b)$ as $b\uparrow\infty$. Here, we are guided by the intuition that, thought of as a function of the stock volatility, the Black-Scholes price for a call option converges to the initial price as the volatility becomes large.  In fact, a similar phenomenon occurs here as $b\uparrow\infty$, as the following shows:

\begin{theorem}\label{T:Psi_a_large}

For fixed $s>0,0\leq t\leq T$ the map $b\mapsto\Psi(s,t;b)$ is increasing with
\begin{equation}\label{E:Psi_a_large}
\lim_{b\uparrow\infty} \Psi(s,t;b) = \begin{cases} (s-K)^+ & t = T\\ s & 0\leq t < T\end{cases}.
\end{equation}
\end{theorem}

\begin{remark} An inspection of the proof of Theorem \ref{T:VanishingTransactionsCost} below shows that $\Psi(s,t;b)$ is continuously increasing in $b$. Thus, if $q_n = \ell_n r_n$ where $\ell_n \rightarrow \ell\geq 0$ then the indifference prices converge to $\Psi(s,t;\sqrt{a\ell})$.
\end{remark}

With the above asymptotics for $p^\infty_a(\ell;s,t)$ in place, we now consider the optimal sale quantity problem in the $n^{th}$ market with transactions cost $\transpm_n$. In order to simplify the presentation, we assume that given $S_t = s$ the investor has the opportunity to sell call options at a price $\tilde{p}^n$ in the $n^{th}$ market. To finance this sale, the investor cashes out her initial position in the stock, receiving $ys(1-\transpm_n)$ for the sale of $y$ shares.   Then, with $x + ys(1-\transpm_n)$ in cash, she identifies the optimal number of options to sell by solving the problem
\begin{equation}\label{E:trans_opt_q}
\sup_{q> 0} u_a(x+ys(1-\transpm_n) + q\tilde{p}^n,0,q;s,t,\transpm_n).
\end{equation}

In the frictionless case, if $\tilde{p}^n$ is arbitrage free in
the $n^{th}$ market, then (see \cite{MR2212897}), an optimal
$\hat{q}_n$ exists and is unique. When considering transactions
costs, rather than identifying the arbitrage free prices in each
market, we use the small and large $\ell$ asymptotics for
$p^\infty_a(\ell;s,t)$ obtained in Theorems
\ref{T:VanishingTransactionsCost}, \ref{T:Psi_a_large} to identify
a maximal range of reasonable prices $\tilde{p}^n$ for which one
can sell the option. Indeed, from the above theorems
\begin{equation*}
\lim_{\ell\downarrow 0}p^\infty_a(\ell;s,t) = \Psi(s,t;0);\qquad \lim_{\ell\uparrow\infty}p^\infty_a(\ell;s,t) = s.
\end{equation*}
It is well known that $\Psi(s,t;0) < s$.  Furthermore, if one is
going to sell options, the effect of the transactions costs is
that the ask price should a) be at least as large as $\Psi(s,t;0)$
and b) be no higher than $p$ since no-one would buy at this
price\footnote{Technically: no one would buy at a price at or
above $p(1+\transpm_n)$ because it would then be preferable to buy
the stock and not trade.  For this to hold as
$\transpm_n\downarrow 0$, we require $\tilde{p}^n\leq p$. Our
results are valid for $\tilde{p}^n < p$.}.  Thus, the only range
of reasonable prices to sell at is $(\Psi(s,t;0),s)$.  With this
motivation we have:

\begin{theorem}\label{T:opt_quant_trans}
Let $\tilde{p}^n \in (\Psi(s,t;0), s)$ for each $n$ with $\tilde{p}^n\rightarrow \tilde{p}$ where $\tilde{p}\in (\Psi(s,t;0),s)$. Let $\transpm_n\rightarrow 0$.  For each $n$ there exists a maximizer $\hat{q}_n >  0$ to \eqref{E:trans_opt_q}. Additionally, for any sequence $\cbra{\hat{q}_n}_{n\in\mathbb{N}}$ of maximizers:
\begin{equation}\label{E:trans_rates}
\liminf_{n\uparrow\infty} \frac{\hat{q}_n}{r_n} > 0;\qquad \limsup_{n\uparrow\infty} \frac{\hat{q}_n}{r_n} < \infty.
\end{equation}
Thus, up to subsequences, $\hat{q}_n/r_n \rightarrow \ell$ and hence for any sequence $\cbra{y_n}_{n\in\mathbb{N}}$ such that $\lambda_n^3|y_n|\rightarrow 0$:
\begin{equation*}
\lim_{n\uparrow\infty} p_a(y_n,\hat{q}_n;s,t,\transpm_n) = p^\infty_a(\ell;s,t) = \Psi(s,t;\sqrt{a\ell}).
\end{equation*}
\end{theorem}

\bigskip

\appendix

\section{Technical Supporting Results}\label{SS:A_technical}

The following propositions provide the main technical tools to prove the optimal position taking results in both the frictionless and transactions cost cases. To seamlessly integrate with the transaction costs case, results are separated into long and short positions.

\subsection{Long Positions}

Assume:
\begin{assumption}\label{A:one_p_pos}
$\cbra{p^n}$ is a family of functions defined on $(0,\infty)$ such that
\begin{itemize}
\item For each $n$, $p^n$ is non-increasing and continuous.
\item There exists a $\gamma>0$ such that $\limsup_{n\uparrow\infty}\sup_{q\leq \gamma} q|p^n(q)| = C(\gamma) < \infty$.
\item There exists $r_n\rightarrow\infty$ and $\delta>0$ such that for $0<\ell < \delta$ we have $\lim_{n\uparrow\infty}p^n(\ell r_n) = p^\infty(\ell)$.
\item With $p^\infty_+(0)\dfn \lim_{\ell\downarrow 0}p^\infty(\ell)$ and $p^n(\infty)\dfn \lim_{q\uparrow\infty} p^n(q)$ we have $\limsup_{n\uparrow\infty} p^n(\infty) < p^\infty_{+}(0)$.
\end{itemize}
\end{assumption}

To find the maximal upper bound of convergence, set
\begin{equation}\label{E:delta_plus}
\begin{split}
\delta^+ &\dfn\sup\cbra{k>0\ |\ \lim_{n\uparrow\infty}p^n(\ell r_n) = p^\infty(\ell),\ \forall\ 0\leq \ell < k}\in [\delta,\infty].
\end{split}
\end{equation}
Note that for $0 < \ell < \delta^+$ we have $p^n(\infty)\leq p^n(\ell r_n)$ so that $\limsup_{n\uparrow\infty} p^n(\infty) \leq p^\infty(\ell) \leq p^\infty_{+}(0)$. As such, a sufficient condition for bullet point four in Assumption \ref{A:one_p_pos} to hold is that $p^\infty(\ell) < p^\infty_{+}(0)$ for some $0<\ell < \delta^+$.

Under Assumption \ref{A:one_p_pos} we have the following result for positive position sizes:

\begin{proposition}\label{prop: one_p_pos}

Let Assumption \ref{A:one_p_pos} hold. Let $\tilde{p}^n\rightarrow \tilde{p}$.
\begin{itemize}
\item If $\limsup_{n\uparrow\infty}p^n(\infty) < \tilde{p} < p^\infty_+(0)$ then for $n$ large enough the optimization problem
\begin{equation}\label{E:pos_opt_prob}
\inf_{q>0}\left(q\tilde{p}^n - qp^n(q)\right),
\end{equation}
admits a minimizer $\hat{q}_n>0$.
\item If $\limsup_{n\uparrow\infty} p^n(\infty) < \tilde{p} < p^\infty_{+}(0)$ then for any sequence of minimizers $\cbra{\hat{q}_n}$:
\begin{equation}\label{E:pos_opt_lim_lb}
0 < \liminf_{n\uparrow\infty}\frac{\hat{q}_n}{r_n}.
\end{equation}
\item If additionally $\lim_{\ell\uparrow\delta^+} p^\infty(\ell) < \tilde{p} < p^\infty_+(0)$ then for any sequence $\cbra{\hat{q}_n}$ of minimizers:
\begin{equation}\label{E:pos_opt_lim_ub}
\limsup_{n\uparrow\infty}\frac{\hat{q}_n}{r_n} < \delta^+.
\end{equation}
\end{itemize}
\end{proposition}

\begin{proof}[Proof of Proposition \ref{prop: one_p_pos}]

First consider the minimization problem in \eqref{E:pos_opt_prob}.  Since $\tilde{p}^n\rightarrow \tilde{p}$ there is some $\eps>0$ and $N_{\eps}$ so that $n\geq N_{\eps}$ implies $\limsup_{n\uparrow\infty} p^n(\infty) +\eps < \tilde{p}^n < p^\infty_{+}(0) - \eps$.  Next, choose $\ell>0$ small enough so that $\tilde{p}^n < p^\infty(\ell) - \eps/2$. By enlarging $N_\eps$ we know for $n\geq N_\eps$ that $p^n(\infty)\leq \limsup_{n\uparrow\infty}p^n(\infty)+\eps/2$ and $p^\infty(\ell) < p^n(\ell r_n) + \eps/4$ and hence
\begin{equation}\label{E:temp_px_ellp}
p^n(\infty) + \eps/2 \leq \tilde{p}^n \leq p^n(\ell r_n) - \eps/4.
\end{equation}
For a fixed $n$, note that
$\lim_{q\uparrow\infty}(\tilde{p}^n-p^n(q)) =
\tilde{p}^n-p^n(\infty)\geq \eps/2$. Thus, if
$\cbra{\hat{q}^m_n}_{m\in\N}$ is a minimizing sequence for
\eqref{E:pos_opt_prob}, then $\cbra{\hat{q}^m_n}$ is bounded and
hence has an accumulation point $\hat{q}_n$.  We now show that
$\hat{q}_n\neq 0$, which combined with the continuity of $qp^n(q)$
proves $\hat{q}_n>0$ is a minimizer. To see that $\hat{q}_n\neq 0$  we use a contradiction argument. Note
 that with the $\gamma$ from Assumption \ref{A:one_p_pos}:
\begin{equation*}
\liminf_{q\downarrow 0} (q\tilde{p}^n - qp^n(q)) = -\limsup_{q\downarrow 0} qp^n(q) \geq -\sup_{q\leq\gamma}q|p^n(q)|.
\end{equation*}
For the given $\eps$, by enlarging $N_\eps$ we may assume that for $n\geq N_\eps$
\begin{equation*}
\liminf_{q\downarrow 0} (q\tilde{p}^n - qp^n(q)) \geq -\limsup_{n\uparrow\infty}\sup_{q\leq\gamma}q|p^n(q)| - \eps = -C(\gamma) - \eps.
\end{equation*}
But, for the $\ell$ from \eqref{E:temp_px_ellp}:
\begin{equation}\label{E:temp_lb_ellp}
\ell r_n\tilde{p}^n - \ell r_np^n(\ell r_n) \leq -\ell r_n\eps/4.
\end{equation}
Combining the last two displays we get that for the chosen $n$, we have
\begin{equation*}
-\ell r_n \eps/4 \geq -C(\gamma)-\eps.
\end{equation*}
However, by potentially enlarging $N_\eps$, and since $r_{n}\rightarrow\infty$, we can always arrange things so that $-\ell r_n \eps/4 <
-C(\gamma)-\eps$. This leads to a contradiction, proving that $\hat{q}_n\neq 0$.

Now, let $\cbra{\hat{q}_n}$ be a sequence of minimizers. We first claim that $\liminf_{n\uparrow\infty}\hat{q}_n>0$. Indeed, assume there is a subsequence (still labeled $n$) so that $\lim_{n\uparrow\infty} \hat{q}_n = 0$. We then have, using the $\gamma$ of Assumption \ref{A:one_p_pos} that
\begin{equation*}
\begin{split}
\liminf_{n\uparrow\infty}(\hat{q}_n\tilde{p}^n -
\hat{q}_np^n(\hat{q}_n)) & =
-\limsup_{n\uparrow\infty}\hat{q}_np^n(\hat{q}_n) \geq
-\limsup_{n\uparrow\infty} \sup_{q\leq \gamma} q|p^n(q)| =
-C(\gamma).
\end{split}
\end{equation*}
But, this directly violates the minimality of $\hat{q}_n$ in view
of \eqref{E:temp_lb_ellp}. As such, there is some $K>0$ so that
$\hat{q}_n\geq K$ for $n$ large enough.

Now, assume that $\liminf_{n\uparrow\infty} \hat{q}_n/r_n  = 0$ and take a subsequence such that $\lim_{n\uparrow\infty}\hat{q}_n/r_n = 0$. For all $0 < c < \delta^+$ we see
\begin{equation}\label{E:the_lb_val}
\begin{split}
\tilde{p}^n - p^n(cr_n) \geq \frac{\hat{q}_n}{c r_n}\left(\tilde{p}^n - p^n(\hat{q}_n)\right).
\end{split}
\end{equation}
As $n\uparrow\infty$ we know that $\tilde{p}^n-p^n(cr_n)\rightarrow \tilde{p} - p^\infty(c)$, $\hat{q}_n/(cr_n)\rightarrow 0$ and $\tilde{p}^n\rightarrow \tilde{p}$.  Recall that $\liminf_{n\uparrow\infty}\hat{q}_n\geq K$ and the $\gamma$ from Assumption \ref{A:one_p_pos}. Note that if $K>\gamma$ then
\begin{equation*}
p^n(K) \leq p^n(\gamma) = \frac{1}{\gamma}\gamma p^n(\gamma) \leq \frac{1}{\gamma}\sup_{q\leq\gamma} q|p^n(q)|,
\end{equation*}
whereas if $K\leq \gamma$ then
\begin{equation*}
p^n(K) = \frac{1}{K}Kp^n(K) \leq \frac{1}{K}\sup_{q\leq\gamma} q|p^n(q)|.
\end{equation*}
Putting these together gives
\begin{equation*}
\limsup_{n\uparrow\infty} p^n(\hat{q}_n) \leq \frac{1}{\gamma\wedge K}\limsup_{n\uparrow\infty}\sup_{q\leq\gamma}q|p^n(q)| = \frac{C(\gamma)}{\gamma\wedge K}.
\end{equation*}
Thus, taking $n\uparrow\infty$ in \eqref{E:the_lb_val} gives $\tilde{p}\geq p^\infty(c)$.  Taking $c\downarrow 0$ gives $\tilde{p}\geq p^\infty_+(0)$ a contradiction.  Therefore, \eqref{E:pos_opt_lim_lb} holds.

Next, assume that $\limsup_{n\uparrow\infty}\hat{q}_n/r_n \geq \delta+$ and take a subsequence so that $\lim_{n\uparrow\infty} \hat{q}_n/r_n = k \geq \delta^+$. For each $c < \delta^+$ we have $\hat{q}_n/r_n \geq c$ and hence for any $K>0$, $\hat{q}_n\geq K$ for $n$ large enough. Thus, we have
\begin{equation}\label{E:the_ub_val}
K\tilde{p}^n - Kp^n(K) \geq \hat{q}_n\left(\tilde{p}^n-p^n(\hat{q}_n)\right) \geq \hat{q}_n\left(\tilde{p}^n-p^n(cr_n)\right).
\end{equation}
Clearly, $K\tilde{p}^n/\hat{q}_n\rightarrow 0$.  Additionally, for any $0 < c' < \delta^+$:
\begin{equation*}
\liminf_{n\uparrow\infty}\frac{p^n(K)}{\hat{q}_n} \geq \liminf_{n\uparrow\infty} \frac{p^n(c'r_n)}{\hat{q}_n} = 0.
\end{equation*}
Thus, dividing by $\hat{q}_n$ in \eqref{E:the_ub_val} and taking $n\uparrow\infty$ yields $0\geq \tilde{p}-p^\infty(c)$. Taking $c\uparrow\delta^+$ gives that $\tilde{p}\leq \lim_{c\uparrow\delta^+}p^\infty(c)$, which is a contradiction. Therefore, \eqref{E:pos_opt_lim_ub} holds.

\end{proof}

\subsection{Short Positions}

We just state the result for $q<0$ as the proof is the exact same. First, we assume:

\begin{assumption}\label{A:one_p_neg}
$\cbra{p^n}$ is a family of functions defined on $(-\infty,0)$ such that
\begin{itemize}
\item For each $n$, $p^n$ is non-increasing and continuous.
\item There exists a $\gamma < 0$ such that $\limsup_{n\uparrow\infty}\sup_{q\geq \gamma } q|p^n(q)| = C(\gamma) < \infty$.
\item There exists $r_n\rightarrow\infty$ and $\delta>0$ such that for $-\delta < \ell < 0$ we have $\lim_{n\uparrow\infty}p^n(\ell r_n) = p^\infty(\ell)$.
\item With $p^\infty_{-}(0)\dfn \lim_{\ell\uparrow 0} p^\infty(\ell)$ and $p^n(-\infty) \dfn \lim_{q\downarrow -\infty} p^n(q)$ we have $p^\infty_{-}(0) < \liminf_{n\uparrow\infty} p^n(-\infty)$.
\end{itemize}
\end{assumption}

To find the minimal lower bound of convergence, set
\begin{equation}\label{E:delta_minus}
\begin{split}
\delta_{-} &\dfn\inf\cbra{k<0\ |\ \lim_{n\uparrow\infty}p^n(\ell r_n) = p^\infty(\ell),\ \forall 0\geq \ell > k}\leq \in [-\infty,\delta_{-}].\\
\end{split}
\end{equation}

As before, we have for any $\delta_{-} < \ell < 0$ that $p^\infty_{-}(0) \leq  p^\infty(\ell) \leq \liminf_{n\uparrow\infty} p^n(-\infty)$ so that a sufficient condition for bullet point four above to hold is that $p^\infty_{-}(0) < p^\infty(\ell)$ for some $\delta_{-} < \ell < 0$. The main result now reads:

\begin{proposition}\label{prop: one_p_neg}
Let Assumption \ref{A:one_p_neg} hold. Let $\tilde{p}^n\rightarrow \tilde{p}$.
\begin{itemize}
\item If $p^\infty_{-}(0) < \tilde{p} < \liminf_{n\uparrow\infty}p^n(-\infty)$ then for $n$ large enough the optimization problem
\begin{equation}\label{E:neg_opt_prob}
\inf_{q<0}\left(q\tilde{p}^n - qp^n(q)\right),
\end{equation}
admits a minimizer $\hat{q}_n<0$.
\item If $p^\infty_{-}(0) < \tilde{p} < \liminf_{n\uparrow\infty} p^n(-\infty)$ then for any sequence of minimizers $\cbra{\hat{q}_n}$:
\begin{equation}\label{E:neg_opt_lim_lb}
0 < \liminf_{n\uparrow\infty}\frac{-\hat{q}_n}{r_n}.
\end{equation}
\item If additionally $p^\infty_{-}(0) < \tilde{p} < \lim_{\ell\downarrow\delta_{-}} p^\infty(\ell)$ then for any sequence $\cbra{\hat{q}_n}$ of minimizers:
\begin{equation}\label{E:neg_opt_lim_ub}
\limsup_{n\uparrow\infty}\frac{-\hat{q}_n}{r_n} < -\delta_{-}.
\end{equation}
\end{itemize}
\end{proposition}

\subsection{Long and Short Positions}

We now combine the long and short results of the previous section into one result which will be used to prove the frictionless results of Section \ref{S:consequences}.  Here, we assume

\begin{assumption}\label{A:one_p}
$\cbra{p^n}_{n\in\mathbb{N}}$ is a sequence of functions on $\reals$ such that
\begin{itemize}
\item For each $n$, $p^n$ is non-increasing and continuous.
\item There exists a $\gamma > 0$ such that $\limsup_{n\uparrow\infty}\sup_{|q|\leq\gamma} q|p^n(q)| = C(\gamma) < \infty$.
\item There exists $r_n\rightarrow \infty$ and $\delta > 0$ such that for $|\ell|<\delta$ we have $p^n(\ell r_n) \rightarrow p^\infty(\ell)$.
\item $\lim_{\ell\rightarrow 0}p^\infty(\ell) = p^\infty(0)$.
\end{itemize}
\end{assumption}

\begin{proposition}\label{prop: one_p}
Let Assumption \ref{A:one_p} hold and define $\delta^+,\delta_-$ as in \eqref{E:delta_plus} and \eqref{E:delta_minus}.  Let $\tilde{p}^n\rightarrow \tilde{p}$.
\begin{itemize}
\item Assume that $\limsup_{n\uparrow\infty} p^n(\infty) < p^{\infty}(0)$. If $\limsup_{n\uparrow\infty}p^n(\infty) < \tilde{p} < p^\infty(0)$ then for $n$ large enough any minimizer to the optimization problem $\inf_{q\in\reals}\left(q\tilde{p} - qp^n(q)\right)$ is positive. Furthermore, for any sequence of minimizers $\cbra{\hat{q}_n}_{n\in\mathbb{N}}$ we have that $0 < \liminf_{n\uparrow\infty} \hat{q}_n/r_n$. If additionally $\lim_{\ell\uparrow\delta^+} p^\infty(\ell) < \tilde{p} < p^\infty(0)$ then for any sequence of minimizers $\cbra{\hat{q}_n}_{n\in\mathbb{N}}$ we have that $\limsup_{n\uparrow\infty}\hat{q}_n/r_n < \delta^+$.
\item Assume that $p^\infty(0) < \liminf_{n\uparrow\infty} p^n(-\infty)$. If $p^\infty(0) < \tilde{p} < \liminf_{n\uparrow\infty} p^n(-\infty)$ then for $n$ large enough, any minimizer to the optimization problem $\inf_{q\in\reals}\left(q\tilde{p}^n - qp^n(q)\right)$ is negative. Furthermore, for any sequence of minimizers $\cbra{\hat{q}_n}_{n\in\mathbb{N}}$ we have that $0 < \liminf_{n\uparrow\infty} -\hat{q}_n/r_n$. If additionally $p^\infty(0) < \tilde{p} < \lim_{\ell\downarrow\delta_{-}} p^\infty(\ell)$ then for any sequence of minimizers $\cbra{\hat{q}_n}$ we have that $\limsup_{n\uparrow\infty} -\hat{q}_n/r_n < -\delta_{-}$.
\end{itemize}
\end{proposition}

\begin{proof}[Proof of Proposition \ref{prop: one_p}]

We will prove the results for $\limsup_{n\uparrow\infty} p^n(\infty) < \tilde{p} < p^\infty(0)$ and  $\lim_{\ell\uparrow\infty}p^\infty(\ell) < \tilde{p} < p^\infty(0)$ respectively; the proof for the other case is the exact same. First, since $p^n(0)$ is well defined for each $n$, we have $0\times \tilde{p}^n - 0\times p^n(0) = 0$. Additionally, for $\eps >0$ so that $\limsup_{n\uparrow\infty} p^n(\infty) +\eps < \tilde{p} < p^\infty(0)-\eps$ we have for $q<0$ and $n$ large enough that
\begin{equation*}
\begin{split}
q\tilde{p} - qp^n(q) \geq q\tilde{p} - qp^n(0) \geq -q\eps/2 > 0,
\end{split}
\end{equation*}
But, from \eqref{E:temp_lb_ellp} we see there is some $\ell>0$ so that $\ell r_n\tilde{p}^n - \ell r_np^n(\ell r_n) < 0$. Thus it suffices to minimize over $q>0$ and hence Proposition \ref{prop: one_p_pos} yields a minimizer to the problem over $(0,\infty)$, as well as the asymptotic behavior $\hat{q}_n/r_n$ of minimizers $\hat{q}_n$ given above, finishing the result.

\end{proof}

\section{Proofs for Section \ref{SS:opt_pos}}\label{S:pf_opt_pos}

The proofs of Theorems \ref{T:opt_pos_lb} and \ref{T:opt_pos_ub} are
based on a more general result that we proved in Appendix
\ref{SS:A_technical}. Hence,  as a precursor to the proofs of
Theorem \ref{T:opt_pos_lb} and \ref{T:opt_pos_ub} we first show
that the functions $p^n(q)\dfn p^n_{a_n}(q)$ satisfy Assumption
\ref{A:one_p} above.

\begin{lemma}\label{L:one_p_holds}

Let Assumptions \ref{A:claim}, \ref{A:no_arb_n}, \ref{A:GE_Opt1}
and \ref{A:px_range} hold.  Then, $p^n(q)\dfn p^n_{a_n}(q)$
satisfies Assumption \ref{A:one_p}.

\end{lemma}

\begin{proof}[Proof of Lemma \ref{L:one_p_holds}]

As shown in Section \ref{SS:discussion}, $p^n_{a_n}(q)$ is
decreasing in $q$ and the map $q\mapsto qp^n_{a_n}(q)$ is concave
and well defined, finite, for all $q\in\reals$.  As such,
$p^n_{a_n}(q)$ is continuous on $(-\infty,0)$ and $(0,\infty)$
respectively.  But, it is well known that continuity at $0$
follows as well and in fact $\lim_{q\rightarrow 0} p^n_{a_n}(q) =
\espalt{\qprob^n_0}{B} = p^n_{a_n}(0) = d_n$.  Thus, bullet point
one in Assumption \ref{A:one_p} holds. Regarding bullet point two,
let $\gamma > 0$.  If $0<q\leq \gamma$ then for any $0<\ell <
\delta^+$ and $n$ sufficiently large so that $r_n \geq
\ell/\gamma$:
\begin{equation*}
\begin{split}
p^n_{a_n}(q) &\leq p^n_{a_n}(0) = d_n =
\espalt{\qprob^n_0}{B};\qquad p^n_{a_n}(q) \geq p^n_{a_n}(\ell
r_n).
\end{split}
\end{equation*}
If $-\gamma\leq q < 0$ then for any $\delta_{-} < \ell' < 0$ and $n$ so that $r_n\geq -\ell'/\gamma$:
\begin{equation*}
\begin{split}
p^n_{a_n}(q) &\geq p^n_{a_n}(0) = d_n
\espalt{\qprob^n_0}{B};\qquad p^n_{a_n}(q) \leq p^n_{a_n}(\ell'
r_n).
\end{split}
\end{equation*}
As such:
\begin{equation*}
\limsup_{n\uparrow\infty}\sup_{|q|\leq\gamma}q|p^n_{a_n}(q)| \leq \gamma\max\cbra{|d|, |p^\infty(\ell)|, |p^\infty(\ell')|} = C(\gamma),
\end{equation*}
and bullet point two holds.  Bullet points three and four are Assumption \ref{A:GE_Opt1}, finishing the result.
\end{proof}

\begin{proof}[Proof of Theorem \ref{T:opt_pos_lb}]

For $\tilde{p}^{n}\in I^n$, the optimal position
$\hat{q}_n(\tilde{p}^{n})$ is the unique solution of the problem
\eqref{E:opt_q_n}.  Using the explicit formula for $U_{a_n}$ in
\eqref{E:util_funct} and $p^n_{a_n}$ in \eqref{E:indiff_px}, this
optimization problem is equivalent to finding
\begin{equation}\label{E:optimization}
\hat{q}_n(\tilde{p}^{n})\in\text{argmin}_{q\in\mathbb{R}}\left(q\tilde{p}^{n}-q
p^{n}_{a_n}(q)\right).
\end{equation}
The results of the theorem will follow from Proposition \ref{prop:
one_p} once the requisite hypotheses are met where $p^n(q) =
p^n_{a_n}(q)$.  By Lemma \ref{L:one_p_holds}, Assumption
\ref{A:one_p} holds. Now, let $\tilde{p}^n\in I^n$,
$\tilde{p}^n\rightarrow \tilde{p}$ where $\tilde{p}$ and
$\tilde{p} < d$. Since $p^n(\infty)\leq \tilde{p}^n$ and
$d=p^\infty(0)$ we have
\begin{equation*}
\limsup_{n\uparrow\infty} p^n(\infty) =
\limsup_{n\uparrow\infty}\underline{B}_n \leq
\lim_{n\uparrow\infty} \tilde{p}^n = \tilde{p} < d = p^\infty(0).
\end{equation*}
Thus, the conclusions of the theorem follow from Proposition
\ref{prop: one_p}. Similarly let $\tilde{p}^n\in I^n$,
$\tilde{p}^n\rightarrow \tilde{p}$ where $\tilde{p}$ and
$\tilde{p} > d$. Since $p^n(-\infty)\geq \tilde{p}^n$ and
$d=p^\infty(0)$ we have
\begin{equation*}
\liminf_{n\uparrow\infty} p^n(-\infty) =
\liminf_{n\uparrow\infty}\bar{B}_n \geq \lim_{n\uparrow\infty}
\tilde{p}^n = \tilde{p} > d = p^\infty(0).
\end{equation*}
Thus, the conclusions of the theorem follow from Proposition \ref{prop: one_p} as well, finishing the result.

\end{proof}

\begin{proof}[Proof of Theorem \ref{T:opt_pos_ub}]

As in the proof of Theorem \ref{T:opt_pos_lb}, it is enough to
show that requisite hypotheses of Proposition \ref{prop: one_p}
are met where $p^n(q) = p^n_{a_n}(q)$ and the optimal position
$\hat{q}_n(\tilde{p}^{n})$ is given in \eqref{E:optimization}.
Again by Lemma \ref{L:one_p_holds}, we have that Assumption
\ref{A:one_p} holds. Now, let $\tilde{p}^n\in I^n$,
$\tilde{p}^n\rightarrow \tilde{p}$ where $\tilde{p}$ and
$p^\infty(\delta^+) < \tilde{p} < d$. Since
$p^n(\infty)\leq \tilde{p}^n$ and $d=p^\infty(0)$ we have
\begin{equation*}
\limsup_{n\uparrow\infty} p^n(\infty) =
\limsup_{n\uparrow\infty}\underline{B}_n \leq
\lim_{n\uparrow\infty} \tilde{p}^n = \tilde{p} < d = p^\infty(0).
\end{equation*}
Thus, the conclusions of the theorem follow from Proposition
\ref{prop: one_p}. Similarly let $\tilde{p}^n\in I^n$,
$\tilde{p}^n\rightarrow \tilde{p}$ where $\tilde{p}$ and
$p^\infty(\delta_{-}) > \tilde{p} > d$.
Since $p^n(-\infty)\geq \tilde{p}^n$ and $d=p^\infty(0)$ we have
\begin{equation*}
\liminf_{n\uparrow\infty} p^n(-\infty) =
\liminf_{n\uparrow\infty}\bar{B}_n \geq \lim_{n\uparrow\infty}
\tilde{p}^n = \tilde{p} > d = p^\infty(0).
\end{equation*}
Thus, the conclusions of the theorem follow from Proposition \ref{prop: one_p} as well, finishing the result.
\end{proof}

\begin{proof}[Proof of Corollary \ref{C:opt_pos}]
Let, for example, $\tilde{p}^n\rightarrow \tilde{p}\in(p^\infty(\delta^+),d)$
so that
\begin{equation*}
0 < \underline{\ell} = \liminf_{n\uparrow\infty}\frac{\hat{q}_n(\tilde{p})}{r_n} \leq \limsup_{n\uparrow\infty} \frac{\hat{q}_n(\tilde{p}^n)}{r_n} = \bar{\ell} < \delta^+.
\end{equation*}
Write $\hat{q}_n$ for $\hat{q}_n(\tilde{p}^n)$ and assume for some subsequence (still labeled $n$) that $\hat{q}_n/r_n\rightarrow \ell \in [\underline{\ell},\bar{\ell}]$. Let $\tau\in [\underline{\ell},\bar{\ell}]$. By the optimality of $\hat{q}_n$
\begin{equation*}
\hat{q}_n\tilde{p}^n - \hat{q}_n p^n_{a_n}(\hat{q}_n) \leq \tau r_n \tilde{p}^n - \tau r_n p^n_{a_n}(\tau r_n).
\end{equation*}
Dividing by $r_n$, letting $n\uparrow\infty$ and using Assumption \ref{A:GE_Opt1} with \eqref{E:continuity} one obtains
\begin{equation*}
\ell \tilde{p} - \ell p^\infty(\ell) \leq \tau\tilde{p} - \tau p^{\infty}(\tau).
\end{equation*}
Since this works for all $\tau\in [\underline{\ell},\bar{\ell}]$, we get that
\begin{equation*}
\ell \tilde{p} - \ell p^\infty(\ell) \leq \inf_{\tau\in[\underline{\ell},\bar{\ell}]}\left(\tau\tilde{p} - \tau p^{\infty}(\tau)\right).
\end{equation*}
Hence, we see that the only possible limit points for $\hat{q}_n/r_n$ are the minimizers of $\ell\tilde{p} - \ell p^\infty(\ell)$ over $[\underline{\ell},\bar{\ell}]$.  But, under the assumption of strict concavity for $\ell p^\infty(\ell)$ any minimizer is unique and hence the result follows.
\end{proof}


\begin{proof}[Proof of Theorem \ref{T:opt_pos_gen_u}]

We start be proving the first bullet, i.e., that we show that maximizers exist to the optimal purchase quantity problem in \eqref{E:opt_q_gen}. To do so we use the following basic result (see \cite[Proposition 2.47]{MR2169807}): if $U\in\Ua$ then with $\underline{\alpha}_U$, $\bar{\alpha}_U$ of \eqref{E:u_bdd_ra} it
holds for $U_a$ from \eqref{E:util_funct} with $a_n\equiv a$ that
\begin{equation*}
\begin{split}
U(x) &= F(U_{\underline{a}_U}(x));\qquad F(t) = U(U_{\underline{a}_U}^{-1}(t)) = U\left(-\frac{1}{\underline{a}_U}\log\left(-\underline{a}_U t\right)\right);\\
U_{\bar{a}_U}(x) &= \hat{F}(U(x));\qquad \hat{F}(t) = U_{\bar{a}_U}(U^{-1}(t)) = -\frac{1}{\bar{a}_U}e^{-\bar{a}_U U^{-1}(t)},
\end{split}
\end{equation*}
and where $F,\hat{F}$ are concave and increasing. Thus, by Jensen's inequality, for any set of random variables $\mathcal{Z}$:
\begin{equation*}
\begin{split}
\hat{F}^{-1}\left(\sup_{Z\in\mathcal{Z}}\espalt{}{U_{\bar{a}_U}(Z)}\right)\leq \sup_{Z\in\mathcal{Z}}\espalt{}{U(Z)} &\leq F\left(\sup_{Z\in\mathcal{Z}}\espalt{}{U_{\underline{a}_U}(Z)}\right),
\end{split}
\end{equation*}
where $\hat{F}^{-1}(s) = U\left(-(1/\bar{a}_U)\log\left(-\bar{a}_U s\right)\right)$ is strictly increasing. Therefore,
\begin{equation*}
U\left(-\frac{1}{\bar{a}_U}\log\left(-\bar{a}_U u^n_{\bar{a}_U}(x-q\tilde{p}^n,q)\right)\right) \leq u^n_U(x-q\tilde{p}^n,q) \leq U\left(-\frac{1}{\underline{a}_U}\log\left(-\underline{a}_U u^n_{\underline{a}_U}(x-q\tilde{p}^n,q)\right)\right).
\end{equation*}
Since for any $a>0$, $u^n_a(x-\tilde{p}^nq,q) =
e^{-a(x-\tilde{p}^nq)}u^n_a(0,q)$ , we obtain from
\eqref{E:indiff_px} that
\begin{equation}\label{E:main_gen_u_ineq}
\begin{split}
&U\left(-\frac{1}{\bar{a}_U}\log(-\bar{a}_U u^n_{\bar{a}_u}(0)) + x-\tilde{p}^nq + qp^n_{\bar{a}_U}(q)\right) \leq u^n_U(x-\tilde{p}^n q,q)\\
&\qquad\qquad \leq
U\left(-\frac{1}{\underline{a}_U}\log(-\underline{a}_U
u^n_{\underline{a}_U}(0)) + x-\tilde{p}^nq +
qp^n_{\underline{a}_U}(q)\right).
\end{split}
\end{equation}
Now, let $\tilde{p}^n \in I^n = (\underline{B}_n,\bar{B}_n)$. As
$\lim_{q\uparrow\infty} p^n_{\underline{a}_U}(q) =
\underline{B}_n$, $\lim_{q\downarrow-\infty}
p^n_{\underline{a}_U}(q) = \bar{B}_n$ we have
\begin{equation*}
\lim_{|q|\uparrow\infty} q(p^n_{\underline{a}_U}(q) - \tilde{p}^n) = -\infty,
\end{equation*}
and hence from the second inequality in \eqref{E:main_gen_u_ineq} and  $\lim_{x\downarrow -\infty} U(x) = -\infty$ (which follows from \eqref{E:u_exp_decay}) we obtain
\begin{equation*}
\lim_{q\uparrow\infty} u^n_U(x-\tilde{p}^nq,q) = -\infty, \lim_{q\downarrow -\infty}u^n_U(x-\tilde{p}^nq,q) = -\infty.
\end{equation*}
As $U(x-\tilde{p}^nq - |q|\|B\|_{\Lb^{\infty}}) \leq
u^n_U(x-\tilde{p}^nq,q) \leq 0$, any maximizing sequence
$\cbra{q^n_m}_{m\in\mathbb{N}}$ must be bounded and has an
accumulation point $\hat{q}_n$. Now, $u^n_U(x-\tilde{p}^nq,q)$
admits the variational representation (see \cite{MR2489605})
\begin{equation}\label{E:main_u_n_def}
u_U^n(x-\tilde{p}^nq,q) =
\inf_{\qprob^n\in\tM^n,y>0}\left(y(x-\tilde{p}^nq) +
yq\espalt{\qprob^n}{B} +
\espalt{}{V\left(y\frac{d\qprob^n}{d\prob}\bigg|_{\F_T}\right)}\right),
\end{equation}
where
\begin{equation}\label{E:dual_funct}
V(y)\dfn \sup_{x\in\reals}\left(U(x)-xy\right).
\end{equation}
Thus, we see that $q\mapsto u^n_U(x-\tilde{p}^nq,q)$ is concave, hence continuous on $\reals$ and $\hat{q}_n$ is indeed a maximizer.

We next show for $p^\infty(\delta^+) <
\tilde{p} < d$ and $I^n\ni\tilde{p}^n\rightarrow \tilde{p}$ that
\eqref{E:opt_pos_pos_limits} holds (the corresponding proof for
negative positions in \eqref{E:opt_pos_neg_limits} is omitted as
it is the exact same).  We first claim that for $n$ large enough,
any maximizer $\hat{q}_n$ is positive.  Indeed, since
$d_n\rightarrow d$ where  $d_n = \espalt{\qprob^n_0}{B} = p^n_a(0)
$ (for any $a>0$) and $\tilde{p} < d$, $\tilde{p}^n\rightarrow
\tilde{p}$ we can find $n$ large enough so that $\tilde{p}^n <
d_n$.  Thus, for $q<0$ we have (since $p^n_a(q)$ is decreasing in
$q$ for any $a>0$) that
\begin{equation*}
q p^n_{\underline{a}_U}(q) - q\tilde{p}^n \leq q\left(d_n - \tilde{p}^n\right) \leq 0.
\end{equation*}
In view of \eqref{E:main_gen_u_ineq} this implies for $q\leq 0$ that
\begin{equation}\label{E:T1}
u^n_U(x-\tilde{p}^nq,q) \leq U\left(-\frac{1}{\underline{a}_U}\log\left(-\underline{a}_U u^n_{\underline{a}_U}(0)\right) + x\right).
\end{equation}
Now, let $\ell > 0$ be so that $\ell\bar{a}_U/a < \delta^+$. At $q = \ell r_n$ we have
\begin{equation*}
p^n_{\bar{a}_U}(\ell r_n)  - \tilde{p}^n = p^n_a(\bar{a}_U \ell/a r_n) - \tilde{p}^n \rightarrow p^\infty(\bar{a}_U\ell /a ) - \tilde{p}.
\end{equation*}
Since $\tilde{p} < p^\infty(0)$ and $p^\infty$ is continuous at
$0$ we can find an $\ell$ small enough so the above quantity is
strictly positive for $n$ large.  Thus, from
\eqref{E:main_gen_u_ineq} we see that
\begin{equation*}
u^n_U(x-\tilde{p}^n\ell r_n,\ell r_n) \geq U\left(-\frac{1}{\bar{a}_U}\log(-\bar{a}_U u^n_{\bar{a}_u}(0)) + x-\tilde{p}^n\ell r_n +\ell r_n p^n_{\bar{a}_U}(\ell r_n)\right).
\end{equation*}
As $n\uparrow\infty$ the right hand side above converges to $0$
whereas the right hand side of \eqref{E:T1}, in view of Assumption
\ref{A:asympt_no_arb} is bounded above by $U(C+x)< 0$ for some
constant $C$.  Thus, for large enough $n$, no maximizer can be
non-positive.

Now, let $\cbra{\hat{q}_n}_{n\in\mathbb{N}}$ be a sequence of
(positive) maximizers. We prove the lower bound in
\eqref{E:opt_pos_pos_limits} by contradiction; i.e. assume
$\liminf_{n\uparrow\infty} \hat{q}_n/r_n  = 0$ and take a sequence
(still labeled $n$) where $\hat{q}_n/r_n \rightarrow  0$.  Let $0
< \ell < \delta^+\bar{a}_U/a$ and assume $\hat{q}_n/r_n \leq
\ell$.  Since $\hat{q}_n$ was an optimizer, we obtain from
\eqref{E:main_gen_u_ineq} that
\begin{equation*}
-\frac{1}{\bar{a}_U}\log(-\bar{a}_U u^n_{\bar{a}_U}(0)) +
x-\tilde{p}^n\ell r_n + \ell r_n p^n_{\bar{a}_U}(\ell r_n) \leq
-\frac{1}{\underline{a}_U}\log(-\underline{a}_U
u^n_{\underline{a}_U}(0)) + x-\tilde{p}^n\hat{q}_n +
\hat{q}_np^n_{\underline{a}_U}(\hat{q}_n).
\end{equation*}
Since $\ell r_n>0$
\begin{equation*}
-\frac{1}{\ell r_n \bar{a}_U}\log(-\bar{a}_U u^n_{\bar{a}_U}(0)) +
\frac{x}{\ell r_n} -\tilde{p}^n + p^n_{\bar{a}_U}(\ell r_n) \leq
-\frac{1}{\ell r_n \underline{a}_U}\log(-\underline{a}_U
u^n_{\underline{a}_U}(0)) + \frac{x}{\ell r_n} +
\frac{\hat{q}_n}{\ell
r_n}\left(p^n_{\underline{a}_U}(\hat{q}_n)-\tilde{p}^n\right).
\end{equation*}
For any $a>0$, $-(1/a) \leq u^n_a(0) =
-(1/a)e^{-\relent{\qprob^n_0}{\prob}}$. Additionally, from
\eqref{E:total_price} it holds for any $a,b>0$ that $p^n_a(q) =
p^n_b(aq/b)$. Thus by Assumptions \ref{A:GE_Opt1} and
\ref{A:asympt_no_arb}
\begin{equation*}
p^\infty\left(\frac{\bar{a}_U\ell}{a}\right) - \tilde{p} \leq \liminf_{n\uparrow\infty} \frac{\hat{q}_n}{\ell r_n}\left(p^n_{\underline{a}_u}(\hat{q}_n)-\tilde{p}^n\right) =0,
\end{equation*}
where the last equality follows since $\hat{q}_n/r_n \rightarrow
0$, $\tilde{p}^n\rightarrow \tilde{p}$ and $|p^n_{\underline{a}_U}(q)|\leq
\|B\|_{\Lb^{\infty}}$.  Taking $\ell\downarrow 0$ gives $\tilde{p}\geq
p^\infty(0)$ a contradiction.  Therefore,
$\liminf_{n\uparrow\infty}\hat{q}_n/r_n >0$.

To obtain the upper bound in \eqref{E:opt_pos_pos_limits}, we first claim that
\begin{equation}\label{E:indiff_px_new_lb}
p^n_U(x,\hat{q}_n) \geq \tilde{p}^n.
\end{equation}
Assuming \eqref{E:indiff_px_new_lb} the upper bound in \eqref{E:opt_pos_pos_limits} readily follows: indeed, assume $\limsup_{n\uparrow\infty}\hat{q}_n/r_n = k\geq \delta^+$ and take a subsequence (still labeled $n$) so that $\hat{q}_n/r_n \rightarrow k$.  Let $0 < \ell < \delta^+$ so that $\hat{q}_n/r_n\geq \ell$ for $n$ large enough.  Since $p^n_U(x,q)$ is decreasing in $q$, \eqref{E:indiff_px_new_lb} implies $\tilde{p}^n\leq p^n_U(x,\ell r_n)$. Taking $n\uparrow\infty$ gives $\tilde{p}\leq p^\infty(\ell)$ and then taking $\ell\uparrow\delta^+$ gives $\tilde{p}\leq p^\infty(\delta^+)$.  But, this is a contradiction and hence \eqref{E:opt_pos_pos_limits} holds.

To prove \eqref{E:indiff_px_new_lb}, come back to \eqref{E:main_u_n_def}.  Write $Z^{\qprob,n} \dfn d\qprob^n_0/d\prob |_{\F_T}$.  From \eqref{E:main_u_n_def} it follows for any $y>0$ that
\begin{equation}\label{E:var_prob_alt_1}
\frac{u^n_U(x-\tilde{p}^nq,q) - u^n_U(x)}{y} +\tilde{p}^nq \leq
q\espalt{\qprob^n_0}{B} +
\frac{1}{y}\left(\espalt{}{V(yZ^{\qprob,n})} + xy -
u^n_U(x)\right).
\end{equation}
Consider the problem
\begin{equation}\label{E:indiff_px_var}
\inf_{y>0}\frac{1}{y}\left(\espalt{}{V(yZ^{\qprob,n})} +xy - u^n_U(x)\right).
\end{equation}
According to \cite[Lemma A.4]{Robertson_2012} the map $y\mapsto \espalt{}{V(yZ^{\qprob,n})}$ is differentiable with derivative $\espalt{}{Z^{\qprob,n}V'(yZ^{\qprob,n})}$.  Thus, we see the derivative of the above map is
\begin{equation*}
\frac{1}{y^2}\left(\espalt{}{yZ^{\qprob,n}V'(yZ^{\qprob,n}) - V(yZ^{\qprob,n})} +u^n_U(x)\right) = \frac{1}{y^2}\left(\espalt{}{\int_0^{yZ^{\qprob,n}}\tau V''(\tau)d\tau} + u^n_U(x)\right),
\end{equation*}
where the last equality follows since $(d/d\tau)(\tau
V'(\tau)-V(\tau)) = \tau V''(\tau)$ and since $U\in\Ua$ implies
$\lim_{\tau\downarrow 0} \tau V'(\tau) = \lim_{\tau\downarrow
0}V(\tau) = 0$.  Since $U\in\Ua$ and Assumption
\ref{A:asympt_no_arb} imply $u^n_U(x) < 0$, the strict convexity
of $V$ yields a unique $y^{\qprob,n}$ solving
\eqref{E:indiff_px_var} and this $y$ satisfies the first order
condition
\begin{equation*}
-u^n_U(x) = \espalt{}{\int_0^{y^{\qprob,n}Z^{\qprob,n}}\tau
V''(\tau)d\tau}.
\end{equation*}
A straightforward calculation shows $\tau V''(\tau) =
1/\alpha_U(I(\tau))$ where $I(\tau) = \left(U'\right)^{-1}(\tau)$.
Since $U\in\Ua$ implies $0 < \underline{a}_U < \alpha_U(x) <
\bar{a}_U$ on $\reals$ we see that $\espalt{}{Z^{\qprob,n}} =1$
gives
\begin{equation*}
\frac{1}{\bar{a}_U}y^{\qprob,n} \leq -u^n_U(x) \leq
\frac{1}{\underline{a}_U} y^{\qprob,n},
\end{equation*}
or equivalently, that $-\underline{a}_U u^n_U(x) \leq y^{\qprob,n} \leq -\bar{a}_U u^n_U(x)$. Using this $y^{\qprob,n}$ in \eqref{E:var_prob_alt_1} gives
\begin{equation*}
\begin{split}
\frac{u^n_U(x-\tilde{p}^nq,q) - u^n_U(x)}{y^{\qprob,n}} + \tilde{p}^nq &\leq q\espalt{\qprob^n}{B} + \frac{1}{y^{\qprob,n}}\left(\espalt{}{V(y^{\qprob,n}Z^{\qprob,n})} + xy - u^n_U(x)\right)\\
&= q\espalt{\qprob^n}{B} +
\inf_{y>0}\frac{1}{y}\left(\espalt{}{V(yZ^{\qprob,n})} + xy -
u^n_U(x)\right).
\end{split}
\end{equation*}
We have already shown the existence of a $\hat{q}_n > 0$ which maximizes $u^n_U(x-\tilde{p}^nq,q)$ and shown that for $n$ large enough $u^n(x-\tilde{p}\hat{q}_n,\hat{q}_n) > u^n_U(x)$. Thus, for this $\hat{q}_n$ we have, using the inequalities for $y^{\qprob,n}$ that
\begin{equation*}
-\frac{1}{\bar{a}_U
u^n_U(x)}\left(u^n_U(x-\tilde{p}^n\hat{q}_n,\hat{q}_n) -
u^n_U(x)\right) + \tilde{p}^n\hat{q}_n \leq
\hat{q}_n\espalt{\qprob^n}{B} +
\inf_{y>0}\frac{1}{y}\left(\espalt{}{V(yZ^{\qprob,n})} + xy -
u^n_U(x)\right),
\end{equation*}
or, since this inequality is valid for any $\qprob^n\in\tM^n$ that
\begin{equation*}
\begin{split}
&u^n_U(x-\tilde{p}^n\hat{q}_n,\hat{q}_n) - u^n_U(x) - \bar{a}_U u^n_U(x)\tilde{p}^n\hat{q}_n\\
&\qquad\qquad \leq -\bar{a}_u u^n_U(x)\left(\inf_{\qprob^n\in \tM^n}\left(\hat{q}_n\espalt{\qprob^n}{B} + \inf_{y>0}\frac{1}{y}\left(\espalt{}{V(yZ^{\qprob,n})} +xy - u^n_U(x)\right)\right)\right)\\
&\qquad\qquad =-\bar{a}_Uu^n_U(x) \hat{q}_np^n_U(x,\hat{q}_n),
\end{split}
\end{equation*}
where the last equality follows from \cite[Proposition 7.1]{MR2489605}. We thus obtain the bounds
\begin{equation}\label{E:gen_u_better_ub}
u^n_U(x) \leq u^n_U(x-\tilde{p}^n\hat{q}_n,\hat{q}_n) \leq u^n_U(x) - \bar{a}_U u^n_U(x)\hat{q}_n\left(p^n_U(x,\hat{q}_n) - \tilde{p}^n\right).
\end{equation}
which, since $u^n_U(x)<0,\hat{q}_n>0$ implies \eqref{E:indiff_px_new_lb}, finishing the result.
\end{proof}

\section{Proofs from Section \ref{SS:BS_trans}}\label{S:pf_trans}

We begin with a lemma\footnote{See the comment in \cite[Section 2.1]{MR1809526}.} showing how the indifference price scales with the initial position and risk aversion. This is an easy consequence of the fact that $\mathcal{A}_t$ is a cone: i.e.~for each $c > 0$, $(L,M)\in\mathcal{A}_t \Leftrightarrow (c L,c M)\in A_t$. Throughout, we assume that $x,y\in\reals$, $0\leq t\leq T$, $s>0$, $a>0$ and $\transpm\in (0,1)$ (resp. $\transpm_n\in (0,1)$).

\begin{lemma}\label{L:Indiff_Px_Scale}
For $p_a$ as in \eqref{E:trans_px} and $q>0$:
\begin{equation}\label{E:lambda_scale}
p_a(qx,qy,q;s,t,\transpm) = p_{qa}(x,y,1;s,t,\transpm).
\end{equation}
\end{lemma}

\begin{proof}[Proof of Lemma \ref{L:Indiff_Px_Scale}]
For $(L,M)\in \mathcal{A}_t$ and $X,Y$ as in \eqref{E:wealth_dynamics} note that
\begin{equation}\label{E:temp_11}
\begin{split}
&-a\left(X^{L,M,qx,t}_T + Y^{L,M,qy,t}_T- q(S_T-K)^+\right) = -qa\left(X^{L/q,M/q,x,t}_T + Y^{L/q,M/q,x,t}_T - (S_T-K)^+\right).
\end{split}
\end{equation}
As $\mathcal{A}_t$ is a cone:
\begin{equation*}
\inf_{(L,M)\in\mathcal{A}_t}\espaltm{}{s,t}{e^{-a\left(X^{L,M,qx,t}_T + Y^{L,M,qy,t}_T- q(S_T-K)^+\right)}} = \inf_{(L,M)\in\mathcal{A}_t}\espaltm{}{s,t}{e^{-qa\left(X^{L,M,x,t}_T + Y^{L,M,y,t}_T- (S_T-K)^+\right)}}.
\end{equation*}
By removing $(S_T-K)^+$ from the above calculations we obtain from \eqref{E:trans_vf_no_claim} and \eqref{E:trans_vf_claim}:
\begin{equation}\label{E:V_Vf_Scale}
u_a(qx,qy,q; s,t,\transpm) = qu_{qa}(x,y,1;s,t,\transpm);\qquad u_a(qx,qy;s,t,\transpm) = qu_{qa}(x,y;s,t,\transpm).
\end{equation}
It is clear for $x'\in\reals$ that $u_{qa}(x+x',y,1;s,t,\transpm) = e^{-qa x'}u_{qa}(x,y,1;s,t,\transpm)$. To make the notation cleaner set $p = p_a(qx,qy,q;s,t,\transpm)$ and $p' = p_{qa}(x,y,1;s,t,\transpm)$ so that \eqref{E:lambda_scale} becomes $p = p'$.  Using the above facts
\begin{align*}
u_{qa}(x,y;s,t,\transpm) &= \frac{1}{q}u_{a}(qx,qy;s,t,\transpm) = \frac{1}{q}u_a(qx+qp,qy,q;s,t,\transpm);\\
&=\frac{1}{q}u_a(qx+qp'+q(p-p'),qy,q;s,t,\transpm);\\
&=u_{qa}(x+p'+(p-p'),y,1;s,t,\transpm);\\
&=e^{-qa\left(p-p'\right)}u_{qa}(x+p',y,1;s,t,\transpm);\\
&=e^{-qa(p-p')}u_{qa}(x,y;s,t,\transpm).
\end{align*}
Thus, $p=p'$.
\end{proof}


As in \cite[pp. 374-375]{MR1809526}, for $\eps > 0$ define
\begin{equation}\label{E:v_vf_def}
\begin{split}
v^{\eps}(x,y,s,t;\transpm) &\dfn 1+\frac{1}{\eps}u_{1/\eps}(x,y,1;s,t,\transpm);\qquad v^{\eps,f}(x,y,s,t;\transpm) \dfn 1+\frac{1}{\eps}u_{1/\eps}(x,y;s,t,\transpm).
\end{split}
\end{equation}
Next, define
\begin{equation}\label{E:z_zf_def}
\begin{split}
z^{\eps}(x,y,s,t;\transpm) &\dfn x + sy + \eps\log\left(1-v^{\eps}(x,y,s,t;\transpm)\right);\\
&= x+sy + \eps\log\left(-\frac{1}{\eps}u_{1/\eps}(x,y,1;s,t,\transpm)\right),\\
z^{\eps,f}(x,y,s,t;\transpm) &\dfn x + sy + \eps\log\left(1-v^{\eps,f}(x,y,s,t;\transpm)\right);\\
&= x+sy + \eps\log\left(-\frac{1}{\eps}u_{1/\eps}(x,y;s,t,\transpm)\right).\\
\end{split}
\end{equation}
Note that by definition $x+py-z^{\eps}$ and $x+py-z^{\eps,f}$ are the respective certainty equivalents in the $\transpm$ transactions costs market with and without the claim. Furthermore:

\begin{lemma}\label{L:z_zf_prop}
$z^{\eps},z^{\eps,f}$ from \eqref{E:z_zf_def} are independent of $x$ and hence write $z^{\eps}(y,s,t;\transpm), z^{\eps,f}(y,s,t;\transpm)$. Furthermore:
\begin{equation}\label{E:estimates}
\begin{split}
\Psi(s,t;0)  - \frac{\eps\mu^2}{2\sigma^2}(T-t) \leq z^{\eps}(y,s,t;\transpm)\leq s(1+\transpm|y-1|);\\
-\frac{\eps\mu^2}{2\sigma^2}(T-t)\leq z^{\eps,f}(y,s,t;\transpm)\leq \transpm s|y|,
\end{split}
\end{equation}
where $\mu$ is the drift of $S$ as in \eqref{E:trans_GBM} and $\Psi(s,t;0)$ is the Black-Scholes price in the frictionless model. Next, for a fixed $(y,s,t)$ and $\eps$, both $z^{\eps},z^{\eps,f}$ are increasing in $\transpm$. Lastly, for a fixed $(y,s,t)$ and $\transpm$, both $z^{\eps}$ and $z^{\eps,f}$ are continuous and decreasing in $\eps$ on $(0,\infty)$.

\end{lemma}

\begin{proof}[Proof of Lemma \ref{L:z_zf_prop}]

That $z^{\eps},z^{\eps,f}$ are independent of $x$ and that \eqref{E:estimates} holds both follow from \cite[Proposition 2.1]{MR1809526}. Next, using the definition of $v^{\eps}$ in \eqref{E:v_vf_def} and \eqref{E:wealth_dynamics} we have
\begin{equation*}
\begin{split}
&z^{\eps}(y,s,t;\transpm) - sy\\
&\quad = \inf_{(L,M)\in\mathcal{A}_t}\eps\log\left(\espaltm{}{s,t}{e^{-\frac{1}{\eps}\left(-\int_t^T S_\tau(1+\transpm)dL_\tau + \int_t^T S_\tau(1+\transpm)dM_\tau + yS_T + S_T(L_T-M_T) - (S_T-K)^+\right)}}\right)\\
&\quad  = \inf_{(L,M)\in\mathcal{A}_t}\eps\log\left(\espaltm{}{s,t}{e^{-\frac{1}{\eps}\left(-\int_t^TS_\tau dL_\tau + \int_t^TS_\tau dM_\tau + yS_T + S_T(L_T-M_T) - (S_T-K)^+\right)}e^{\frac{\transpm}{\eps}\int_t^TS_\tau(dL_\tau+dM_\tau)}}\right).
\end{split}
\end{equation*}
It is thus evident that $z^{\eps}(y,s,t;\transpm)$ is increasing in $\transpm$. Since the same formula holds for $z^{\eps,f}$, just absent the $(S_T-K)^+$ term, $z^{\eps,f}(y,s,t;\transpm)$ is also increasing in $\transpm$. Also, that $z^{\eps}(y,s,t;\transpm), z^{\eps,f}(y,s,t;\transpm)$ are decreasing in $\eps$ follows from Holder's inequality. Lastly, note that the map
\begin{equation*}
\begin{split}
\gamma &\mapsto \inf_{(L,M)\in\mathcal{A}_t}\espaltm{}{s,t}{e^{-\gamma\left(-\int_t^TS_\tau(1+\transpm)dL_\tau + \int_t^TS_\tau(1+\transpm)dM_\tau + yS_T + S_T(L_T-M_T) - (S_T-K)^+\right)}},
\end{split}
\end{equation*}
is convex on $(0,\infty)$ (and again, also when the $(S_T-K)^+$ term is absent).  Indeed, take $0<\gamma_1 < \gamma_2$ and $0 < \lambda < 1$. Set $\gamma_\lambda = \lambda\gamma_1 + (1-\lambda)\gamma_2$ and let $(L_1,M_2), (L_2,M_2)\in\mathcal{A}_t$.  Since $z\mapsto e^{-z}$ is convex and
\begin{equation*}
(L,M) = \frac{\lambda\gamma_1}{\gamma_\lambda}(L_1,M_1) + \frac{(1-\lambda)\gamma_2}{\gamma_\lambda}(L_2,M_2)\in \mathcal{A}_t
\end{equation*}
the convexity follows by first minimizing over $(L_1,M_1)$ then
over $(L_2,M_2)$. Since convex functions are continuous on the
interior of their effective domain and since $z^\eps,z^{\eps,f}$
are finite by \eqref{E:estimates} we see that
$z^{\eps}(y,s,t;\transpm),z^{\eps,f}(y,s,t;\transpm)$ are
continuous in $\eps$ on $(0,\infty)$.

\end{proof}


\begin{proof}[Proof of Proposition \ref{prop:bs_conv_price}]

Using Lemma \ref{L:Indiff_Px_Scale} at $q= (\eps a)^{-1}$ gives
\begin{equation*}
p_a\left(\frac{x}{\eps a},\frac{y}{\eps a},\frac{1}{\eps a};s,t;\transpm\right) = p_{1/\eps}\left(x,y,1;s,t,\transpm\right),
\end{equation*}
so that
\begin{equation*}
\begin{split}
v^{\eps}\left(x + p_a\left(\frac{x}{\eps a},\frac{y}{\eps a},\frac{1}{\eps a};s,t;\transpm\right), y, p, t;\mu\right) 
&=v^{\eps,f}(x,y,s,t;\transpm).
\end{split}
\end{equation*}
Thus, using \eqref{E:v_vf_def}, \eqref{E:z_zf_def} one obtains,
since Lemma \ref{L:z_zf_prop} shows $z^{\eps},z^{\eps,f}$ are
independent of the capital $x$, that
\begin{equation*}
\begin{split}
 p_a\left(\frac{x}{\eps a},\frac{y}{\eps a},\frac{1}{\eps a};s,t,\transpm\right) &= z^{\eps}\left(x+ p_a\left(\frac{x}{\eps\gamma},\frac{y}{\eps\gamma},\frac{1}{\eps a};s,t;\transpm\right),y,s,t;\transpm\right) - z^{\eps,f}(x,y,s,t;\transpm)\\
&= z^{\eps}\left(y,s,t;\transpm\right) - z^{\eps,f}(y,s,t;\transpm).
\end{split}
\end{equation*}
Thus, $p_a$ is independent of $x$. The conclusions of the theorem
now readily follow: namely let $r_n = \transpm_n^{-2}$ and set
$q_n = \ell r_n$. Let $y_n\in\reals$. Take $\eps_n =
\transpm_n^2/(a\ell) = (q_n a)^{-1}$ so that $q_n = (\eps_n
a)^{-1}$ and $\transpm_n = \sqrt{\eps_n}\sqrt{a\ell}$.  We then
have
\begin{equation*}
\begin{split}
p_a(y_n,q_n;s,t;\transpm_n) &= p_a\left(\frac{y_n\transpm_n^2/\ell}{\eps_n a},\frac{1}{\eps_n a};s,t,\sqrt{\eps_n}\sqrt{a\ell}\right)\\
&=z^{\eps_n}\left(\frac{y_n\transpm^2_n}{\ell},s,t;\sqrt{\eps_n}\sqrt{a\ell}\right) - z^{\eps_n,f}\left(\frac{y_n\transpm^2_n}{\ell},s,t;\sqrt{\eps_n}\sqrt{a\ell}\right).
\end{split}
\end{equation*}
Now, by \cite[Theorem 3.1]{MR1809526} we have for any $y_0\in\reals$ that
\begin{equation}\label{E:z_zf_n_limits}
\lim_{n\uparrow\infty} z^{\eps_n}\left(y_0,s,t;\sqrt{\eps_n}\sqrt{a\ell}\right) = \Psi(s,t;\sqrt{a\ell});\qquad \lim_{n\uparrow\infty} z^{\eps_n,f}\left(y_0,s,t;\sqrt{\eps_n}\sqrt{a\ell}\right) = 0.
\end{equation}
Furthermore, as shown on \cite[pp. 389]{MR1809526}
\begin{equation*}
\left|z^{\eps_n}\left(\frac{y_n\transpm_n^2}{\ell},s,t;\sqrt{\eps_n}\sqrt{a\ell}\right) - z^{\eps_n}(0,s,t;\sqrt{\eps_n}\sqrt{a\ell})\right| \leq \transpm_ns\frac{\transpm^2_n|y_n|}{\ell},
\end{equation*}
with the same inequality also holding for $z^{\eps_n,f}$.  Thus, if $\lim_{n\uparrow\infty}\transpm^3_n|y_n| = 0$ we see that
\begin{equation*}
\lim_{n\uparrow\infty} p_a(y_n,q_n;s,t;\transpm_n) = \Psi(p,t;\sqrt{a\ell}),
\end{equation*}
which is the desired result.

\end{proof}


\begin{proof}[Proof of Theorem \ref{T:VanishingTransactionsCost} ]
The proof of convergence follows the weak viscosity limits of
\cite{BarlesPerthame}, see also Chapter VII of
\cite{MR1199811}.  Let us
define
\[
\Psi^{*}(s,t)=\limsup_{\rho\downarrow 0}\limsup_{b\downarrow 0} \sup\left\{\Psi(\hat{s},\hat{t};b): |s-\hat{s}|+|t-\hat{t}|<\rho\right\},
\]

and
\[
\Psi_{*}(s,t)=\liminf_{\rho\downarrow 0}\liminf_{b\downarrow 0} \inf\left\{\Psi(\hat{s},\hat{t};b): |s-\hat{s}|+|t-\hat{t}|<\rho\right\}.
\]

\textit{Step 1: $\Psi^{*}(s,t)$ is a viscosity subsolution to the linear Black-Scholes equation.}

Let $w(s,t)$ be a smooth test function and assume that $(s_{0},t_{0})\in (0,\infty)\times[0,T]$ is a strict local maximizer of the difference $\Psi^{*}(s,t)-w(s,t)$ on $[0,\infty)\times[0,T]$ such that $\Psi^{*}(s_{0},t_{0})=w(s_{0},t_{0})$. We may, and will do so,  assume that $w_{ss}(s_{0},t_{0})\neq 0$. We verify that $\Psi^{*}$ is a viscosity subsolution, by proving that if $t_{0}<T$, then
\[
 -w_{t}(s_{0},t_{0})-\frac{1}{2}s_0^{2}\sigma^{2}w_{ss}(s_{0},t_{0})\leq 0,
\]
whereas if $t_{0}=T$, then either the previous inequality holds or  $\Psi^{*}(s_{0},T)\leq (s_{0}-K)^{+}$.

Let us assume that either $t_{0}<T$ or that $t_{0}=T$ and $\Psi^{*}(s_{0},T)> (s_{0}-K)^{+}$. Consider a sequence $b_{n}\downarrow 0$ and local maximizers $(s_{n},t_{n})\in (0,\infty)\times[0,T)$ of the function
\[
 (s,t)\mapsto \Psi(s,t;b_n)-w(s,t),
\]
such that
\[
 (s_{n},t_{n})\rightarrow (s_{0},t_{0}), \Psi(s_{n},t_{n};b_n)\rightarrow \Psi^{*}(s_{0},t_{0}), \text{ and } \Psi(s_{n},t_{n};b_n)-w(s_{n},t_{n})\rightarrow 0.
\]

The existence of such a sequence and maximizers is shown in \cite{BarlesPerthame}. Notice that for $n$ large enough we have  $t_{n}<T$. Indeed, if $t_{0}<T$, then $t_{n}<T$ for large enough $n$ follows by the convergence $t_{n}\rightarrow t_{0}$. Let's now assume that $t_{0}=T$ and $\Psi^{*}(s_{0},T)>(s_{0}-K)^{+}$ and let $t_{n}=T$. We calculate
\[
\Psi^{*}(s_{0},t_{0})=\lim_{n\rightarrow\infty}\Psi(s_{n},T; b_{n})=(s_{0}-K)^{+}.
\]
But, since we have assumed that $\Psi^{*}(s_{0},T)>(s_{0}-K)^{+}$ we get a contradiction, which implies that $t_{n}<T$ for all $n$ large enough.

Let us set now $k_{n}=\Psi(s_{n},t_{n};b_n)-w(s_{n},t_{n})$ and define the operator
\[
\mathcal{G}_{b}[\Psi]=\frac{1}{2}\sigma^{2}s^{2}\Psi_{ss}(s,t)\left(1+S\left(b s^{2}\Psi_{ss}(s,t)\right)\right).
\]

By the fact that $\Psi(;b_n)$ is a continuous viscosity solution of \eqref{Eq:NonlinearBlackScholesEquation} and that the function $A\mapsto A(1+S(A))$ is increasing function, we get the following
\begin{align}
 0&\geq -w_{t}(s_{n},t_{n})-\mathcal{G}_{b_{n}}[w(s_{n},t_{n})+k_{n}].\nonumber
\end{align}

Taking now $n\rightarrow\infty$ and using the facts that $\ell_{n}\rightarrow 0$, $(s_{n},t_{n})\rightarrow (s_{0},t_{0})$, $k_{n}\rightarrow 0$ and $S(0)=0$, we get
\[
 -w_{t}(s_{0},t_{0})-\frac{1}{2}\sigma^{2}s_{0}^{2}w_{ss}(s_{0},t_{0})\leq 0,
\]
completing the proof of the viscosity subsolution property of $\Psi^{*}$.

\textit{Step 2: $\Psi_{*}(s,t)$ is a viscosity supersolution to the linear Black-Scholes equation.}

The proof if this step is almost identical to the proof of the previous step. Let $w(s,t)$ be a smooth test function and assume that $(s_{0},t_{0})\in (0,\infty)\times[0,T]$ is a strict global minimizer of the difference $\Psi_{*}(s,t)-w(s,t)$ on $[0,\infty)\times[0,T]$  such that $\Psi_{*}(s_{0},t_{0})=w(s_{0},t_{0})$. We may, and will do so, assume that $w_{ss}(s_{0},t_{0})\neq 0$. We verify that $\Psi_{*}$ is a viscosity supersolution, by proving that if $t_{0}<T$, then
\[
 -w_{t}(s_{0},t_{0})-\frac{1}{2}s^{2}\sigma^{2}w_{ss}(s_{0},t_{0})\geq 0.
\]

If $t_{0}=T$, then by construction we have the supersolution property $\Psi_{*}(s,T)\geq (s-K)^{+}$. We need to show the viscosity property.

 Consider a sequence $b_{n}\downarrow 0$ and local minimizers $(s_{n},t_{n})\in (0,\infty)\times[0,T)$ of the function
\[
 (s,t)\mapsto \Psi(s,t;b_n)-w(s,t),
\]
such that
\[
 (s_{n},t_{n})\rightarrow (s_{0},t_{0}), \Psi(s_{n},t_{n};b_n)\rightarrow \Psi_{*}(s_{0},t_{0}), \text{ and } \Psi(s_{n},t_{n};b_n)-w(s_{n},t_{n})\rightarrow 0.
\]
The existence of such a sequence and minimizers is shown in \cite{BarlesPerthame}. Notice that, as in the viscosity subsolution case, for $n$ large enough, we have that  $t_{n}<T$. 

By the fact that $\Psi(;b_n)$ is a viscosity solution of (\ref{Eq:NonlinearBlackScholesEquation}) and that the function $A\mapsto A(1+S(A))$ is increasing function, we get the following
\begin{align}
 0& \leq -w_{t}(s_{n},t_{n})-\mathcal{G}_{b_{n}}[w(s_{n},t_{n})+k_{n}].\nonumber
\end{align}

Taking now $n\rightarrow\infty$ and using the facts that $\ell_{n}\rightarrow 0$, $(s_{n},t_{n})\rightarrow (s_{0},t_{0})$, $k_{n}\rightarrow 0$ and $S(0)=0$, we get
\[
 -w_{t}(s_{0},t_{0})-\frac{1}{2}\sigma^{2}s_{0}^{2}w_{ss}(s_{0},t_{0})\geq 0,
\]
completing the proof of the viscosity supersolution property of $\Psi_{*}$.

\textit{Step 3: Putting the estimates together}

By construction we have that $\Psi_{*}\leq \Psi^{*}$. Then a comparison argument as in proof of Theorem 3.1 of \cite{MR1809526}, or equivalently see Section VII.8 of \cite{MR1199811}, gives the opposite inequality, i.e.,
$\Psi_{*}\geq \Psi^{*}$. Thus we have that $\Psi_{*}= \Psi^{*}$ and the function $\Psi^{0}=\Psi_{*}= \Psi^{*}$ is solution to the equation
\[
 \Psi_{t}+\frac{1}{2}\sigma^{2}s^{2}\Psi_{ss}= 0;\qquad \Psi(T,s) =
 (s-K)^+.
\]
Classical arguments, e.g. Theorem 7.1 of  \cite{MR1199811}, then imply that the equality $\Psi_{*}= \Psi^{*}$ implies the local uniform convergence $\Psi^{\ell}\rightarrow\Psi^{0}$ as $\ell\rightarrow 0$. This completes the proof of the theorem.
\end{proof}


\begin{proof}[Proof of Theorem \ref{T:Psi_a_large}]

From Lemma \ref{L:z_zf_prop} at $\transpm = b\sqrt{\eps}$ it follows that $z^{\eps}(y,s,t;b\sqrt{\eps})$ is increasing in $b$.  Since \cite[Theorem 3.1]{MR1809526} implies $\lim_{\eps\rightarrow 0} z^{\eps}(y,s,t;b\sqrt{\eps})= \Psi(s,t;b)$, it follows that $\Psi(s,t;b)$ is increasing in $b$. As for the asymptotics in \eqref{E:Psi_a_large} by construction $\Psi(s,T;b)= (s-K)^+$ for $p>0, b>0$.  Thus, we only consider when $t<T$. Here, we recall from Proposition \ref{prop:bs_conv_price} that $\lim_{A\uparrow\infty} S(A)/A = 1$.  Furthermore, as shown in \cite{MR1809526}, $S(A) > 0$ for $A>0$.  Thus, let $\gamma > 0$ and pick $A_\gamma$ so that $S(A)\geq (1-\gamma)A$ for $A\geq A_\gamma$.

Now, let $\psi: (0,\infty)\times [0,T]$ be a smooth function with $\psi_{ss}\geq 0$.  Write
\begin{equation*}
H[\psi] \dfn \psi_t + \frac{1}{2}\sigma^2s^2\psi_{ss}\left(1 + S(b^2s^2\psi_{ss})\right).
\end{equation*}
We have the following basic estimate, since $\psi_{ss}\geq 0$ and $A\mapsto A(1+S(A))$ is increasing:
\begin{equation*}
\begin{split}
H[\psi] & \geq \psi_t + 1_{s^2\psi_{ss}\geq \frac{A_\gamma}{b^2}}\left(\frac{1}{2}\sigma^2s^2\psi_{ss}\left(1+(1-\gamma)b^2s^2\psi_{ss}\right)\right)\\
&= \psi_t + 1_{s^2\psi_{ss}\geq \frac{A_\gamma}{b^2}}\left(\frac{1-\gamma}{2}\sigma^2\left(bs^2\psi_{ss}+\frac{1}{2(1-\gamma)b}\right)^2 - \frac{\sigma^2}{8b^2(1-\gamma)}\right)\\
&\geq \psi_t - \frac{\sigma^2}{8b^2(1-\gamma)} + \frac{1-\gamma}{2}\sigma^2\left(bs^2\psi_{ss}+\frac{1}{2(1-\gamma)b}\right)^2 - 1_{s^2\psi_{ss}<\frac{A_\gamma}{b^2}}\frac{1-\gamma}{2}\sigma^2\left(bs^2\psi_{ss}+\frac{1}{2(1-\gamma)b}\right)^2\\
&\geq \psi_t - \frac{\sigma^2}{8b^2(1-\gamma)} + \frac{1-\gamma}{2}\sigma^2\left(bs^2\psi_{ss}+\frac{1}{2(1-\gamma)b}\right)^2 - \frac{1-\gamma}{2}\sigma^2\left(\frac{A_\gamma}{b}+\frac{1}{2(1-\gamma)b}\right)^2\\
&=\psi_t - \frac{\sigma^2K_{\gamma}}{2b^2} + \frac{1-\gamma}{2}\sigma^2\left(bs^2\psi_{ss}+\frac{1}{2(1-\gamma)b}\right)^2\\
\end{split}
\end{equation*}
where
\begin{equation*}
K_\gamma \dfn \frac{1}{4(1-\gamma)} + (1-\gamma)\left(A_\gamma + \frac{1}{2(1-\gamma)}\right).
\end{equation*}
To recap, we have for $\psi$ smooth with $\psi_{ss}\geq 0$ that
\begin{equation}\label{E:H_lb}
H[\psi]\geq \psi_t - \frac{\sigma^2K_{\gamma}}{2b^2} + \frac{1-\gamma}{2}\sigma^2\left(bs^2\psi_{ss}+\frac{1}{2(1-\gamma)b}\right)^2.
\end{equation}

Now, let $C>0$ and denote by $\phi(s,t; C)$ the Black-Scholes price at $(s,t)$ for a call option with strike $K$, maturity $T$ when the interest rate is $0$ and the asset volatility is $C$. Let $M\in \reals$ and consider the function
\begin{equation*}
\psi(s,t) = \phi(s,t;C) - M(T-t).
\end{equation*}
Clearly, $\psi$ is smooth and from the explicit formula for $\phi(s,t;C)$ it follows that $\psi_{ss}\geq 0$.  We then have from \eqref{E:H_lb} (writing $\phi^C$ to denote the dependence upon $C$) that
\begin{equation*}
\begin{split}
H[\psi]&\geq \phi^C_t + M - \frac{\sigma^2 K_\gamma}{2a^2} + \frac{1}{2}(1-\gamma)\sigma^2\left(b s^2 \phi^C_{ss} + \frac{1}{2(1-\gamma)b}\right)^2;\\
&= -\frac{1}{2}C^2s^2\phi^C_{ss} + M - \frac{\sigma^2 K_\gamma}{2b^2} + \frac{1}{2}(1-\gamma)\sigma^2\left(b s^2 \phi^C_{ss} + \frac{1}{2(1-\gamma)b}\right)^2.
\end{split}
\end{equation*}
The quadratic form $(1/2)(1-\gamma)\sigma^2b^2x^2 + (1/2)(\sigma^2-C^2)x$ is bounded below by
\begin{equation*}
-\frac{1}{8}\frac{(\sigma^2-C^2)^2}{(1-\gamma)\sigma^2b^2}.
\end{equation*}
Plugging this into the above (with $s^2\phi^C_{ss}$ playing the role of $x$) yields
\begin{equation*}
\begin{split}
H[\psi]&\geq - \frac{(\sigma^2-C^2)^2}{8(1-\gamma)\sigma^2b^2} + M - \frac{\sigma^2 K_\gamma}{2b^2} + \frac{\sigma^2}{8(1-\gamma)b^2}.
\end{split}
\end{equation*}
Clearly, setting
\begin{equation}\label{E:MC_Val}
\begin{split}
M &=  \frac{(\sigma^2-C^2)^2}{8(1-\gamma)\sigma^2b^2}+ \frac{\sigma^2 K_\gamma}{2b^2} - \frac{\sigma^2}{8(1-\gamma)b^2}= \frac{C^4}{8(1-\gamma)\sigma^2b^2} - \frac{C^2}{4(1-\gamma)b^2} + \frac{\sigma^2K_\gamma}{2b^2},
\end{split}
\end{equation}
yields that $H[\psi]\geq 0$ and hence by the comparison argument shown in \cite[Theorem 3.1, pp. 395-396]{MR1809526} it follows that $\Psi(s,t;b)\geq \psi(s,t)$. To connect with the results therein, set
\begin{equation*}
\begin{split}
z^*(s,t) = \psi(s,t)= \phi(s,t;C) - M(T-t);\qquad z_*(s,t) = \Psi(s,t;b),
\end{split}
\end{equation*}
and note that $z^*$ is a (classical) sub-solution; $z_*$ is a continuous viscosity super-solution; $\lim_{s\uparrow\infty} z^*(s,t)/s = 1$, $\lim_{s\uparrow\infty} z_*(s,t)/s = 1$ uniformly in $0\leq t\leq T$; and that $z^*(0,t) = -M(a)(T-t) \leq z_*(0,t) = 0$ for any $t\leq T$ if $C>\sqrt{2\sigma}$.  Thus, the argument in \cite[pp. 395-396]{MR1809526} goes through.

Now, so far the choice of $C>0$ was arbitrary.  Consider then when $C = b^{1/4}$.  Here we have as $b\rightarrow \infty$ that
\begin{equation*}
\begin{split}
C & = C(b) \rightarrow \infty,\\
M &= M(b) = \frac{1}{8(1-\gamma)\sigma^2b} - \frac{1}{4(1-\gamma)b^{3/2}} + \frac{\sigma^2K_\gamma}{2b^2} \rightarrow 0.
\end{split}
\end{equation*}
Thus, we have from the comparison principle that
\begin{equation*}
\liminf_{b\uparrow\infty} \Psi(s,t;b) \geq \liminf_{b\uparrow\infty} \phi(s,t;C(b)) - M(b)(T-t) = s,
\end{equation*}
where the last equality follows from the well known fact that the price of a call in the Black-Scholes model converges to the initial stock price as the volatility approaches infinity.  This completes the proof since it was shown in \cite[Proposition 2.1, Theorem 3.1]{MR1809526} that $\Psi(s,t;b)\leq s$ for all $b>0$.
\end{proof}


\begin{proof}[Proof of Theorem \ref{T:opt_quant_trans}]

We verify that Proposition \ref{prop: one_p_pos} holds, yielding the desired result. As a first step towards this direction, we rewrite the involved optimization problem in a form that is easier to work with. For $\tilde{p}^n\in (\Psi(s,t;0), s)$ recall the optimal sale quantity problem in \eqref{E:trans_opt_q}:
\begin{equation*}
\max_{q> 0} u_a(x+ys(1-\transpm_n)+q\tilde{p}^n,0,q;s,t;\transpm_n).
\end{equation*}
With $\tilde{x} = x + ys(1-\transpm_n)$ we have, in view of \eqref{E:V_Vf_Scale} and \eqref{E:v_vf_def}, \eqref{E:z_zf_def}, that for $q>0$:
\begin{equation}\label{E:H_z_pos_rep}
\begin{split}
u_a(\tilde{x}+q\tilde{p}^n,0,q;s,t;\transpm_n)&= \frac{1}{a}\left(v^{\frac{1}{qa}}\left(\frac{\tilde{x}}{q}+\tilde{p}^n,0,s,t;\transpm_n\right)-1\right)=-\frac{1}{a}e^{-a\left(\tilde{x}+q\tilde{p}^n-qz^\frac{1}{qa}(0,s,t;\transpm_n)\right)},
\end{split}
\end{equation}
and hence it suffices to consider the optimization problem
\begin{equation}\label{E:px_z_pos_rep}
\sup_{q>0}\left(q\tilde{p}^n - qz^{\frac{1}{qa}}(0,s,t;\transpm_n)\right) = -\inf_{q>0}\left(q(-\tilde{p}^n) - q\left(-z^{\frac{1}{qa}}(0,s,t;\transpm_n)\right)\right).
\end{equation}
The existence of a maximizer $\hat{q}_n>0$, as well as the asymptotic behavior of $\hat{q}_n/r_n$ in \eqref{E:trans_rates} as $\transpm_n\rightarrow 0$ will follow from Proposition \ref{prop: one_p_pos} once the requisite hypotheses are shown to hold. Here, $p^n$ is the map
\begin{equation*}
q\mapsto p^n(q) = -z^{\frac{1}{qa}}(0,s,t;\transpm_n).
\end{equation*}
We first consider Assumption \ref{A:one_p_pos}.  As for bullet point one, note that by Lemma \ref{L:z_zf_prop}, $p^n$ is continuous and non-increasing on $(0,\infty)$. Regarding bullet point two, \eqref{E:estimates} gives
\begin{equation*}
\begin{split}
-qs(1+\transpm_n) \leq qp^n(q) & \leq -q\Psi(s,t;0) + \frac{\mu^2}{2a\sigma^2}(T-t),
\end{split}
\end{equation*}
so that for any $\gamma >0$
\begin{equation*}
\limsup_{n\uparrow\infty} \sup_{q\leq \gamma} q|p^n(q)| \leq \gamma\max\cbra{\Psi(s,t;0) + \frac{\mu^2}{2a\sigma^2}(T-t), \gamma s} \dfn C(\gamma) < \infty,
\end{equation*}
verifying bullet point two. Regarding bullet point three, from \eqref{E:z_zf_n_limits} where $\eps_n = \transpm_n^2/(a\ell)$, $q_n = \ell r_n$ and $r_n = \transpm_n^{-2}$ it holds for all $\ell > 0$ that $p^n(\ell r_n)\rightarrow -\Psi(s,t;\sqrt{a\ell}) = p^\infty(\ell)$. Thus, bullet point three holds with $\delta = \delta^+ = \infty$.  Lastly, regarding bullet point four, since Theorem \ref{T:Psi_a_large} shows that $\lim_{\ell\uparrow\infty} \Psi(s,t;\sqrt{a\ell}) = -\lim_{\ell\uparrow\infty} p^\infty(\ell) = -s$ and $s > \Psi(s,t;0) = -p^\infty_{+}(0)$, bullet point four holds (see the sufficient condition Assumption \ref{A:one_p_pos}).  Therefore, Assumption \ref{A:one_p_pos} holds.  Lastly, as stated above for $\tilde{p} \in (\Psi(s,t;0),s)$ we have
\begin{equation*}
-s = \lim_{\ell\uparrow\infty} p^\infty(\ell) < -\tilde{p} < p^\infty_{+}(0) = \lim_{\ell\downarrow 0} (-\Psi(s,t;\sqrt{a\ell})) = -\Psi(s,t;0).
\end{equation*}
Therefore, the results of Proposition \ref{prop: one_p_pos} go through, finishing the proof.

\end{proof}

\bibliographystyle{siam}
\def\polhk#1{\setbox0=\hbox{#1}{\ooalign{\hidewidth
  \lower1.5ex\hbox{`}\hidewidth\crcr\unhbox0}}}

\end{document}